\documentclass[journal]{IEEEtran}

\usepackage{amsmath, amssymb, amsthm}
\usepackage{algorithm, algpseudocode}
\usepackage{graphicx, float}

\DeclareMathOperator*{\soft}{soft}
\DeclareMathOperator*{\argmin}{arg\,min}
\DeclareMathOperator{\ba}{\mathbf{a}}
\DeclareMathOperator{\bv}{\mathbf{v}}
\DeclareMathOperator{\bw}{\mathbf{w}}
\DeclareMathOperator{\bs}{\mathbf{s}}
\DeclareMathOperator{\be}{\mathbf{e}}
\DeclareMathOperator{\bff}{\mathbf{f}}
\DeclareMathOperator{\bg}{\mathbf{g}}
\DeclareMathOperator{\bh}{\mathbf{h}}
\DeclareMathOperator{\bc}{\mathbf{c}}

\DeclareMathOperator{\bu}{\mathbf{u}}
\DeclareMathOperator{\bx}{\mathbf{x}}
\DeclareMathOperator{\bd}{\mathbf{d}}
\DeclareMathOperator{\bA}{\mathbf{A}}
\DeclareMathOperator{\bF}{\mathbf{F}}
\DeclareMathOperator{\bS}{\mathbf{S}}
\DeclareMathOperator{\bD}{\mathbf{D}}

\DeclareMathOperator{\bH}{\mathbf{H}}

\DeclareMathOperator{\blambda}{\boldsymbol{\lambda}}
\newcommand{\opt}[1]{#1^\text{\tiny \textbf{OPT}}}
\providecommand{\e}[1]{\ensuremath{\mathrm{e}{#1}}}

\newtheorem{theorem}{Theorem}
\newtheorem{corollary}{Corollary}

\usepackage{etoolbox}
\makeatletter
\patchcmd{\@makecaption}{\scshape}{}{}{}
\patchcmd{\@makecaption}{\\}{.\ }{}{}
\makeatother

\title{Enveloped Sinusoid Parseval Frames}

\author{\IEEEauthorblockN{
	Geoff Goehle\IEEEauthorrefmark{1}, 
	Benjamin Cowen\IEEEauthorrefmark{2}~\IEEEmembership{Member,~IEEE}, 
	J. Daniel Park\IEEEauthorrefmark{3}~\IEEEmembership{Member,~IEEE}, 
	and
	Daniel C. Brown\IEEEauthorrefmark{4}~\IEEEmembership{Senior Member,~IEEE}} \\ 
\IEEEauthorblockA{Applied Research Laboratory, \\
Pennsylvania State University,\\ State College, PA\\
 Email: {\tt \IEEEauthorrefmark{1}goehle@psu.edu, \IEEEauthorrefmark{2}bmc6220@psu.edu,
 \IEEEauthorrefmark{3}jdp971@psu.edu, \IEEEauthorrefmark{4}dcb19@psu.edu}}
 \thanks{This work was sponsored in part by the Department of the Navy, Office of Naval Research under ONR award numbers N00014-18-1-2820 and N00014-19-1-2221.}
 }

\begin{document}

\maketitle

\begin{abstract}
	This paper presents a method of constructing Parseval frames from any collection of complex envelopes. The resulting Enveloped Sinusoid Parseval (ESP) frames can represent a wide variety of signal types as specified by their physical morphology. Since the ESP frame retains its Parseval property even when generated from a variety of envelopes, it is compatible with large scale and iterative optimization algorithms. ESP frames are constructed by applying time-shifted enveloping functions to the discrete Fourier Transform basis, and in this way are similar to the short-time Fourier Transform.
	This work provides examples of ESP frame generation for both synthetic and experimentally measured signals. Furthermore, the frame's compatibility with distributed sparse optimization frameworks is demonstrated, and efficient implementation details are provided. Numerical experiments on acoustics data reveal that the flexibility of this method allows it to be simultaneously competitive with the STFT in time-frequency processing and also with Prony's Method for time-constant parameter estimation, surpassing the shortcomings of each individual technique.
\end{abstract}

\section{Introduction}
\label{sec:introduction}
\IEEEPARstart{T}{he} problem of decomposing a digital signal into oscillating components captures a broad spectrum of applications. Basic signal processing tools such as the Fourier Transform, the short-time Fourier Transform (STFT), wavelets, empirical mode decomposition, and others, generally seek to rewrite the signal as a linear combination of elementary signals or \textit{atoms}~\cite{marple, tqwt, emd}. Examination of the signal in the alternative representation can reveal important properties~\cite{gearbox-decomposition}, enable specialized denoising~\cite{speckle-reduction, synthVsAnalysis} and detection~\cite{spindle-detection} techniques, and make parameter estimation easier~\cite{sparse-DOA}.

However, physical signals often violate the assumptions underlying these techniques, making the resulting decompositions difficult to interpret.
Researchers have developed a wide variety of approaches for representing complicated and nuanced signals in terms of elementary components, ranging from common, generic textures~\cite{shearlet, curvelet} to tunable wavelets~\cite{tqwt}, to fully data-driven methods~\cite{dictlearn-bach, dictlearn-OG}. Sometimes a superset of multiple transforms is employed to capture temporally overlapping components of the signal~\cite{mca-old, mca-eeg}.

The modern extension of this approach is to fix an overcomplete basis (or frame) of elementary atoms, then employ convex optimization to infer a weighting of the atoms that best fits the data~\cite{ista}. Regularization plays an important role in characterizing the solution to overdetermined systems with a common choice being sparsity. Sparse optimization can produce coefficients that are mostly zero, so that only a small subset of available atoms are actually used to compose the signal's representation. Sparse representations are both efficient and interpretable~\cite{sparsity-old}.

A different type of approach is to approximate the signal with a physics-based model whose parameters are inferred. These parameters may indirectly define a superposition of atoms, but atom coefficients are not optimized directly. Prony's Method is a classical example that seeks to compute the poles of a filter whose impulse response best matches the given data~\cite{marple}, indirectly representing the data as a superposition of exponentially decaying sinusoids. This technique is still employed in practice \cite{prony-physics} because its model is explainable in terms of the signal morphology. However, the underlying assumption that the signal has a decaying exponential envelope and identification of the signal start time are critical to its performance.

This paper presents a procedure for generating a frame using any number of complex signal envelopes, called the Enveloped Sinusoid Parseval (ESP) frame.
The ESP frame shares similarities with the STFT but while the STFT is signal agnostic an ESP frame can be easily tuned to represent a wide variety of signals by supplying relevant envelopes.
It is compatible with modern convex analysis techniques, and because it retains the Parseval property, it is efficient and practical to deploy in large-scale and distributed iterative optimization algorithms.

Section \ref{sec:definition} presents a derivation of the ESP approach, and Section~\ref{sec:regularization} provides implementation details for incorporating it into regularized least squares problems. Examples of setting up an ESP frame on a synthetic signal and visualizing its coefficients are given in Section~\ref{sec:examples}. Section~\ref{sec:denoising} provides a more realistic ESP frame construction for harmonic oscillators, and demonstrates ESP as a nonlinear denoising filter in comparison with an STFT frame. Section~\ref{sec:params} demonstrates the ESP frame's utility in parameter estimation in comparison with Prony's Method.

\section{Enveloped Sinusoid Parseval Frames}
\label{sec:espframes}

Section \ref{sec:definition} defines the ESP frame for a general set of complex envelopes, and proves that they can always be normalized to produce a Parseval frame.
Optimization algorithms for \(L_1\)-regularized least-squares based coefficient inference are presented in Section~\ref{sec:regularization}. Finally, examples of ESP frame coefficients generated for simple signals are provided.

\subsection{Definition}
\label{sec:definition}

Consider the complex finite-dimensional Hilbert space \(\mathbb{C}^N\). All norms ($\|\cdot\|$) are computed in the $\ell_2$ sense unless otherwise stated. A {\em tight frame} is a collection of vectors \(\{\ba_i\}_{i=0}^{K-1}\) in \(\mathbb{C}^N\) and $\alpha > 0$ such that
\begin{equation}
\label{eq:tight-identity}
\|\bw\|^2 = \alpha\sum_{i=0}^{K-1} |\langle \bw, \ba_i\rangle|^2 \ \text{for all \(\bw\in \mathbb{C}^N\).}
\end{equation}
A {\em Parseval frame} is a tight frame with \(\alpha = 1\).  Given a tight frame we define the {\em analysis operator} to be the linear map from vectors \(\bw\in\mathbb{C}^N\) to frame coefficients \(\bA\bw \in \mathbb{C}^K\) given by
\begin{equation}
Aw[i] = \langle \bw, \ba_i \rangle.
\end{equation}
Elements of \(\bA\) are given by the matrix coefficients \(A[m,n] = \overline{a_m[n]}\).
The defining characteristic of tight frames is that the original vector can be reconstructed from the frame coefficients via the formula \cite[Prop. 3.11]{framesforundergraduates}
\begin{equation}
\bw = \frac{1}{\alpha}\bA^*\bA\bw = \frac{1}{\alpha}\sum_{i=0}^{K-1} \langle \bw, \ba_i \rangle \ba_i.
\end{equation}
The \textit{synthesis operator} \(\bA^*\) is the conjugate transpose of the analysis operator, and maps frame coefficients back to the signal space. In the case that \(\bA\) is a Parseval frame, \(\bA^*\) serves as the frame's left-inverse.

We define the class of {\em Enveloped Sinusoid Parseval Frames} by applying enveloping functions to the (non-unitary) Discrete Fourier Transform (DFT) basis. First, specify the envelopes as a collection of vectors \(\{\be_l\}_{l=0}^{L-1}\) which are not identically zero. Let \(\{\bs_k\}_{k=0}^{N-1}\) denote the DFT basis
\begin{equation*}
	s_k[n] = \exp(2\pi j kn/N)\ \text{for \(k,n = 0,\ldots, N-1\)}.
\end{equation*}
Let \(\bS:\mathbb{C}^N\to\mathbb{C}^N\) perform right-circular-shifting and \(\bD:\mathbb{C}^N\to M(\mathbb{C}^N)\) perform diagonalization, such that
\begin{align*}
	Sw[n] &= w[n-1 \bmod N], \\
	[D(v)w][n] &= v[n]w[n],
\end{align*}
for \(\bv,\bw \in\mathbb{C}^N\) and \(n=0,\ldots,N-1\). Note that \(\bD(\bv)\bw\) is operator notation for componentwise multiplication of \(\bv\) with \(\bw\).  We then define the ESP frame vectors \(\{\ba_{l,k,m}\}\) to be translates of the enveloped sinusoids
\begin{equation}
	\ba_{l,k,m} = \bS^m \bD(\be_l) \bs_k
\end{equation}
for \(l=0,\ldots, L-1\) and \(k,m=0,\ldots,N-1\). Expanding the terms reveals that
\begin{align*}
	a_{l,k,m}[n] &= e_l[n-m \bmod N] s_k[n-m] \\
	&= e_l[n-m \bmod N]\exp(2\pi j k(n-m)/N),
\end{align*}
and using the fact that \(n-m \bmod N\) ranges from 0 to \(N-1\),
\begin{align}
	\label{eq:norm}
	\|\ba_{l,k,m}\|^2 
	&= \sum_n |e_l[n-m \bmod N]|^2 = \|\be_l\|^2
\end{align}
 for all \(l, k, m\).

Theorem \ref{thm:esp} shows that the collection \(\{\ba_{l,k,m}\}\) is a tight frame. Notably, the conditions on the envelopes are minimal: even a set of unrelated envelopes will admit a frame under this procedure.
\begin{theorem}
\label{thm:esp}
Given a set of nonzero \(N\)-dimensional vectors \(\{\be_l\}_{l=0}^{L-1}\), the vectors \(\{\ba_{l,k,m}\}\) defined by
\[
a_{l,k,m}[n] = e_l[n-m \bmod N]\exp(2\pi j k(n-m)/N)
\]
for \(l=0,\ldots, L-1\) and \(k,m,n=0,\ldots,N-1\) form a tight frame with \(\alpha = N \sum_l \|\be_l\|^2\).
\end{theorem}

\begin{proof}
The tight frame condition is given in \eqref{eq:tight-identity}.  Using the fact that \(\bS\) is unitary, we have for \(\bw\in \mathbb{C}^N\)
\begin{align*}
	\sum_{k,l,m} |\langle \ba_{k,l,m}, \bw\rangle|^2 &= \sum_{k,l,m} |\langle \bS^{m} \bD(\be_l) \bs_k, \bw \rangle|^2 \\
	&= \sum_{l,m} \sum_k|\langle \bs_k, \bD(\overline{\be_l})\bS^{-m} \bw\rangle|^2.
\end{align*}
Because \(\{\bs_k\}\) is the (non-unitary) DFT basis, Plancherel's theorem implies
\begin{align*}
\sum_{k,l,m}& |\langle \ba_{k,l,m}, \bw\rangle|^2= N\sum_{l,m} \|\bD(\overline{\be_l})\bS^{-m} \bw\|^2 \\
=& N\sum_{l,m,n} | [D(\overline{e_l})S^{-m}w][n] |^2 \\
=& N\sum_{l,m,n} |\overline{e_l[n]} w[n+m \bmod N]|^2 \\
=& N\sum_{l,n} |e_l[n]|^2 \sum_m |w[n+m \bmod N]|^2.
\end{align*}
Since \(n+m\bmod N\) ranges from \(0\) to \(N-1\),
\begin{align*}
\sum_{k,l,m} |\langle \ba_{k,l,m}, \bw\rangle|^2=& N\sum_{l,n} |e_l[n]|^2 \|\bw\|^2 \\
&= \left(N \sum_l \|e_l\|^2 \right)\|\bw\|^2. \qedhere
\end{align*}
\end{proof}

The following corollary normalizes the frame to have the Parseval property and thus justifies the ESP naming convention.
\begin{corollary}
Given a set of nonzero \(N\)-dimensional vectors \(\{\be_l\}_{l=0}^{L-1}\) such that \(\|\be_l\| = (NL)^{-1/2}\) for all \(l\) then the vectors \(\{\ba_{k,l,m}\}\) form a Parseval frame.
\end{corollary}

\subsubsection*{Note on Efficient Implementation}
The analysis and synthesis transforms can be rewritten to
\begin{align}
	\label{eq:analysis}
	\bc_{k,l} = \bA_{k,l}\bw &= \bF^* \bD( \bS^k \bF \bH(\be_l))  \bF\bw, \\
	\label{eq:synth}
	\bw = \frac{1}{\alpha}\sum_{k,l} \bA^*_{k,l}\bc_{k,l}
	&= \frac{1}{\alpha} \bF^* \sum_{k,l} \bD(\bS^{-k} \bF \be_l) \bF \bc_{k,l},
\end{align}
where $\bF$ is the FFT and $\bH$ is the conjugate linear operator \(Hw[n] = \overline{w[N-n \bmod N]}\).
Using these equations, ESP frame analysis and synthesis can be computed efficiently using FFT diagonalization and graphics processing unit (GPU) parallelization. This is critical for the following iterative optimization algorithms where synthesis and analysis operations are performed in multitudes.

\subsection{Coefficient Inference}
\label{sec:regularization}

Frames, by their nature, do not uniquely represent vectors and for any given vector \(\bw\) there will be a linear subspace of frame coefficients \(\bc\) such that \(\frac{1}{\alpha}\bA^* \bc = \bw\).  The coefficients supplied by the analysis operator \(\bA\) are characterized by the fact that they minimize the \(L_2\)-norm \cite[Prop 6.8]{framesforundergraduates}:
\[
\bA\bw = \argmin_{\bc} \|\bc\| \ \text{such that \(\frac{1}{\alpha}\bA^*\bc = \bw\).}
\]
This produces the set of frame coefficients with minimum power, but in general this representation is not sparse. Sparse frame representations are important in the ESP setting because the canonical frame coefficients are computed independently, and thus for a general set of envelopes are expected to contain redundant information. Through application of \(L_1\)-regularization, the various frame vectors can be made to ``compete'' with each other and a sparse coefficient vector can be computed that maintains exact reconstruction of the input signal (or in the case of noisy data, maintains some allowable error).  In the ideal scenario, regularization can be used to identify superimposed components of a signal formed from a linear combination of frame vectors. This is most effective with highly distinct envelopes.

In formal terms, \(L_1\)-regularization entails finding frame coefficients that solve either the Basis Pursuit (BP) problem
\begin{equation}
\label{eq:bp}
\argmin_{\bc} \|\bD(\blambda) \bc\|_1\ \text{such that \(\frac{1}{\alpha}\bA^* \bc = \bw\)}
\end{equation}
or the Basis Pursuit Denoising (BPD) problem
\begin{equation}
\label{eq:bpd}
\argmin_{\bc} \|\bD(\blambda) \bc\|_1 + \frac{1}{2} \left\|\frac{1}{\alpha}\bA^* \bc - \bw\right\|^2,
\end{equation}
where \(\blambda>0\) is a weight vector that allows the user to control relative penalization between coefficients. A constant parameter \(\lambda\) is often used in place of a vectorized \(\blambda\).

These convex optimization problems can be solved by the Split Augmented Lagrangian Shrinkage Algorithm (SALSA), which is an instance of the Alternating Direction Method of Multipliers (ADMM) for which convergence is proved~\cite{salsa}.  The steps of SALSA~\cite[Algorithm 4]{salsa2} applied to the BP problem~\eqref{eq:bp} are written in Algorithm \ref{alg:bp}, where \(\text{soft}(\cdot)\) denotes the soft-thresholding function
\[
\text{soft}(x, T) = \begin{cases} \frac{|x|-T}{|x|} x & |x| > T \\ 0 & |x| \leq T. \end{cases}
\]
Note that for a Parseval frame \(\alpha = 1\) in Algorithm \ref{alg:bp}.

\begin{algorithm}
\caption{\(L_1\) Basis Pursuit Algorithm}
\label{alg:bp}
\begin{algorithmic}
\Procedure{Basis Pursuit}{$\bA, \bw, \alpha, \blambda$}
\State  Initialize \(\mu > 0\), \(\bx_0 = \bA\bw\), and \(\bd_0 = 0\)
\While{stopping criteria not satisfied}
\State \(\bu_n = \soft(\bx_{n-1} + \bd_{n-1}, \blambda/\mu) \)
\State \(\bv_n = \bu_n - \bd_{n-1}\)
\State \(\bx_n = \bv_n + \bA\left(\bw - \frac{1}{\alpha}\bA^*\bv_n\right)\)
\State \(\bd_n = \bx_n - \bv_n\)
\EndWhile
\EndProcedure
\end{algorithmic}
\end{algorithm}

In general, the large dimension of the frame space makes exact convergence to sparse coefficients computationally expensive, especially in the presence of multiple similar envelopes. Accordingly, this work also employs Reweighted Basis Pursuit, where \(\blambda\) is iteratively adapted to drive coefficients to zero in fewer iterations as described in \cite{candes}. At iteration $n$,
\begin{equation}
\label{eq:reweight}
\blambda_n[k,l,m] = \frac{1}{|u_n[k,l,m]| + \epsilon}
\end{equation}
where \(\epsilon\) is smaller than the smallest expected non-zero coefficient value.  This scheme drives small coefficients to zero faster by assigning them a heavier weight. It has been shown to produce exact sparse solutions when such solutions exist~\cite{candes}.

When the frame does not admit a sparse representation of the signal, as is the case with noisy data, BPD may be used to search for a sparse solution at the cost of reconstruction error. In this case, SALSA can be written as in Algorithm \ref{alg:bpd}, which is also guaranteed to converge. The only difference between Algorithms \ref{alg:bp} and \ref{alg:bpd} is the \((1 + \mu/\alpha)^{-1}\) coefficient in the computation of \(\bx_k\).
Both algorithms can be sped up by the use of \eqref{eq:analysis} and \eqref{eq:synth} for analysis and synthesis, as well as the addition of the predictor-corrector-type acceleration described in \cite[Algorithm 8]{fastadmm}.
\begin{algorithm}
\caption{\(L_1\) Basis Pursuit Denoising Algorithm}
\label{alg:bpd}
\begin{algorithmic}
\Procedure{Basis Pursuit Denoising}{$\bA, \bw, \alpha, \blambda$}
\State  Initialize \(\mu > 0\), \(\bx_0 = \bA\bw\), and \(\bd_0 = 0\)
\While{stopping criteria not satisfied}
\State \(\bu_n = \soft(\bx_{n-1} + \bd_{n-1}, \blambda/\mu) \)
\State \(\bv_n = \bu_{n} - \bd_{n-1}\)
\State \(\bx_n = \bv_n + \left(1+\mu/\alpha\right)^{-1} \bA\left(\bw - \frac{1}{\alpha}\bA^*\bv_n\right)\)
\State \(\bd_n = \bx_n - \bv_n\)
\EndWhile
\EndProcedure
\end{algorithmic}
\end{algorithm}

In the following BP examples, \(\lambda=1\) (because constant $\lambda$ does not impact the BP solution). For BPD examples, note that there exists \(\lambda_{\max}\) given by
\begin{equation}
\label{eq:lambdamax}
\lambda_{\max} = \|\bA \bw\|_\infty
\end{equation}
such that for all \(\lambda \geq \lambda_{\max}\) the BPD solution is zero~\cite[Section V.B]{compressive}.  The subsequent BPD experiments set \(\lambda\) as a percentage of \(\lambda_{\max}\) with \(\lambda = 0.1 \lambda_{\max}\) as a common choice.

At convergence, the choice of \(\mu\) does not impact the solution for either Algorithm \ref{alg:bp} or \ref{alg:bpd}, but it impacts the convergence rate. In the subsequent experiments \(\mu\) is set to \(\lambda / p\) where \(p\) is the 99th percentile of the initial coefficient magnitudes \(|c[k,l,m]|\).  This causes the first soft threshold of either the BP or BPD algorithm to zero-out 99\% of the coefficients.  In the case of vector-weighted BPD \(\mu =  \text{mean}(\blambda)/p\) is selected.

\subsection{Examples}
\label{sec:examples}
In this section we present an ESP frame composed of Gaussian envelopes and analyze exemplary synthetic signals constructed directly from the frame. The envelopes \({g_l:[0,T]\to \mathbb{R}}\) are defined by
\begin{equation}
\label{eq:gaussenv}
g_l(t) = \exp\left(-\frac{t^2}{2\sigma_l^2}\right),
\end{equation}
where \(\sigma_l > 0\) is the standard deviation. The ESP frame constructed with circularly shifted modulations of these envelopes is similar to a Morlet wavelet \cite{grossmann1984decomposition}. Let \(T = 5\)ms with a sampling frequency of \(f_s = 100\)kHz so that \(N = 500\). Define
\[
	\sigma_l = 10^{l/2 -4}\ \text{for }l = 0, \ldots, 4
\]
so there are \(L = 5\) envelopes and the standard deviations range from 0.1ms to 10ms. Using the \(g_l\) defined by these standard deviations, the ESP frame functions are given by
\begin{align*}
a_{l, k, m}&(t) = c_l g_l(t-t_m) \exp(2\pi j k f_s (t - t_m)/N) 
\end{align*}
where the \(c_l\) are chosen so that the vector \(\be_l\) satisfies \(\|\be_l\| = (NL)^{-1/2}\).  The corresponding vectors \(\ba_{l,k,m}\) form a Parseval frame by Theorem \ref{thm:esp}. With this configuration there are \(N^2L = 1.25\) million frame functions parameterized by
\begin{itemize}
	\item the standard deviation \(\sigma_l\) (0.1ms - 10ms),
	\item the frequency \(f_k\) (-50kHz - 50kHz),
	\item and the circular time shift \(t_m\) (0ms - 10ms).
\end{itemize}

The two signals to be used for analysis are given by
\begin{align*}
	\bff &= \sqrt{NL} \ba_{2, 325, 50} & \bg = \sqrt{NL} \ba_{0, 350, 100},
\end{align*}
so that $||\bff||=||\bg||=1$.
Note that \(\bff\) has a standard deviation of \(10\)ms, frequency of \(15\)kHz, and time shift of 0.5ms, and \(\bg\) has a standard deviation of \(1\)ms, frequency of \(20\)kHz, and time shift of \(1.0\)ms. The time series $\bff$ and $\bg$ can be visualized in the top left and right (respectively) of Figure~\ref{fig:noreg-synth} where we have wrapped the signal in time.
Note that $\bff$ and $\bg$ can be generated by applying the synthesis operator to corresponding sparse vectors of coefficients
\begin{align}
	\bff &= \bA^*\opt{\bc_f} \\
	\bg  &= \bA^*\opt{\bc_g},
\end{align}
where $\opt{}$ denotes that they are the sparsest possible set of coefficients that can produce \(\bff\) and \(\bg\), making them the optimal solutions corresponding to \eqref{eq:bp}.  In this case \(\opt{\bc_f}\) and \(\opt{\bc_g}\) are Kronecker delta vectors corresponding to the frame vectors used to define \(\bff\) and \(\bg\).

\subsubsection*{Frame Coefficient Visualization}
\label{sec:unregframe}

Equation \eqref{eq:analysis} can be used to efficiently compute the frame coefficients for $\bff$ and $\bg$. It is possible to visualize the magnitude of these coefficients in three dimensions, but the results are generally difficult to interpret. As an alternative we propose to use the Maximum Intensity Projections (MIP) of the frame coefficients along the frequency and time shift axes on a dB scale referenced against the maximum coefficient magnitude to produce pairs of plots as in Figure~\ref{fig:noreg-synth}.

\begin{figure}
\centering
\includegraphics[width=1.55in]{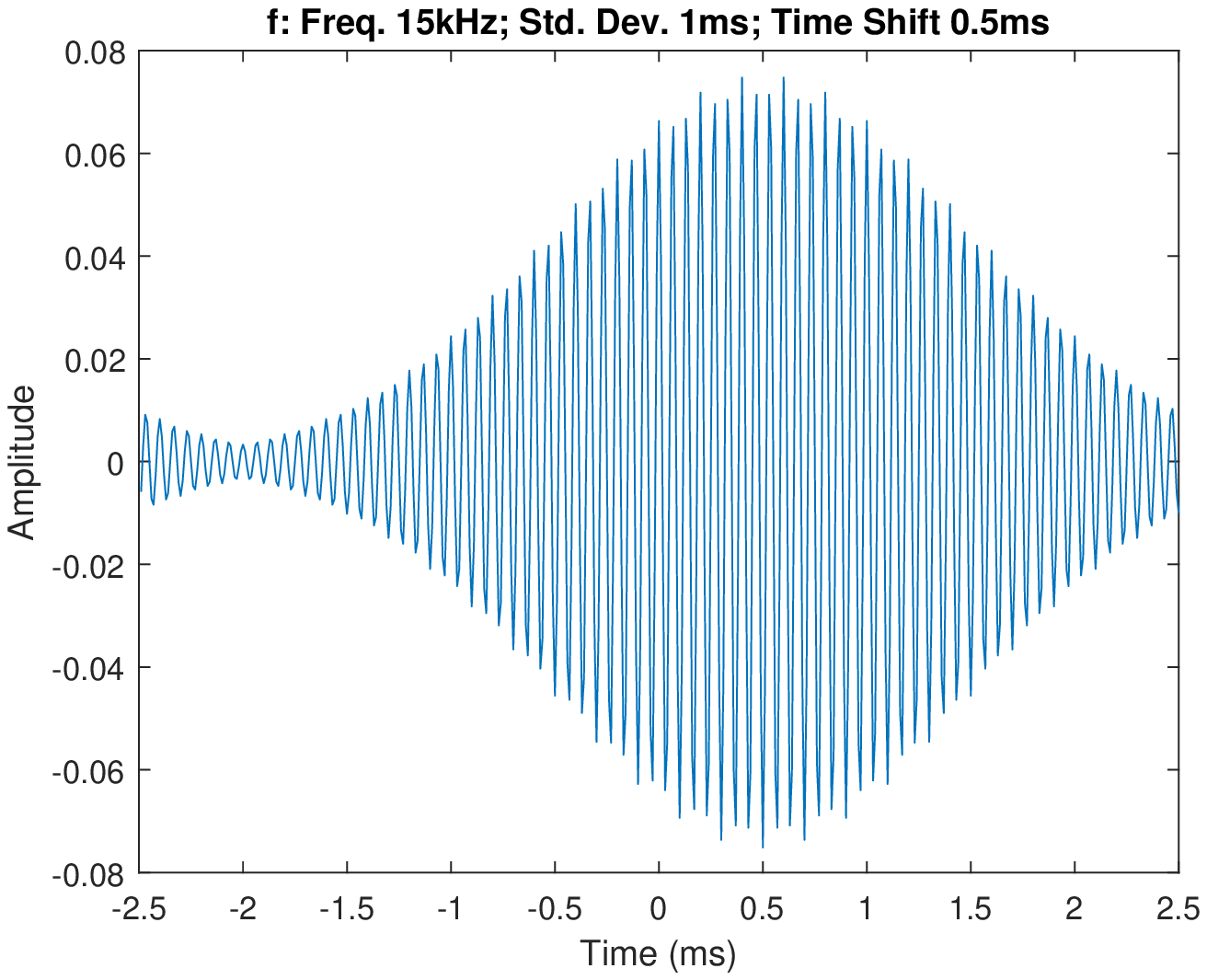}
\includegraphics[width=1.55in]{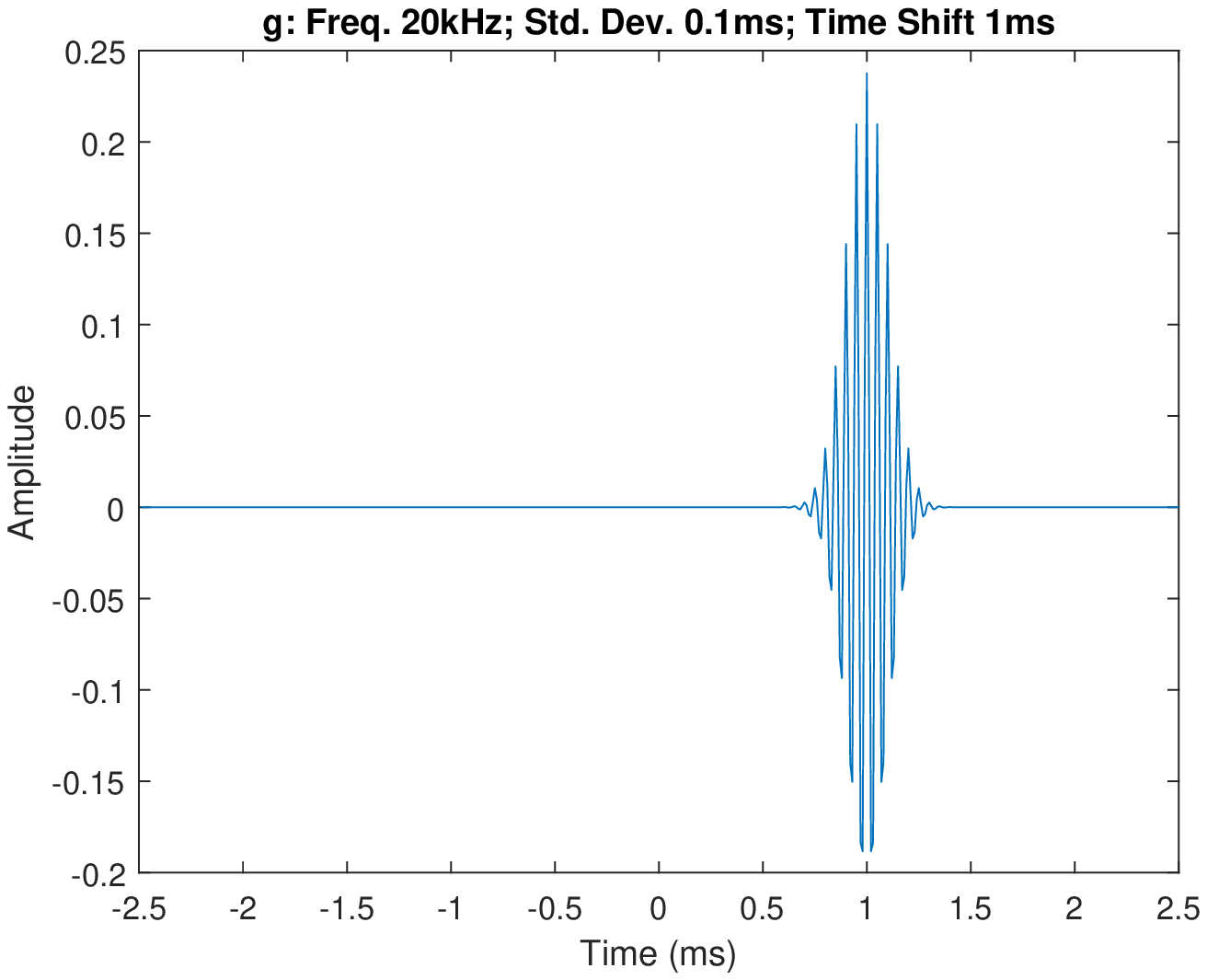}

\vspace{.1in}
\includegraphics[width=3.2in]{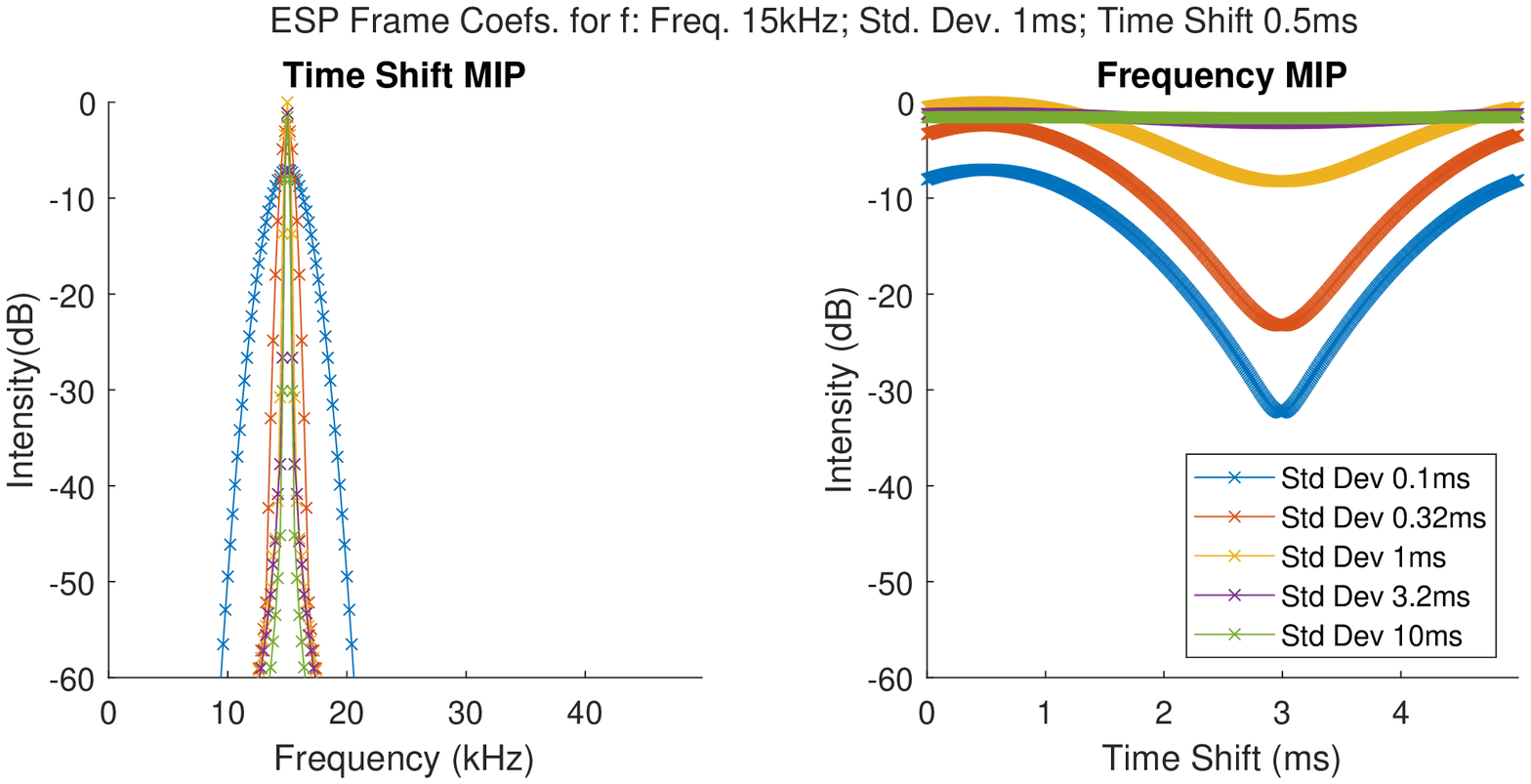}

\vspace{.1in}
\includegraphics[width=3.2in]{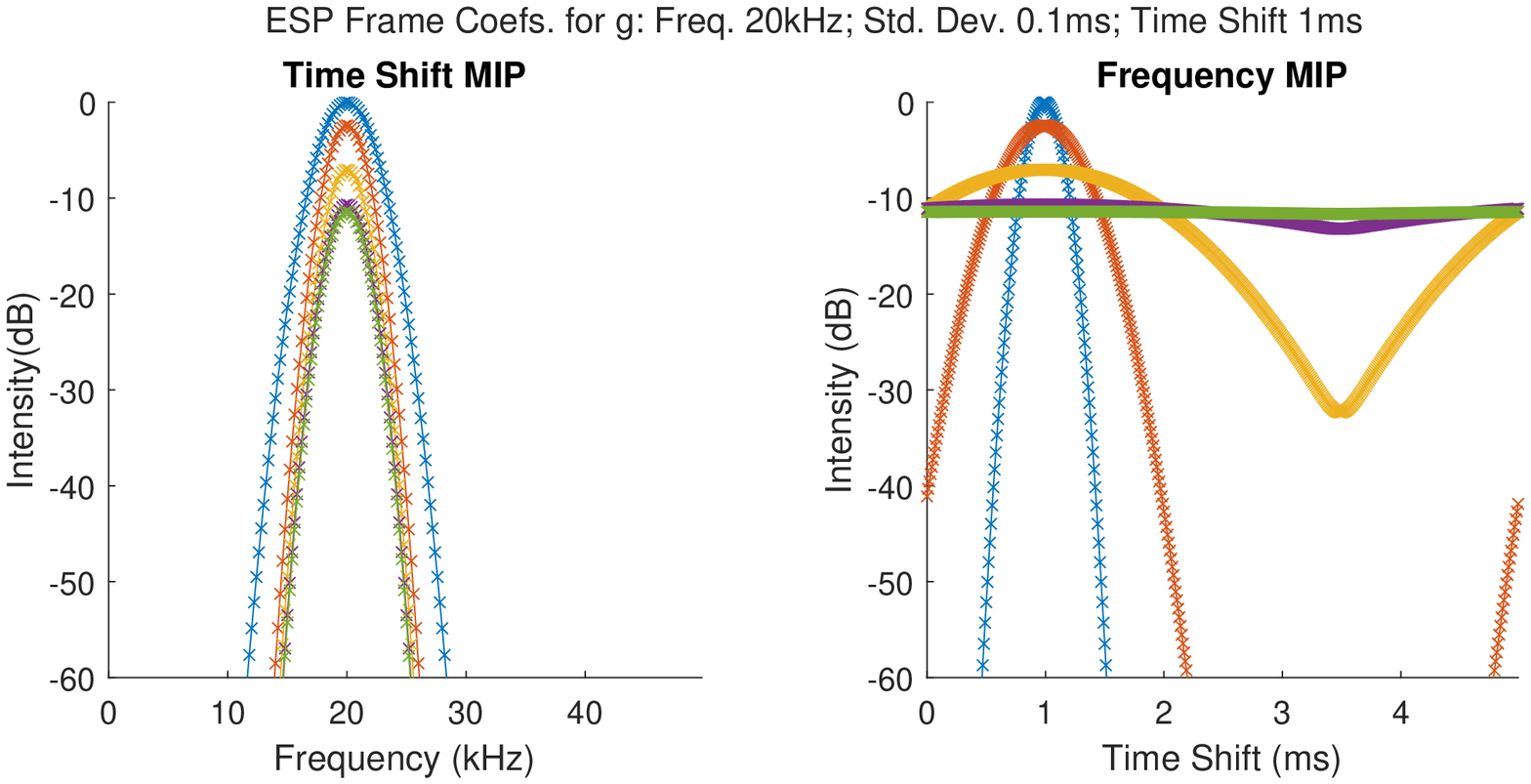}

\caption{Signals (top), and ESP frame coefficient time shift MIP (middle left and bottom left) and frequency MIP (middle right and bottom right) for \(\bff\) (top left and middle) and \(\bg\) (top right and bottom).  Intensities are shown on a dB scale relative to the maximum total frame coefficient.}
\label{fig:noreg-synth}
\end{figure}

Note that both sets of frame coefficients have a single maximum, correctly corresponding to the parameters of the underlying signal. (Note that the time shift is circular so the coefficient maximums wrap around the \(x\)-axis.) This is expected since, by the Cauchy-Schwartz inequality and \eqref{eq:norm},
\begin{align}
\label{eq:cs}
|c_{k,l,m}| &= |\langle \bv, \ba_{k,l,m}\rangle| \leq \|\bv\| \|\ba_{k,l,m}\| \\
 &= \|\bv\| \|\be_{l}\| = \frac{\|\bv\|}{\sqrt{NL}} \nonumber
\end{align}
with equality if and only if \(\bv\) is a scalar multiple of \(\ba_{k,l,m}\). In practice, the signal will not exactly match a frame vector, but we still expect the unregularized frame coefficients to have peaks near where a component of the signal is most approximately equal to a frame vector. Accordingly, these peaks can be used to identify the frequency, time shift, and envelope parameters for signal components. Notably, these are biased estimations because the underlying frame vectors are not orthonormal.

In this example the resolution for the longer signal \(f\) is better along the frequency axis than it is on the time shift axis, while the resolution for the shorter signal \(g\) is sharper on the time shift axis than it is on the frequency axis.  This tradeoff is due to the Fourier uncertainty principle. The resolution in the ``standard deviation axis'' is poor for both signals because (1) there are so few standard deviation parameters, and (2) there is less orthogonality between frame vectors of different standard deviations than there is between frame vectors with different frequencies or time shifts.

\subsubsection*{Sparse Frame Coefficients}
\label{sec:framereg}

As discussed above, a much sparser set of coefficients exists for \(\bff\) and \(\bg\), which are scaled Kronecker delta vectors $\opt{\bc_f}$ and $\opt{\bc_g}$. Algorithm \ref{alg:bp} was run using iterative reweighting with \(\lambda = 1\) and \(\epsilon = 50\) on \(\bff\) and \(\bg\) until the sequence \(\bu_n\) converged to the known solution. While this takes quite a few iterations, the convergence stabilizes relatively early on and Algorithm \ref{alg:bp} reasonably approximates the solutions by the 1000th iteration, as shown in Figure~\ref{fig:bp-synth}.
Note that the frequency, time shift, and standard deviation parameters that define \(\bff\) and \(\bg\) are all clearly represented with peaks around 90dB larger than the bulk of the frame coefficients.  This high level of performance is due to the signals perfectly matching the frame functions.

\begin{figure}
\centering
\includegraphics[width=3.2in]{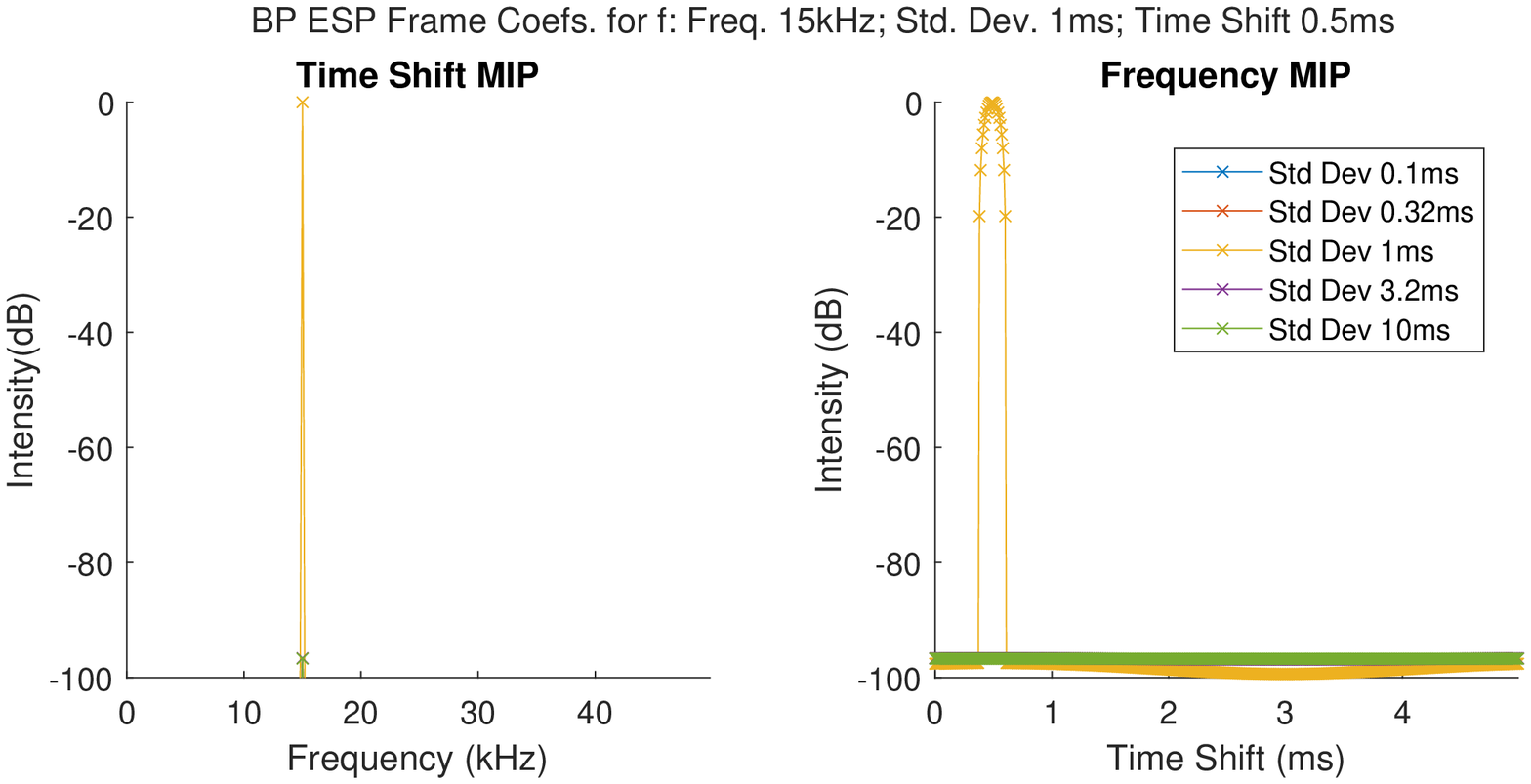}

\vspace{.1in}
\includegraphics[width=3.2in]{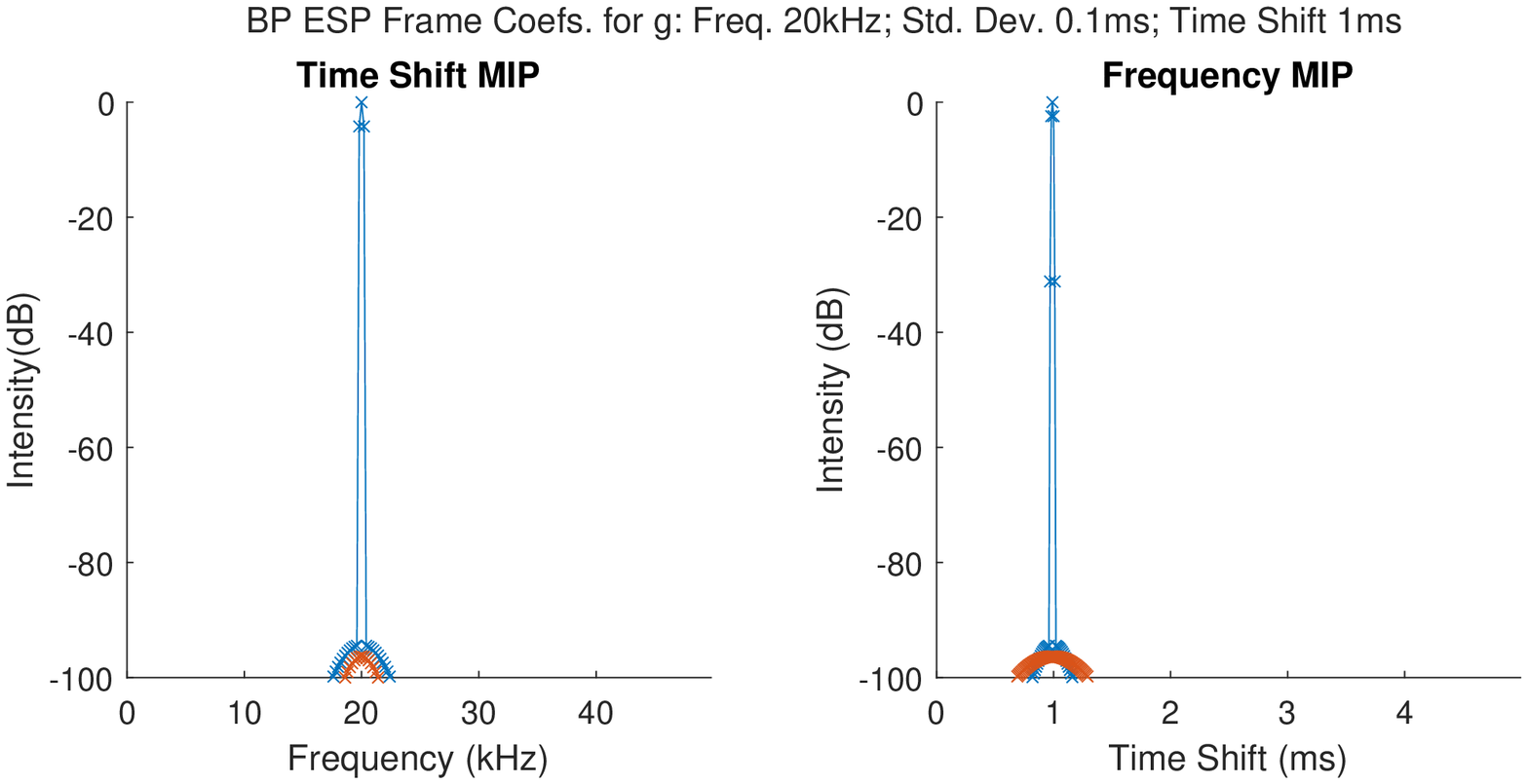}
\caption{ESP frame coefficient time shift MIP (left) and frequency MIP (right) for BP regularized coefficients for \(\bff\) (top) and \(\bg\) (bottom) after 1000 iterations with \(\lambda = 1\).  Intensities are shown on a dB scale relative to the maximum total frame coefficient.}
\label{fig:bp-synth}
\end{figure}

In this example we know {\em a priori} that BP regularization will converge to the desired coefficients and can confidently use regularized coefficients for parameter estimation.
In practice, the signal is unlikely to be a small linear combination of frame signals, either because of noise or because the envelope parameters do not exactly line up. In this case, since the exact solution is not represented by a particular envelope, parameter estimation via the maximum coefficient is inherently biased.

\section{Denoising}
\label{sec:denoising}

In this section the robustness of the ESP frame approach to noise is evaluated. BPD is used to filter both synthetic and experimentally collected noisy time series. Notionally, if the envelopes are chosen so that the ESP frame vectors are a good model for the signal, then the \(L_1\)-regularized representation of the signal will be sparse.  Then when BPD is applied the sparsification of the signal will preferentially remove noise and increase SNR. The ESP frame signal denoising performance is compared to the same of an STFT-based frame. The selected STFT is windowed with cosine functions so that it is also a Parseval frame as described in \cite{salsa2}. The STFT frame uses a window length of 128, resulting in 2,432 frame coefficients. For coefficient inference, Algorithm~\ref{alg:bpd} applies directly. Finally, the tradeoffs between sparsity and the reconstruction error for BPD regularized ESP frame and STFT frame coefficients are compared.

For the following sections we will start with either a synthetically or experimentally generated signal \(\bh\) and will form a noisy signal \(\bh_N\) by adding white Gaussian noise at some specified SNR (measured against the power of \(\bh\)).  We do not expect BPD to exactly reconstruct the original signal, even in the noise free case, and we track the reconstruction error using the relative error of the reconstructed signal \(\bh_R\) with the pure signal \(\bh\)
\[
E = \frac{\|\bh_R - \bh\|}{\|\bh\|}.
\]
The residual \(\bh_R  - \bh\) can be used to compute the reconstructed SNR via
\[
\text{SNR} = 20\log\left(\frac{\|\bh\|}{\|\bh_R-\bh\|}\right).
\]
We will look for this reconstructed SNR to produce a gain over the SNR of the added Gaussian noise as an indication that the BPD noise reduction process has been successful.

Section \ref{sec:denoising-synthetic} presents denoising applied to a synthetically generated time series using a specially constructed ESP frame.  Section~\ref{sec:denoising-experimental} presents a similar denoising analysis, but applied to experimentally collected time series.  Finally Section~\ref{sec:denoising-analysis} describes a comparative analysis between ESP frame denoising and STFT based denoising.

\subsection{Synthetic Data}
\label{sec:denoising-synthetic}
Consider an ESP frame engineered to detect resonance frequencies, such as those from the transfer function \({H:\mathbb{C}\to\mathbb{C}}\) defined by
\begin{align}
\label{eq:transfer}
H(z) &= \frac{z-1}{(z - \alpha)(z - \overline{\alpha})(z-\beta)(z-\overline{\beta})},  \\
\alpha &= -1/\tau_a + 2\pi j f_a,  \quad
\beta = -1/\tau_b + 2\pi j f_b, \nonumber
\end{align}
where \(f_a = 5\)kHz, \(\tau_a = 3\)ms, \(f_b = 13\)kHz, and \(\tau_b = 0.8\)ms.
The synthetic signal is generated by applying this transfer function to an impulse, i.e. the Kronecker delta vector \(\bv\) where \(v_{50} = 1\) and \(v_i = 0\) for \(i \ne 50\) (visualized in the top-left subplot of Figure~\ref{fig:h}).  The impulse is chosen so that the synthetic signal starts at \(0.5\)ms.  It can be shown via partial fractions \cite{oppenheim} that \(\bh\) is a combination of shifted exponentially decaying sinusoids. This suggests that an ESP frame constructed from exponential envelopes would be appropriate for analysis of this signal.

Define the ESP envelopes \(e_l:[0,T]\to \mathbb{R}\) with
\begin{equation}
\label{eq:expenv}
e_l(t) = \exp(-t/\tau_l)
\end{equation}
for \(\tau_l > 0\).  These envelopes are parameterized by the time constants \(\tau_l\) and we will use them to construct an ESP frame with vectors consisting of shifted exponentially decaying sinusoids.  Specifically, let \(T = 10\)ms and consider a sampling frequency of \(f_s = 100\)kHz so that \(N = 1000\). Let
\[
\tau_l = 10^{l/5 -4}\ \text{for } l = 0, \ldots, 8.
\]
Then \(L = 9\) and the time constants range from 0.1ms to 10ms. Using the envelopes \(e_l\) defined by these time constants we construct the ESP frame functions
\begin{align}
\label{eq:exp-frame}
a_{l, k, m}(t) & = c_l e_l(t-t_m) \exp(2\pi j k f_s (t - t_m)/N) 
\end{align}
where the \(c_l\) are chosen so that the vector \(\be_l\) satisfies \(\|\be_l\| = (NL)^{-1/2}\).  The vectors \(\ba_{l,k,m}\) form a Parseval frame by Theorem \ref{thm:esp}. With this configuration there are \(N^2L = 11\) million frame functions parameterized by
\begin{itemize}
\item the time constant \(\tau_l\) (0.1ms - 10ms),
\item the frequency \(f_k\) (-50kHz - 50kHz),
\item and the circular time shift \(t_m\) (0ms - 5ms).

\end{itemize}

Figure \ref{fig:h} displays the original signal, its STFT representation, and the unregularized ESP frame coefficients. Unlike in Section~\ref{sec:examples}, the time constants \(\tau_a\) and \(\tau_b\) of the signal resonances are not exactly represented in the frame envelopes, and so a trivial optimal solution is not known {\em a priori}. Furthermore, the unregularized coefficients are not sparse, and do not admit a clear peak along the time-constant axis.
The two main frequency peaks are visible in the time shift MIP and the time shift is visible in the frequency MIP. As expected, the frequency with the smaller time constant has less power than the frequency with the longer time constant. The time shift and both frequency peaks are visible in the STFT frame coefficients. However, since the time constant is not directly tracked as a parameter for the STFT frame, it would need to be estimated indirectly from the signal decay rate.

\begin{figure}
\centering
\includegraphics[width=1.58in]{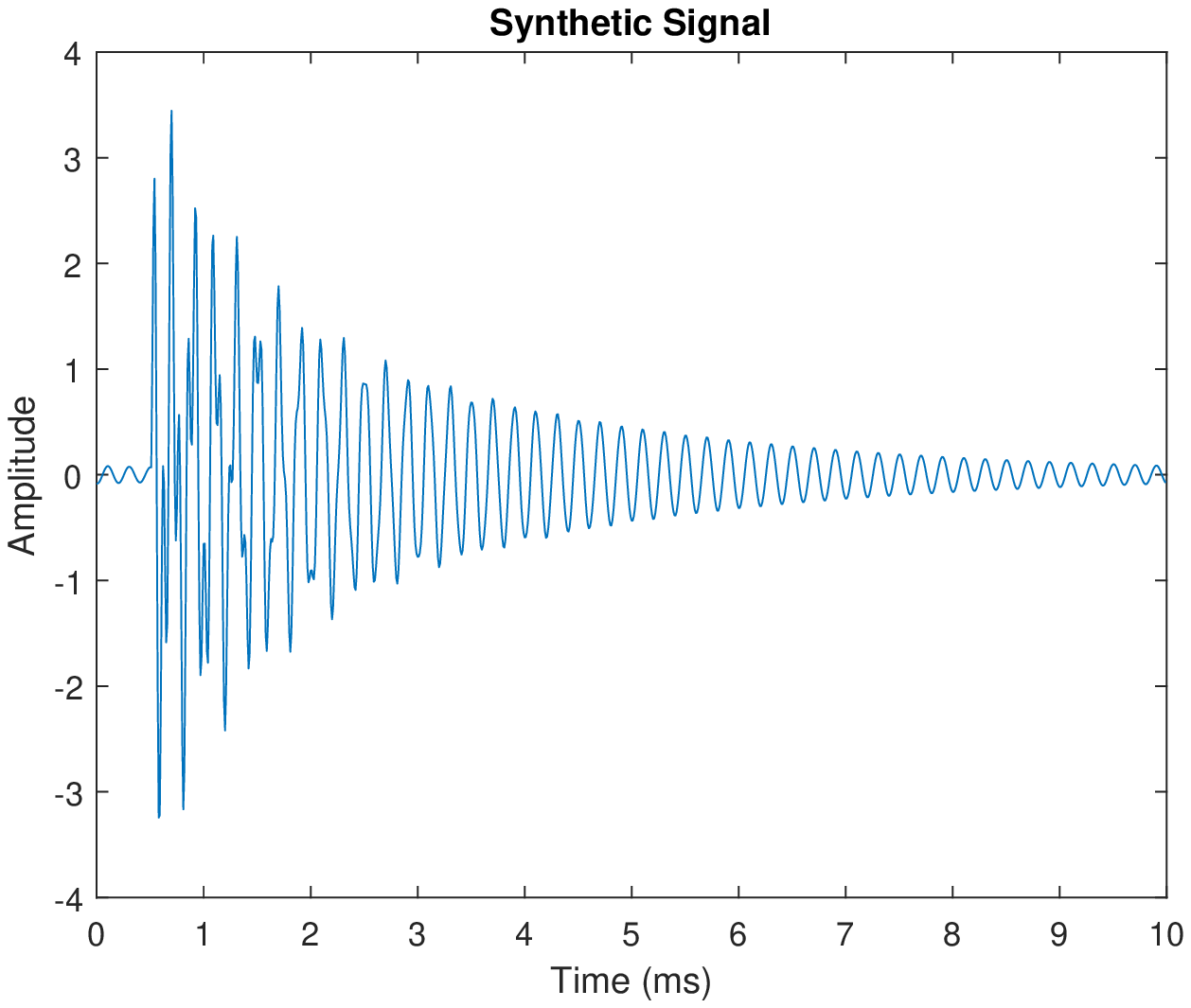}
\includegraphics[width=1.58in]{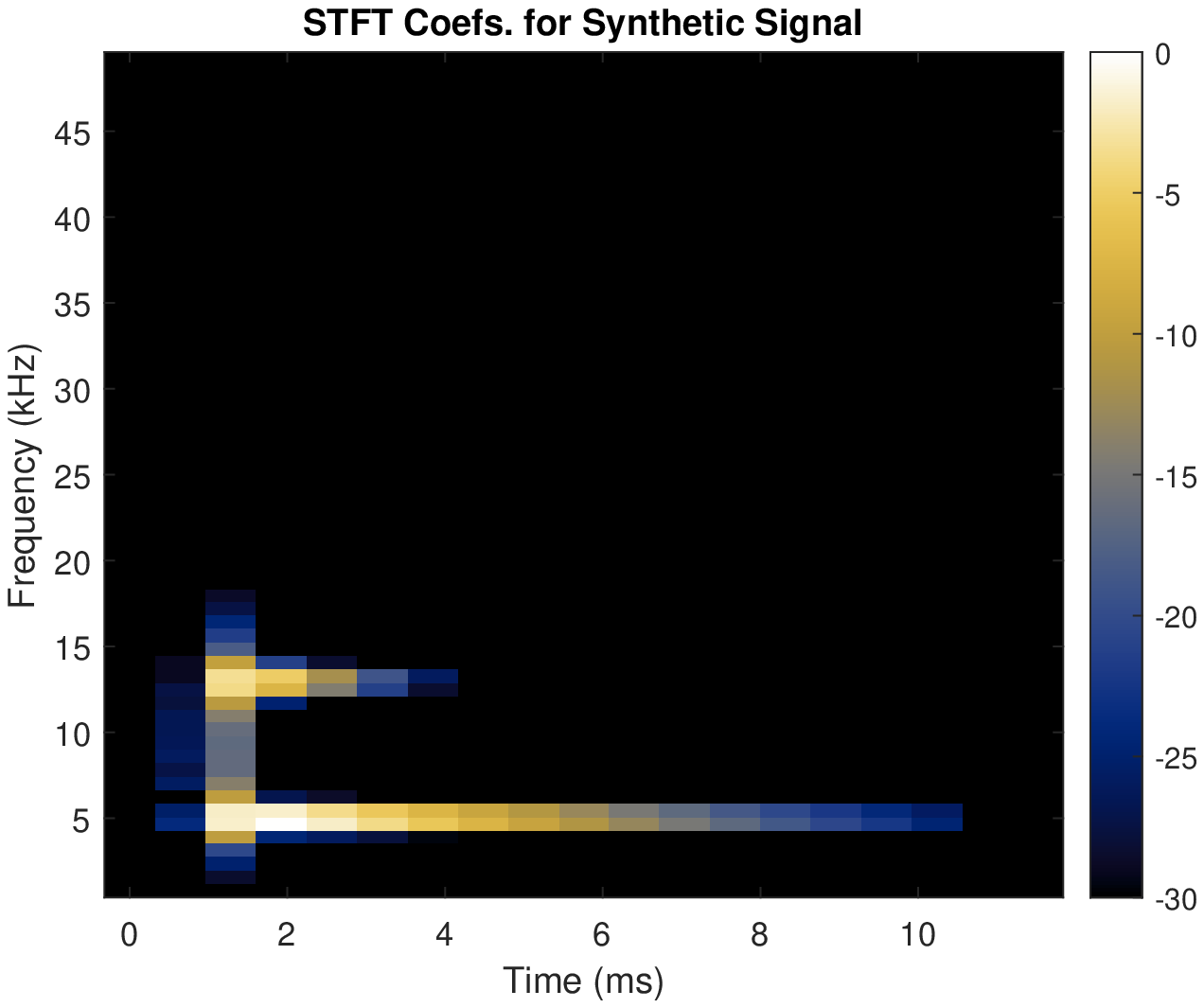}

\vspace{.1in}
\includegraphics[width=3.2in]{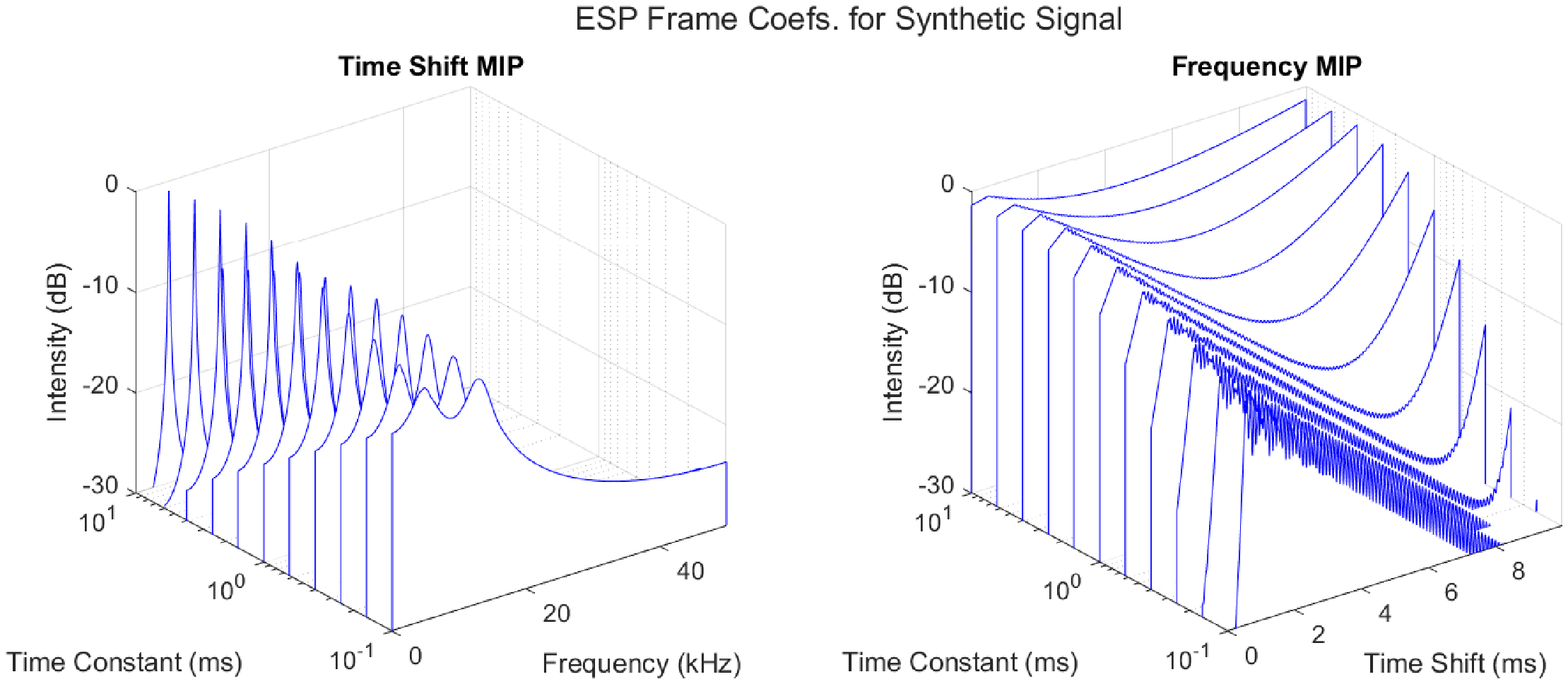}
\caption{Plots of \(\bh\) (top left), the STFT frame coefficients for \(\bh\) with a window length of 128 (top right), ESP frame coefficients time shift MIP separated by time constant (bottom left), and frequency MIP separated by time constant (bottom right).   All intensities are plotted in dB relative to the maximum coefficient amplitude. The estimated time constants for the 5kHz and 13kHz peaks are 2.52ms and 0.633ms, respectively.}
\label{fig:h}
\end{figure}

\subsubsection*{BPD on Synthetic Data}
\label{sec:denoising-synthetic-bpd}
Figure~\ref{fig:rh} shows the BPD coefficient MIPs and corresponding reconstruction of a noisy signal \(\bh_N\). The noisy signal was generated using an SNR of 10dB, and 1000 iterations of BPD were computed with \(\lambda = 0.1 \lambda_{\max}\). The computed sparse coefficient vector has 1596 nonzero ESP frame coefficients with a sparsity of 99.98\%, all clustered around the correct frequency, time shift, and time constant parameters.  The relative error between the regularized reconstructed signal and the pure signal \(\bh\) is 16.2\% with a reconstructed SNR of 15.8dB, a 5.8dB gain from the initial SNR.

The top subplot of Figure \ref{fig:rh} reveals that the reconstructed signal has less additive noise, particularly later on in the time series. The frame coefficients in Figure \ref{fig:rh} are less tidy than the regularized coefficients shown in Figure \ref{fig:bp-synth}.  This is due to the addition of Gaussian noise, which introduces signals that ``start'' at all times in the signal interval (in addition to the mismatch between signal and frame envelopes parameters).

\begin{figure}
\centering
\includegraphics[width=1.58in]{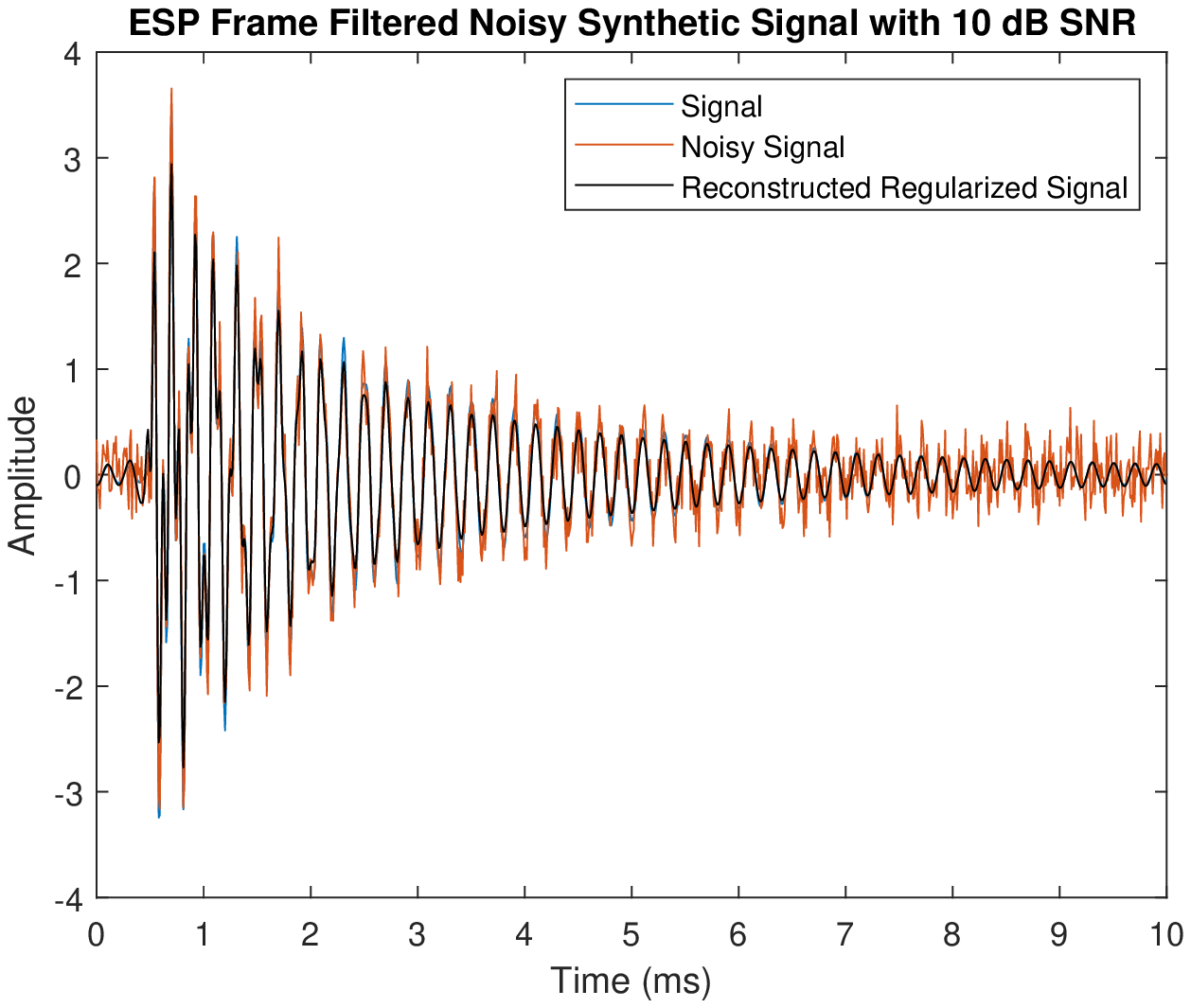}

\vspace{.1in}
\includegraphics[width=3.2in]{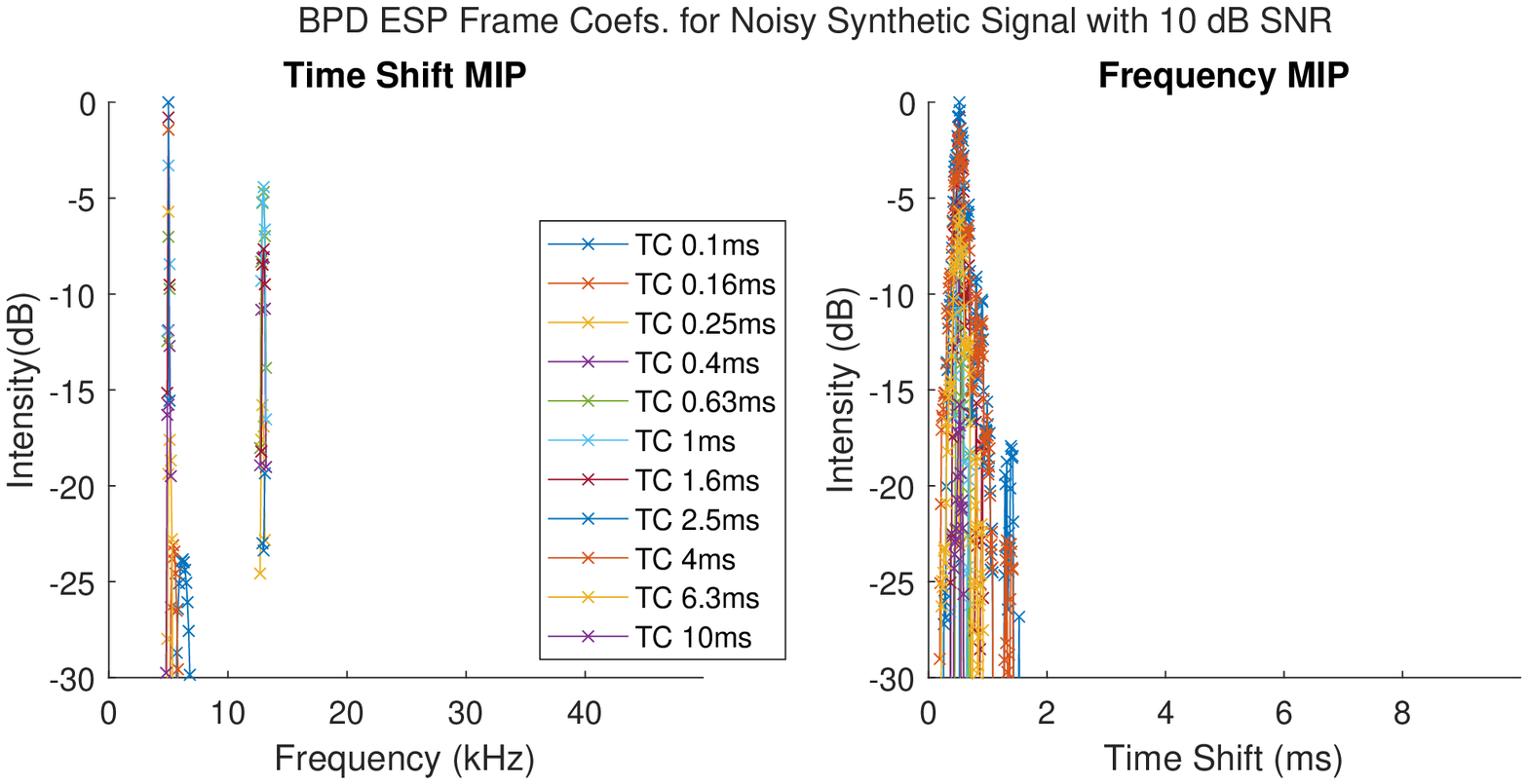}
\caption{Plots of \(\bh\), \(\bh_N\) with SNR 10dB, and the signal reconstructed from the BPD solution (top), ESP frame coefficient time shift MIP (bottom left), and frequency MIP (bottom right).   The regularization is shown at iteration 1000 with \(\lambda = 0.1\lambda_{\max}\). All intensities are plotted in dB relative to the maximum ESP frame coefficient amplitude.  The ESP Frame coefficient sparsity is 99.98\% and the relative error between the reconstructed signal and \(\bh\) is 16.2\% (15.8dB reconstructed SNR). The estimated time constants for the 5kHz and 13kHz peaks are 2.48ms and 0.952ms, respectively}
\label{fig:rh}
\end{figure}

Using Algorithm \ref{alg:bpd} with the STFT frame in place of the ESP frame on \(\bh_N\) produces the coefficients shown in Figure \ref{fig:rhstft}. The \(L_1\)-parameter \(\lambda = 0.1\lambda_{\max}\) was set using the $\lambda_{\max}$ computed by the STFT, and 1000 iterations were used. The solution has 58 nonzero coefficients with a sparsity of 97.61\%. The relative error between the regularized reconstructed signal and the pure signal is 24.9\% with a reconstructed SNR of 12.1dB. The increase in reconstructed SNR above the base signal SNR is 3dB smaller than the ESP frame case, and the reconstruction is visibly worse, particularly late in the time series.

\begin{figure}
\centering
\includegraphics[width=1.58in]{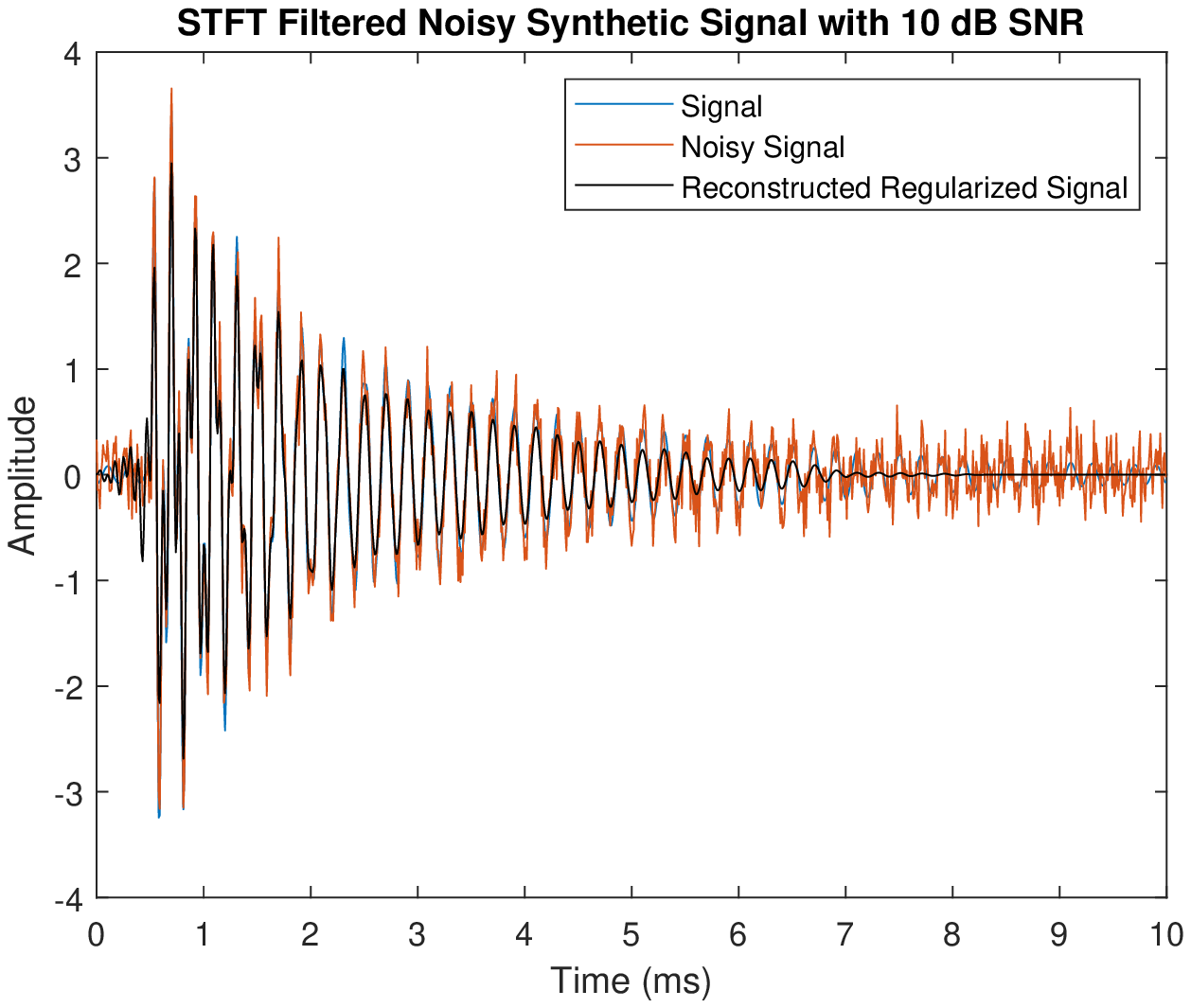}
\includegraphics[width=1.58in]{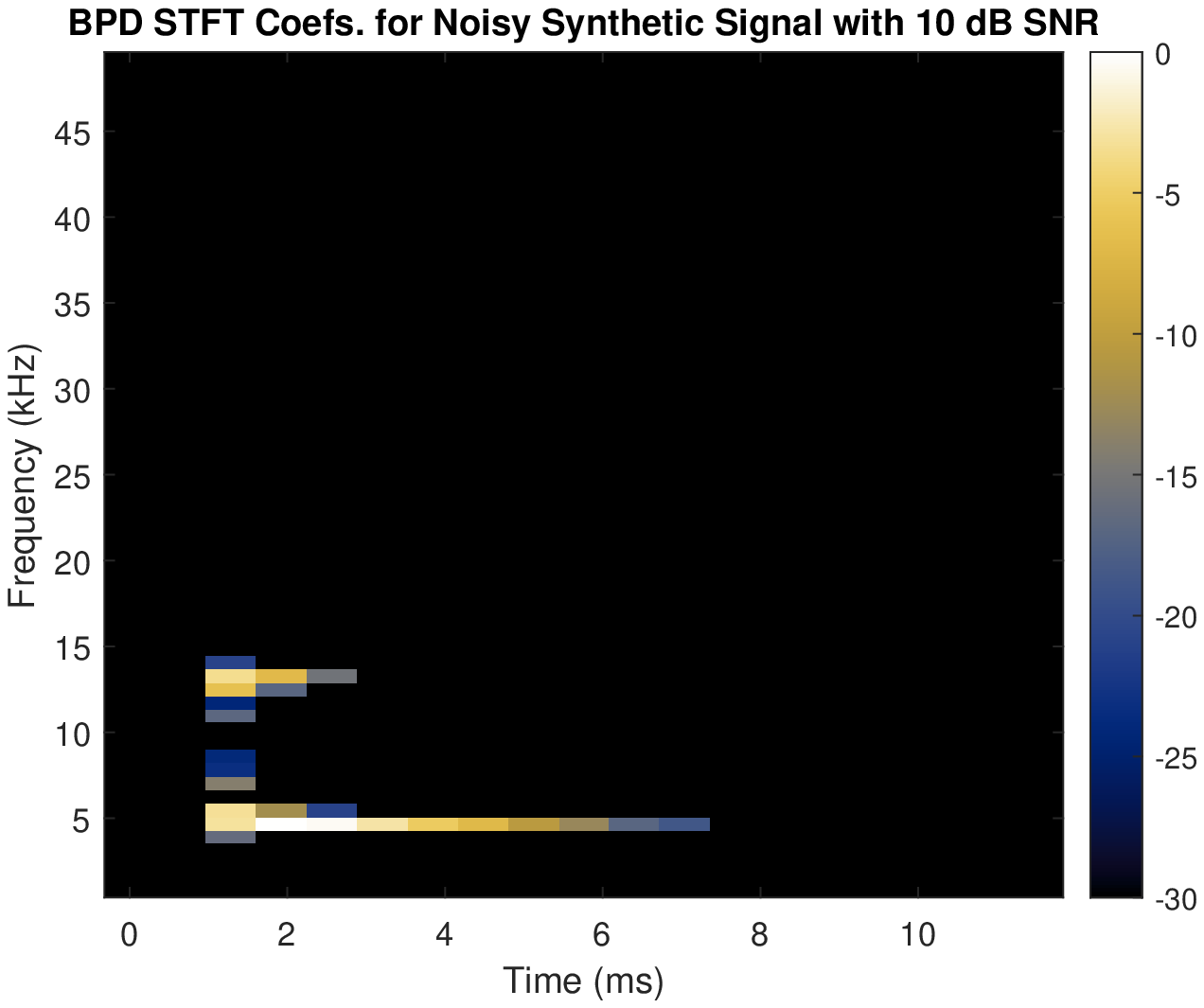}
\caption{Plots of \(\bh\), \(\bh_N\) with SNR 10dB and the reconstructed regularized signal (left) and scaled color plot of the STFT frame coefficients for \(\bh_N\) with a window length of 128 (right).  The coefficients have been regularized using BPD with \(\lambda = 0.1\lambda_{\max}\) and 1000 iterations.  The projected intensities are plotted in dB relative to the maximum frame coefficient amplitude.  The STFT frame coefficient sparsity is 97.61\% and the relative error between the reconstructed signal and \(\bh\) is 24.9\% (12.1dB reconstructed SNR).  }
\label{fig:rhstft}
\end{figure}

\subsection{Experimental Data}
\label{sec:denoising-experimental}

The datasets for this section were collected by tapping a steel cylinder and a wooden cylinder with an impact hammer and recording the emitted sound. The recordings were taken in an anechoic chamber with a sampling frequency of \(f_s = 16\) kHz, after downsampling. The hammer was outfitted with a force sensor which triggered the time series to start recording at the moment the hammer impacted the cylinder. Overall 2,200 samples were taken over 0.55 seconds for each tap. The normalized spectral power densities for both cylinders are displayed in Figure~\ref{fig:spec}. The densities were computed using the first 250ms of the corresponding time series.

\begin{figure}
\centering
\includegraphics[width=1.58in]{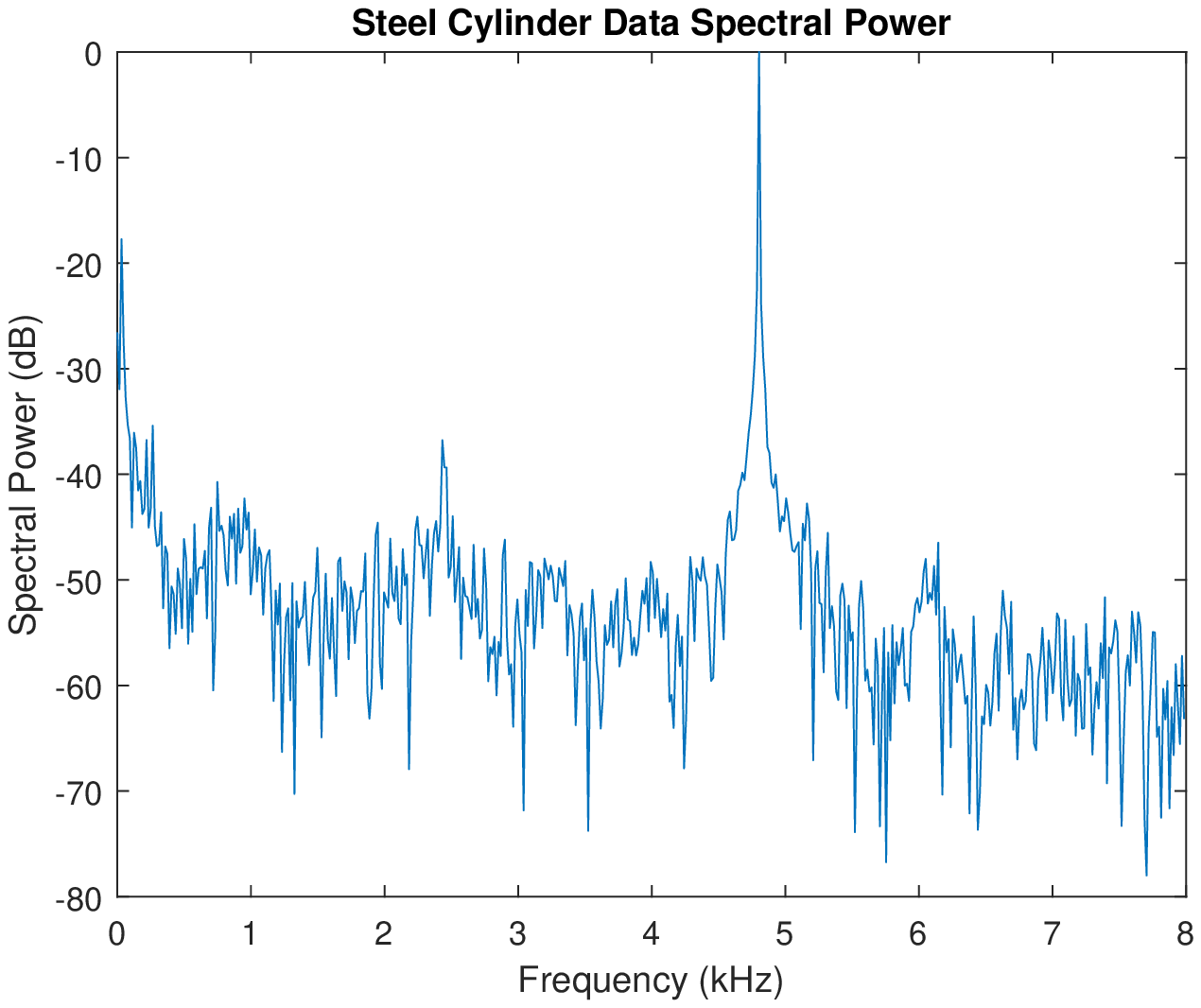}
\includegraphics[width=1.58in]{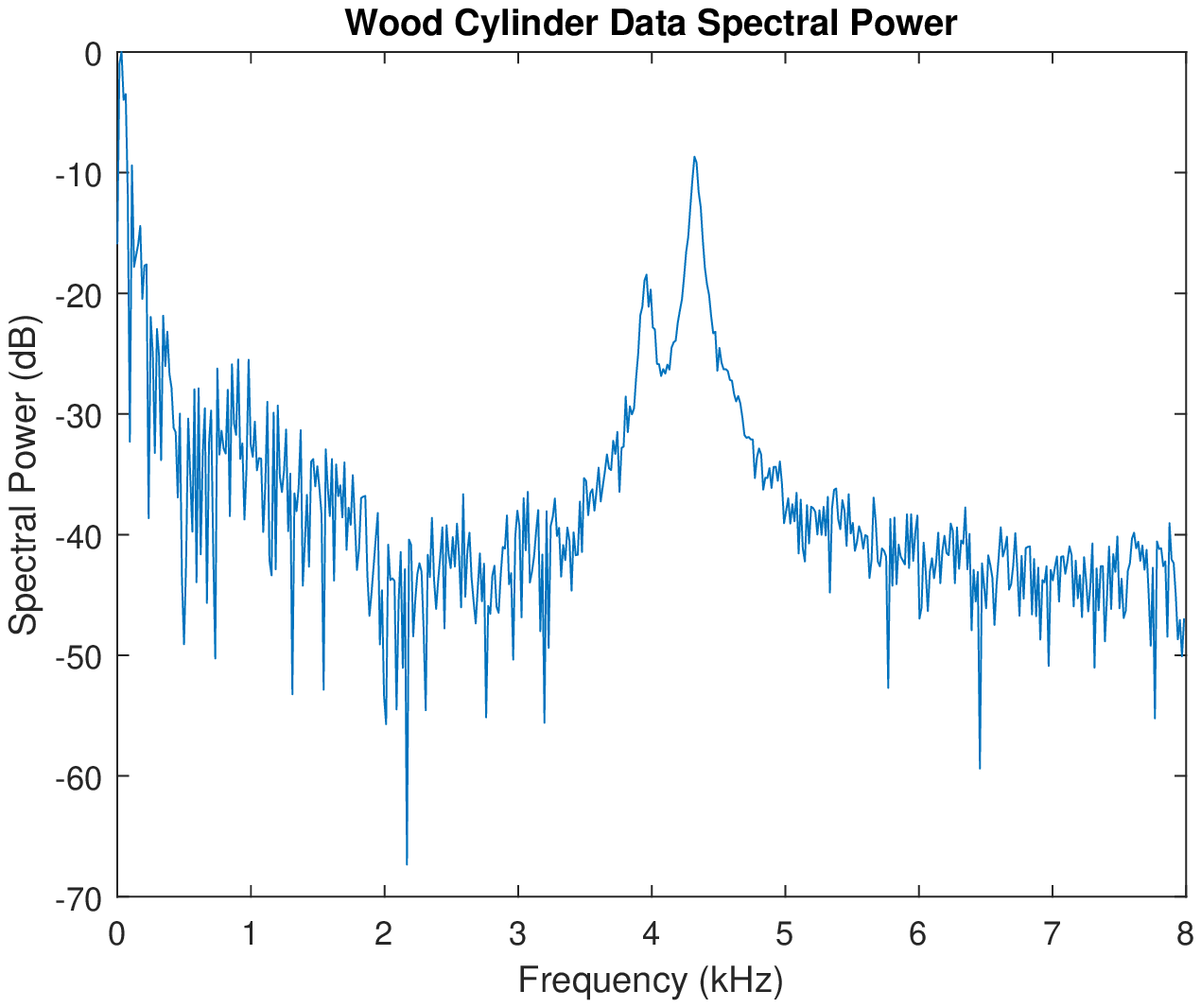}
\caption{Spectral power density for the steel cylinder (left) and wood cylinder (right) datasets. The primary peak for the steel cylinder dataset occurs at 4.83kHz.  The primary peak for the wood cylinder dataset occurs at 4.35kHz and a secondary peak occurs at 3.94kHz.  }
\label{fig:spec}
\end{figure}

In order to avoid any transient effects of the impact, the signal used for analysis is \(N = 1024\) samples or \(T = 70.25\)ms long, starting at the 100th sample. Since the signal is decaying, the same exponential envelopes \(e_l\) defined in \eqref{eq:exp-frame} are appropriate for this analysis, with time constants
\[
\tau_l = 10^{l/4-3}\ \text{for }l = 0, \ldots, 10,
\]
so \(L = 11\) and the \(\tau_l\) range from 1ms to 316ms. About 11.5 million ESP frame vectors \(\ba_{l,k,m}\) follow via \eqref{eq:exp-frame}, and are parameterized by
\begin{itemize}
\item the time constant \(\tau_l\) (1ms - 316ms),
\item the frequency \(f_k\) (-8kHz - 8kHz),
\item and the circular time shift \(t_m\) (0ms - 70.25ms).
\end{itemize}


The unregularized steel and wood cylinder ESP coefficients are displayed in Figure \ref{fig:tap}.  The main response in the 4-5kHz frequency range is clearly visible in both sets of coefficients, as is the secondary frequency response in the wood cylinder time series.  However, the coefficients are spread across the entire time shift axis, which is in conflict with the physical setup of the experiment, and there is no clear separation in the time constant axis.

\begin{figure}
\centering
\includegraphics[width=3.2in]{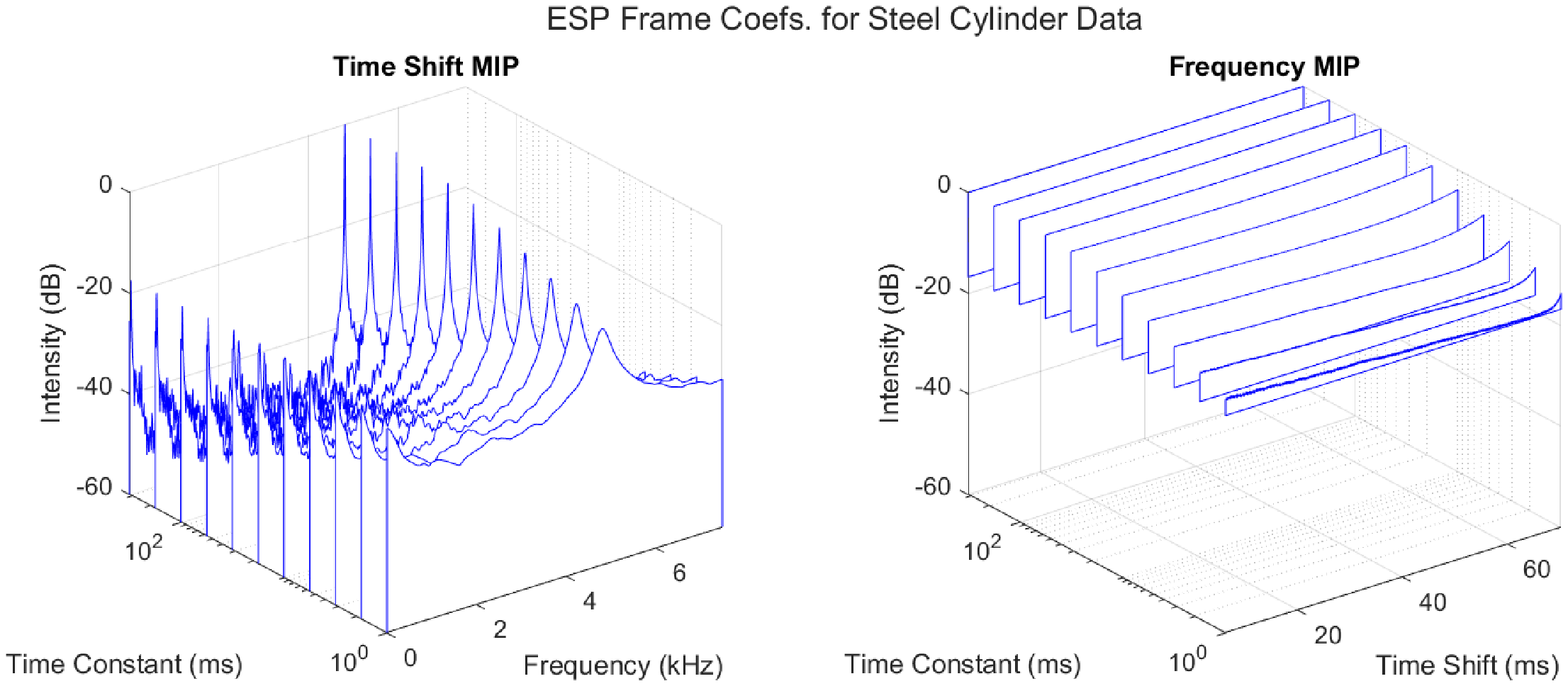}

\vspace{.1in}
\includegraphics[width=3.2in]{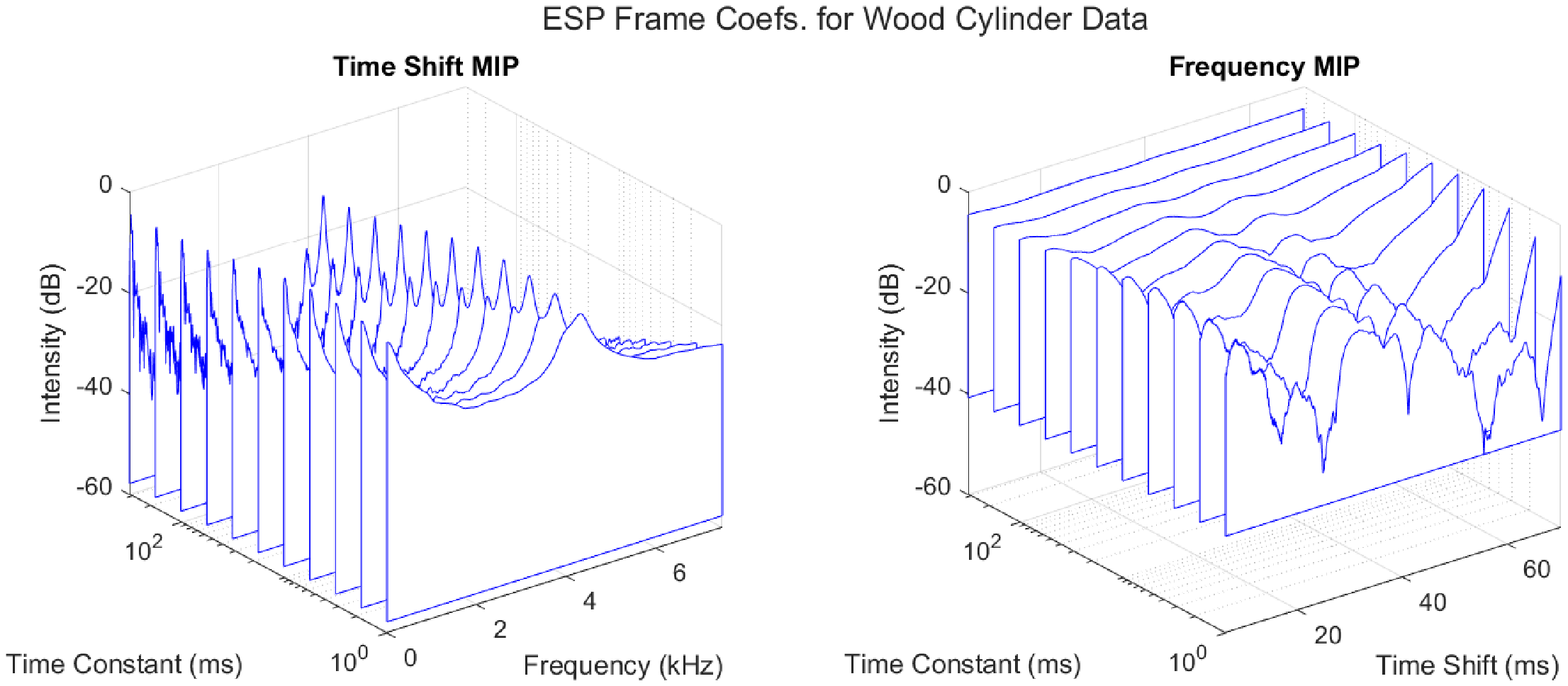}
\caption{Frame coefficient time shift MIP separated by time constant (left) and frequency MIP separated by time constant (right) for steel cylinder (top) and wood cylinder (bottom) tap data.  Intensities are shown on a dB scale relative to the maximum frame coefficient amplitude. The resonant peaks are located at 4.80kHz with an estimated time constant of 177.9ms for the steel cylinder and at 4.32kHz with an estimated time constant of 5.66ms for the wood cylinder.}
\label{fig:tap}
\end{figure}

The unregularized STFT frame coefficients are shown in Figure \ref{fig:tapstft}.  The primary frequency responses are clearly visible, and the time constant for the wood cylinder is apparently shorter than that for the steel cylinder (as expected). The signal power is much more spread out across time when compared to the synthetic signals in previous sections.

\begin{figure}
\centering
\includegraphics[width=1.58in]{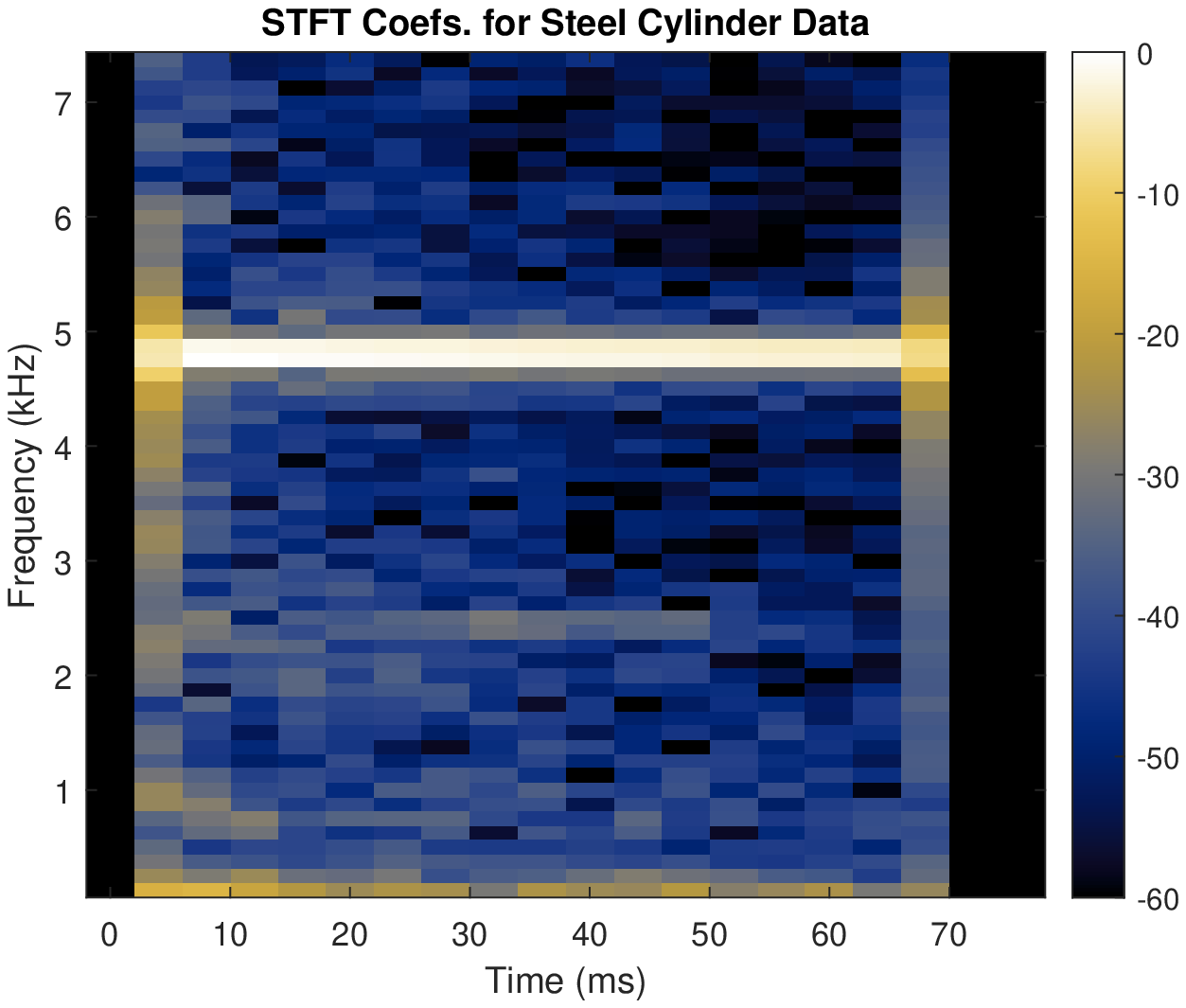}
\includegraphics[width=1.58in]{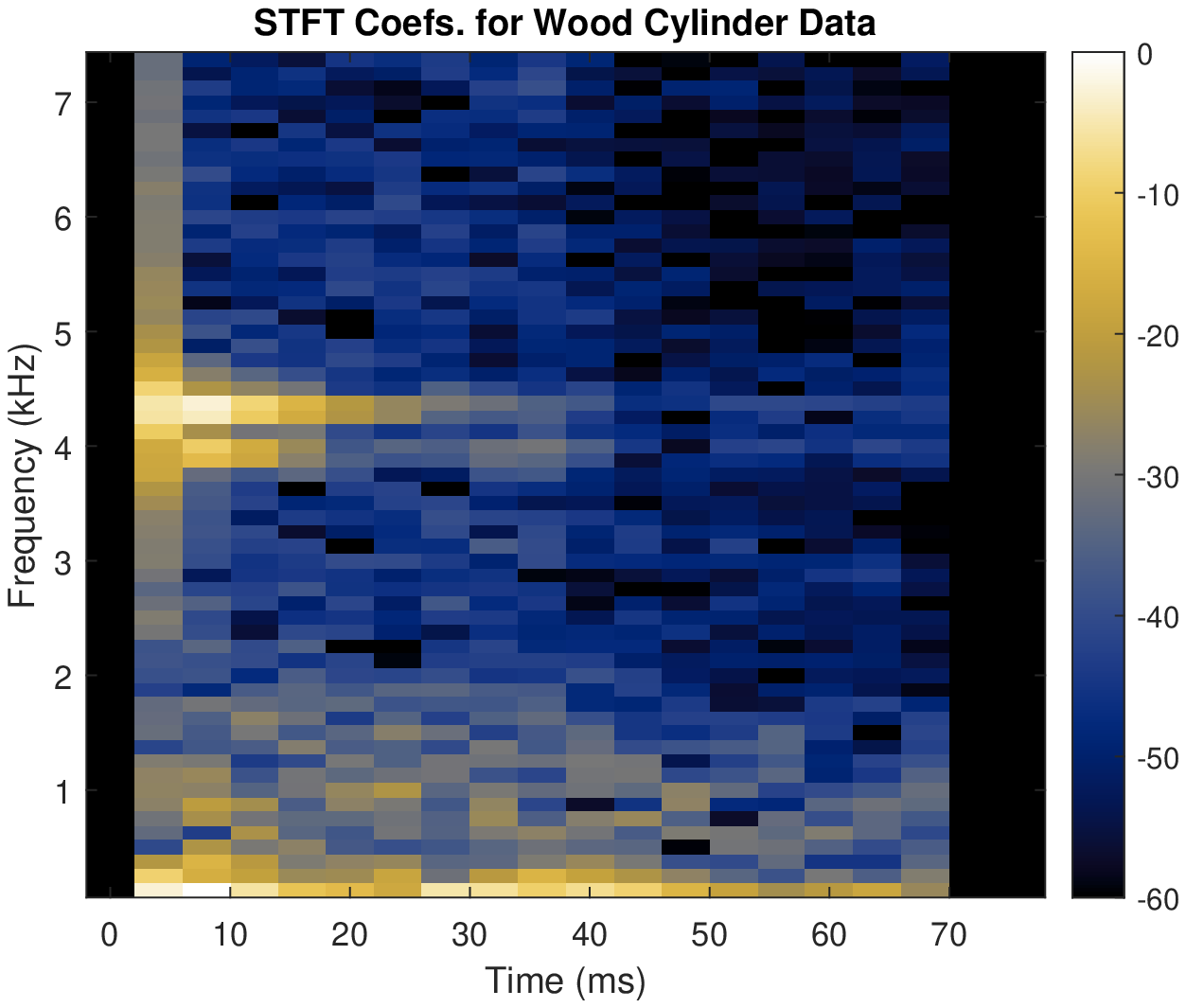}
\caption{Scaled color plots of the STFT frame coefficients for steel cylinder (left) and wood cylinder (right) tap data.  Intensities are shown on a dB scale relative to the maximum frame coefficient amplitude.}
\label{fig:tapstft}
\end{figure}

\subsubsection*{BPD on Experimental Data}
\label{sec:denoising-exp-bpd}

For this section, a noisy signal \(\bh_N\) is generated by adding white Gaussian noise to each experimental time series at 10dB SNR. Figure \ref{fig:rtap} shows the results of BPD applied to this experimental data, using \(1000\) iterations and \(\lambda = 0.1\lambda_{\max}\) as before. The resulting coefficients for the steel cylinder data have a relative reconstruction error of 17.6\% with a reconstructed SNR of 15.1dB, a 5.1dB gain. The regularized coefficients have a sparsity of 99.95\% with 5745 nonzero coefficients. The coefficients for the wood cylinder data have a relative reconstruction error of 27.6\% with a reconstructed SNR of 11.2dB.  The regularized coefficients have a sparsity of 99.97\% with 2998 nonzero coefficients. The primary frequency responses for both objects are visible, as well as the secondary frequency response for the wood cylinder, but the frame coefficients are spread across the entire time shift axis, which is somewhat inconsistent with our expectations for the physics of the experiment and may represent intrinsic (non-Gaussian) experimental noise.

\begin{figure}
\centering
\includegraphics[width=1.58in]{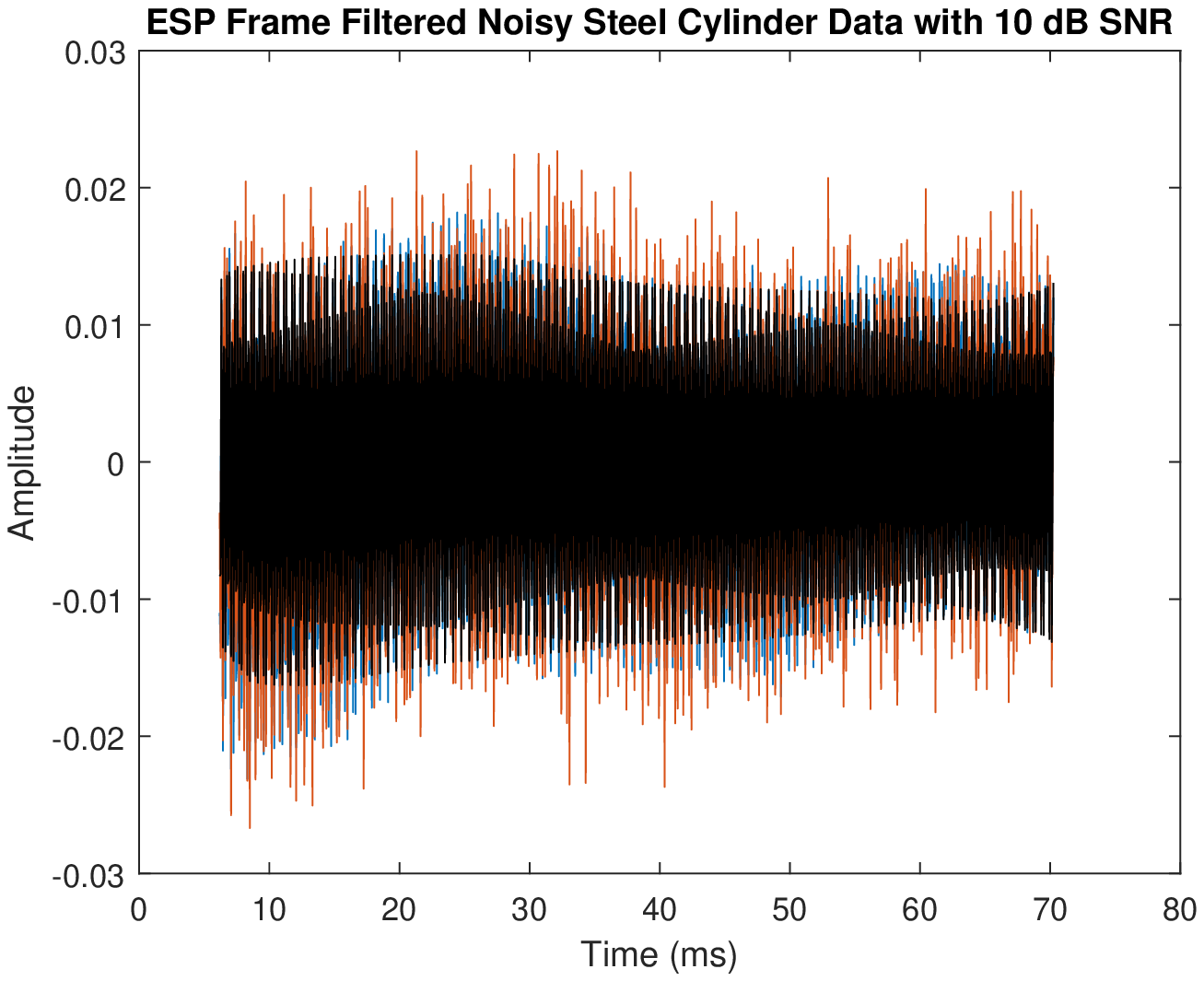}
\includegraphics[width=1.58in]{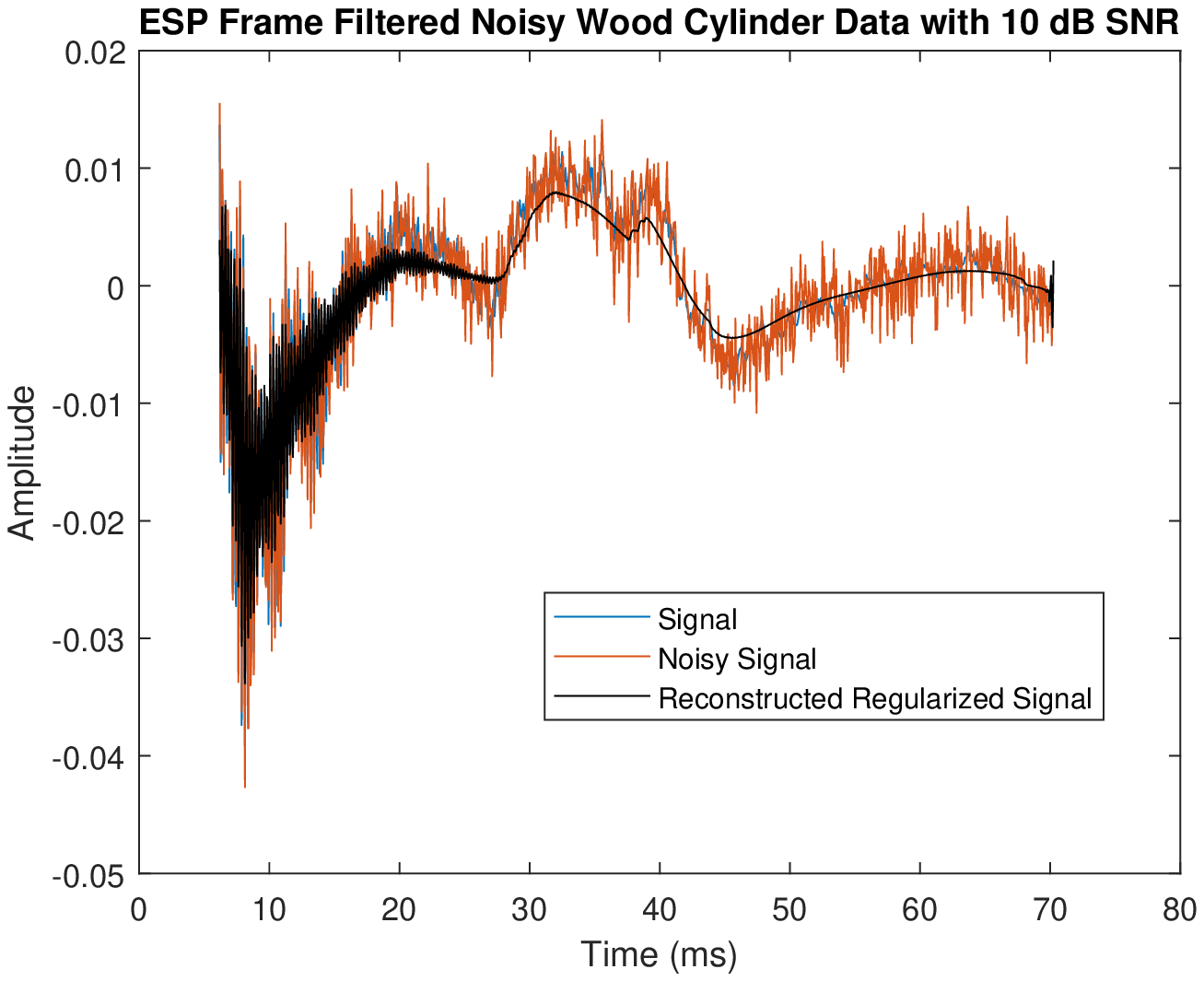}

\vspace{.1in}
\includegraphics[width=3.2in]{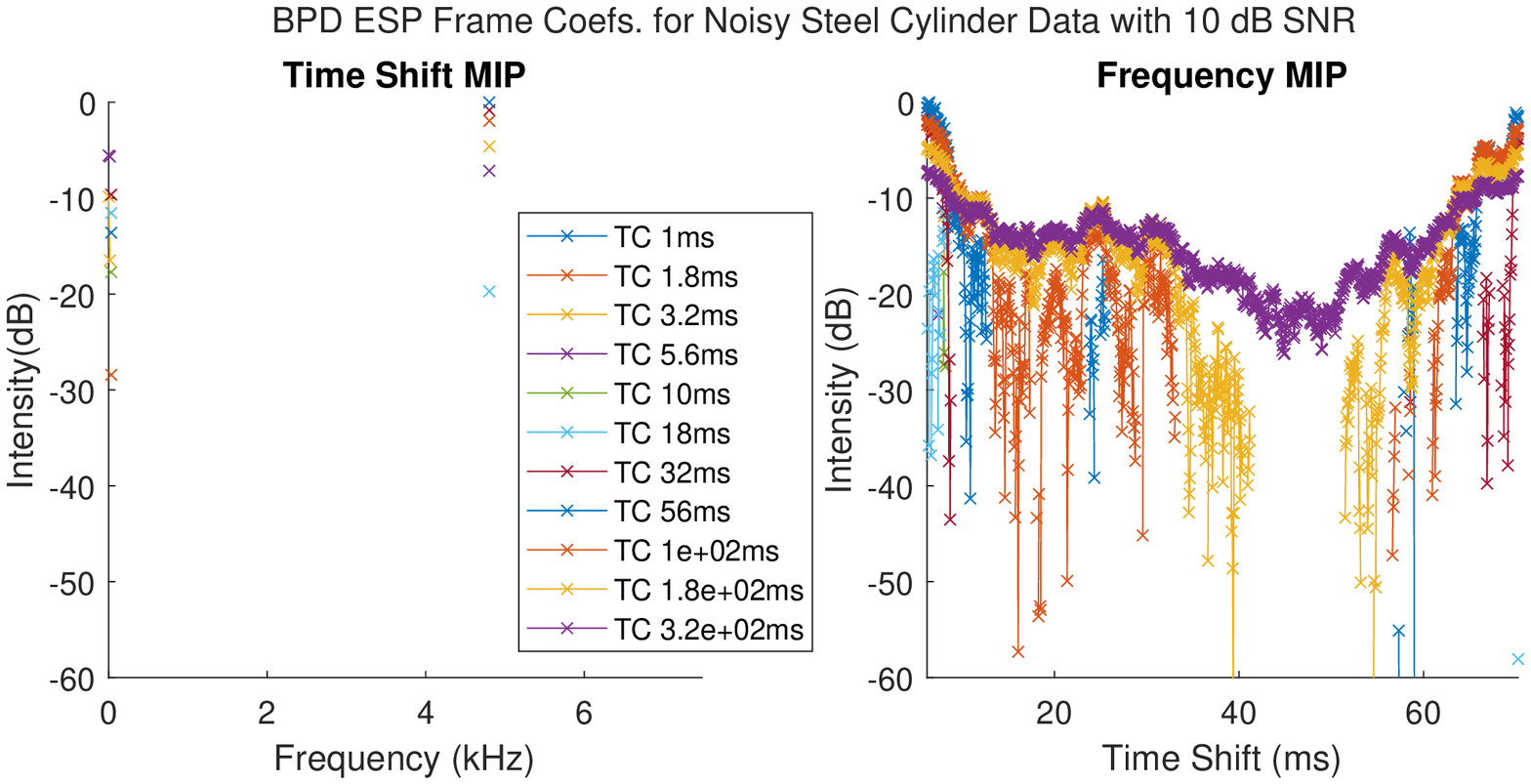}

\vspace{.1in}
\includegraphics[width=3.2in]{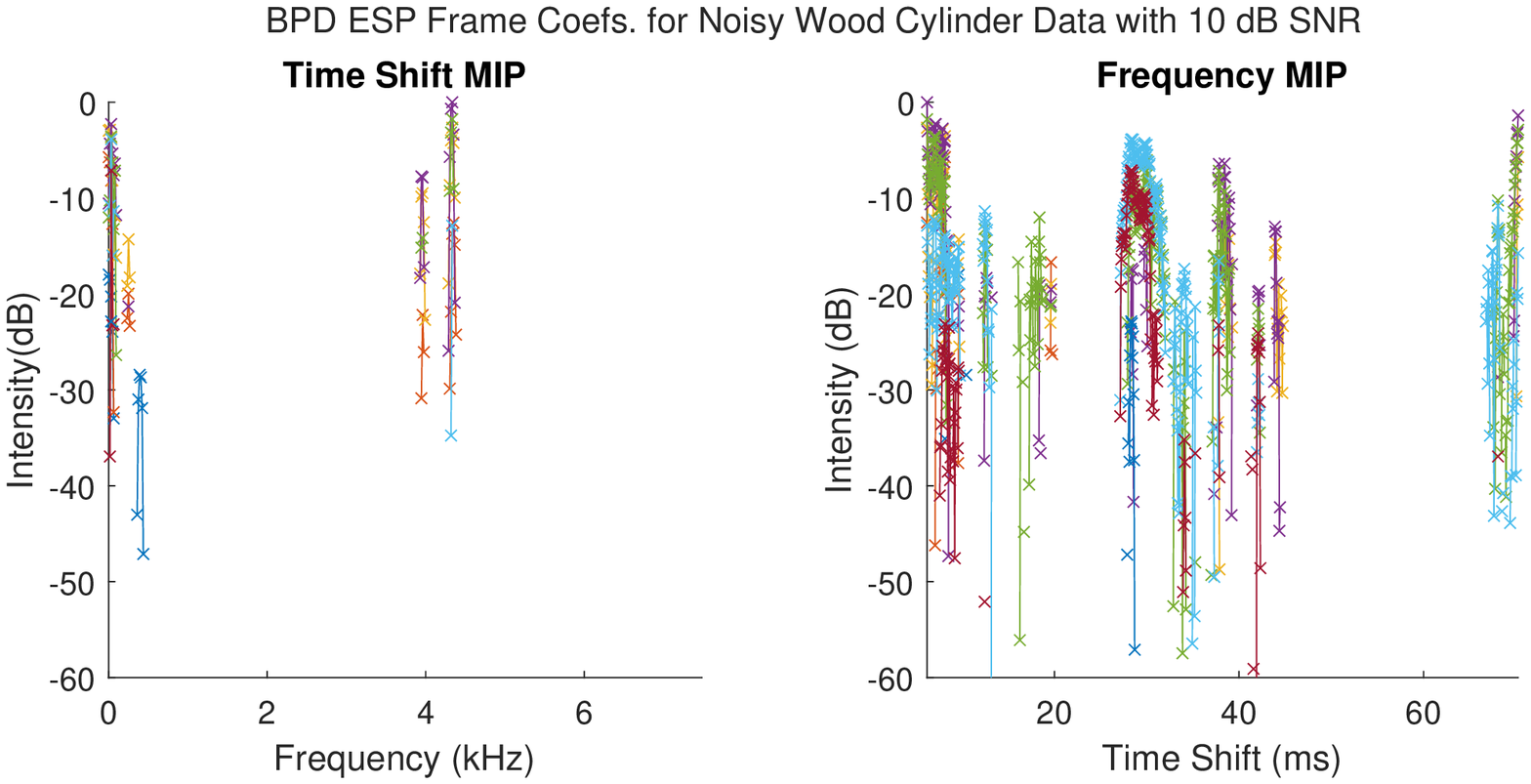}

\caption{Plots of \(\bh\), \(\bh_N\) with SNR 10dB and the reconstructed regularized signal (top left and top right), frame coefficient time shift MIP separated by time constant (middle left and bottom left) and frequency MIP separated by time constant (middle right and bottom right) for steel cylinder (top left and middle) and wood cylinder (top right and bottom) tap data after \(1000\) iterations with \(\lambda = 0.1\lambda_{\max}\).  Intensities are shown on a dB scale relative to the maximum frame coefficient amplitude. The coefficients for the steel cylinder data have a relative reconstruction error of 17.6\% (15.1dB reconstructed SNR) and a sparsity of 99.95\% and the peak frequency occurs at 4.80kHz with an estimated time constant of 54.95ms.  The coefficients for the wood cylinder data have a relative reconstruction error of 27.6\% (11.2dB reconstructed SNR) and a sparsity of 99.97\% and the peak frequency occurs at 4.33kHz with an estimated time constant of 5.72ms. }
\label{fig:rtap}
\end{figure}

Since the regularization process prioritizes sparsity, it tends to filter out both the noise and low magnitude resonant components. This can be seen for both data sets. In the frequency MIP, all of the frequency components are either clustered around the main frequency peaks or are very low frequency, so that quieter frequency responses (such as the small peak at 2.43kHz in the steel cylinder data) have been filtered out.
A promising observation is that the reconstructed signal for the wood cylinder data seems to have very few high frequency components after 30ms. This is consistent with the physics of the experiment since the wood cylinder resonant frequencies have short time constants.

The regularized STFT frame coefficients are displayed in Figure \ref{fig:rtapstft}. They are computing using \(1000\) iterations of BPD and \(\lambda  = 0.1\lambda_{\max}\). The STFT coefficients are much sparser than the unregularized STFT coefficients, with very little power outside the primary frequency responses. The steel cylinder coefficients have a reconstruction error of 17.3\% with a reconstructed SNR of 15.2dB and a sparsity of 97.3\% with 68 nonzero coefficients.  The wood cylinder data has a reconstruction error of 38.6\% with a reconstructed SNR of 8.2 dB, a 1.8 dB loss, and a sparsity of 99.1\% with 24 nonzero coefficients.

\begin{figure}
\centering
\includegraphics[width=1.58in]{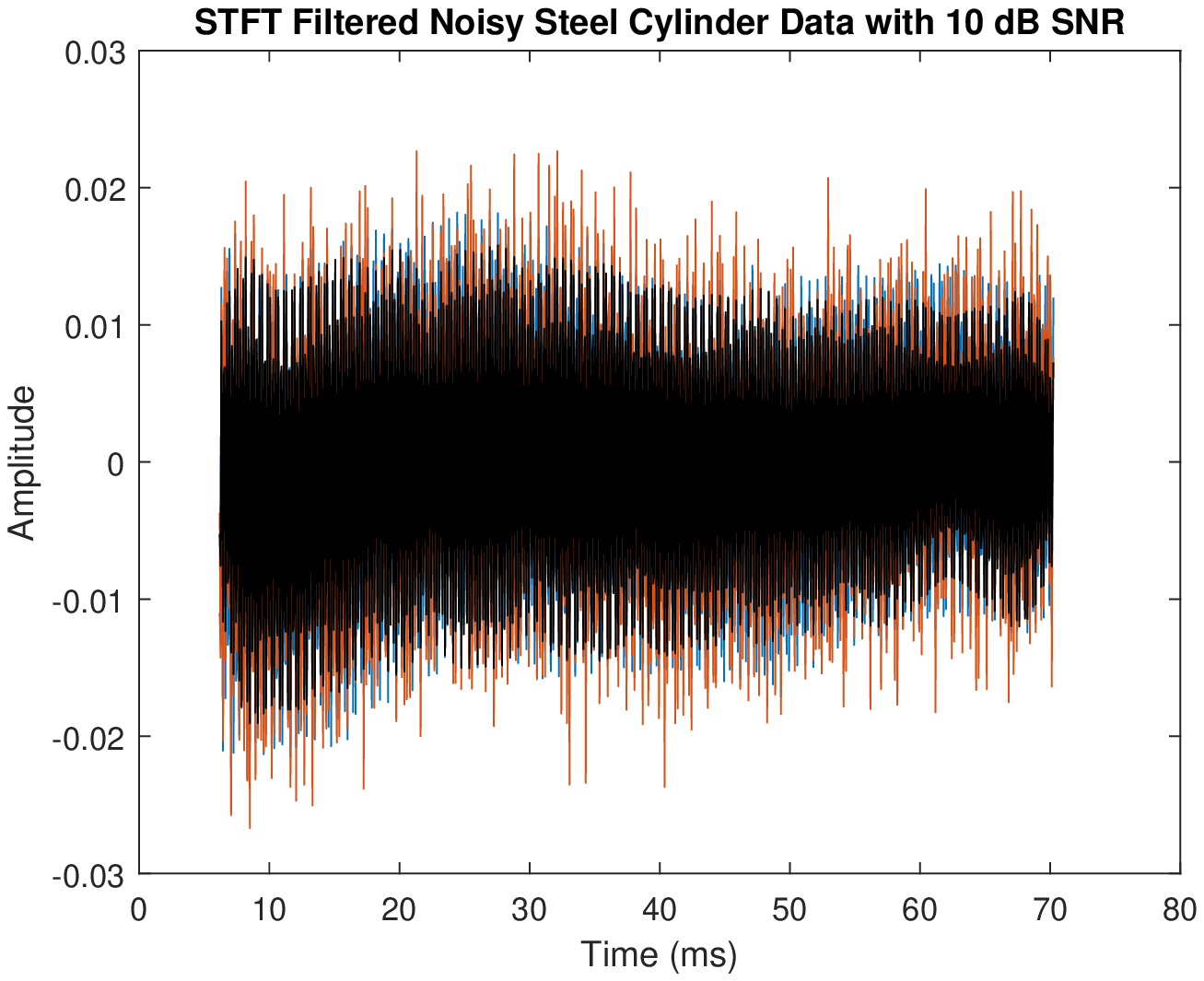}
\includegraphics[width=1.58in]{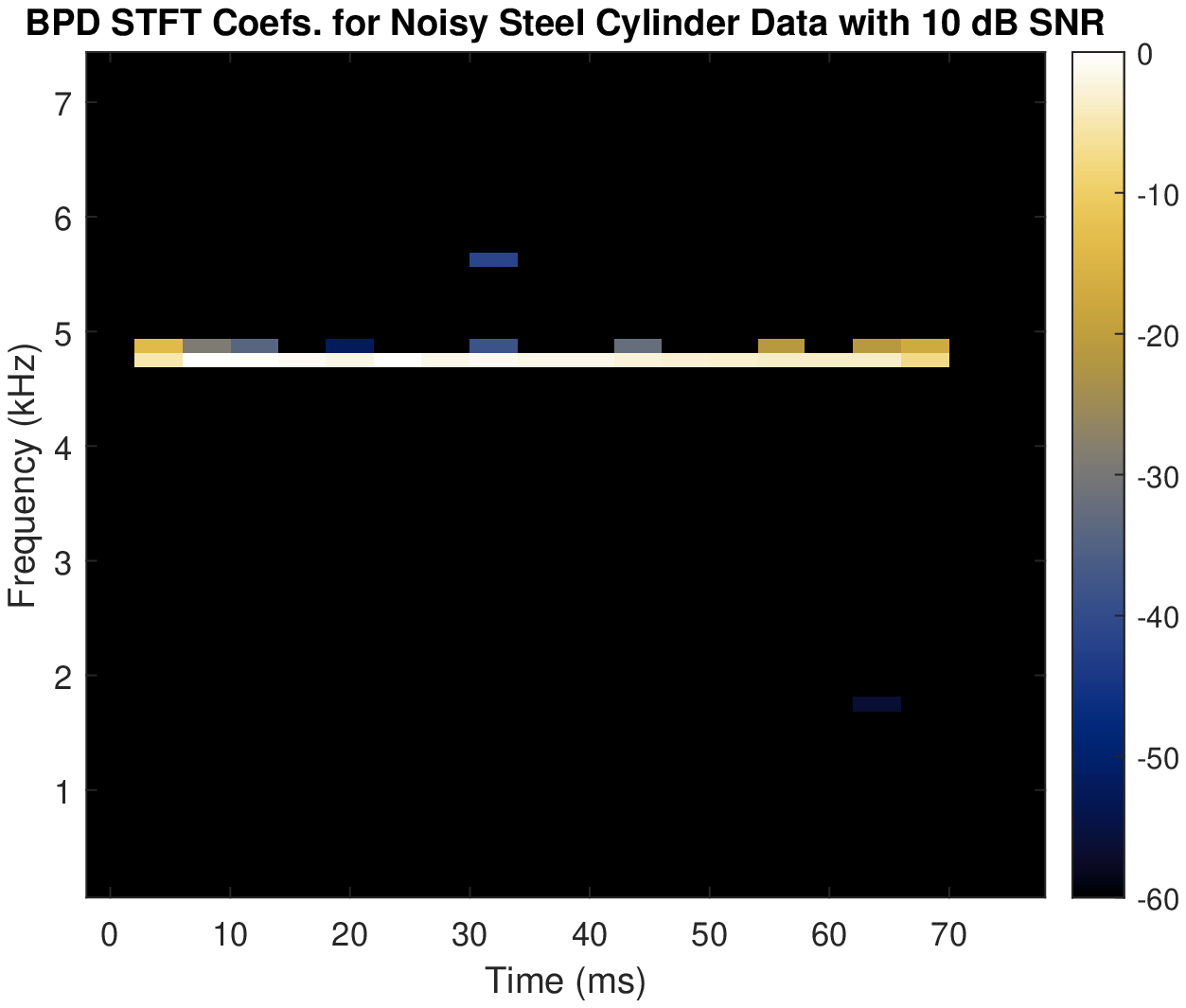}

\vspace{.1in}
\includegraphics[width=1.58in]{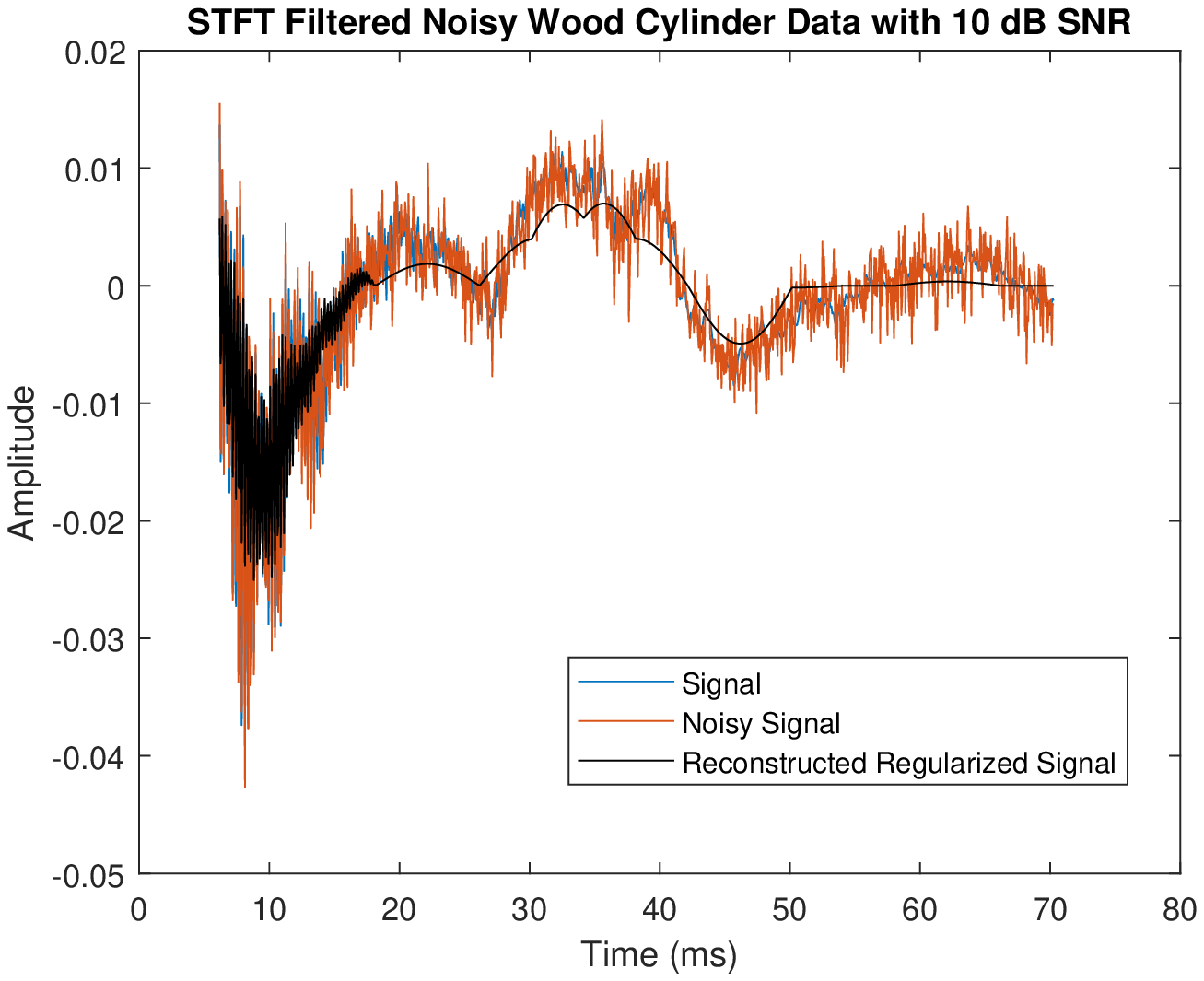}
\includegraphics[width=1.58in]{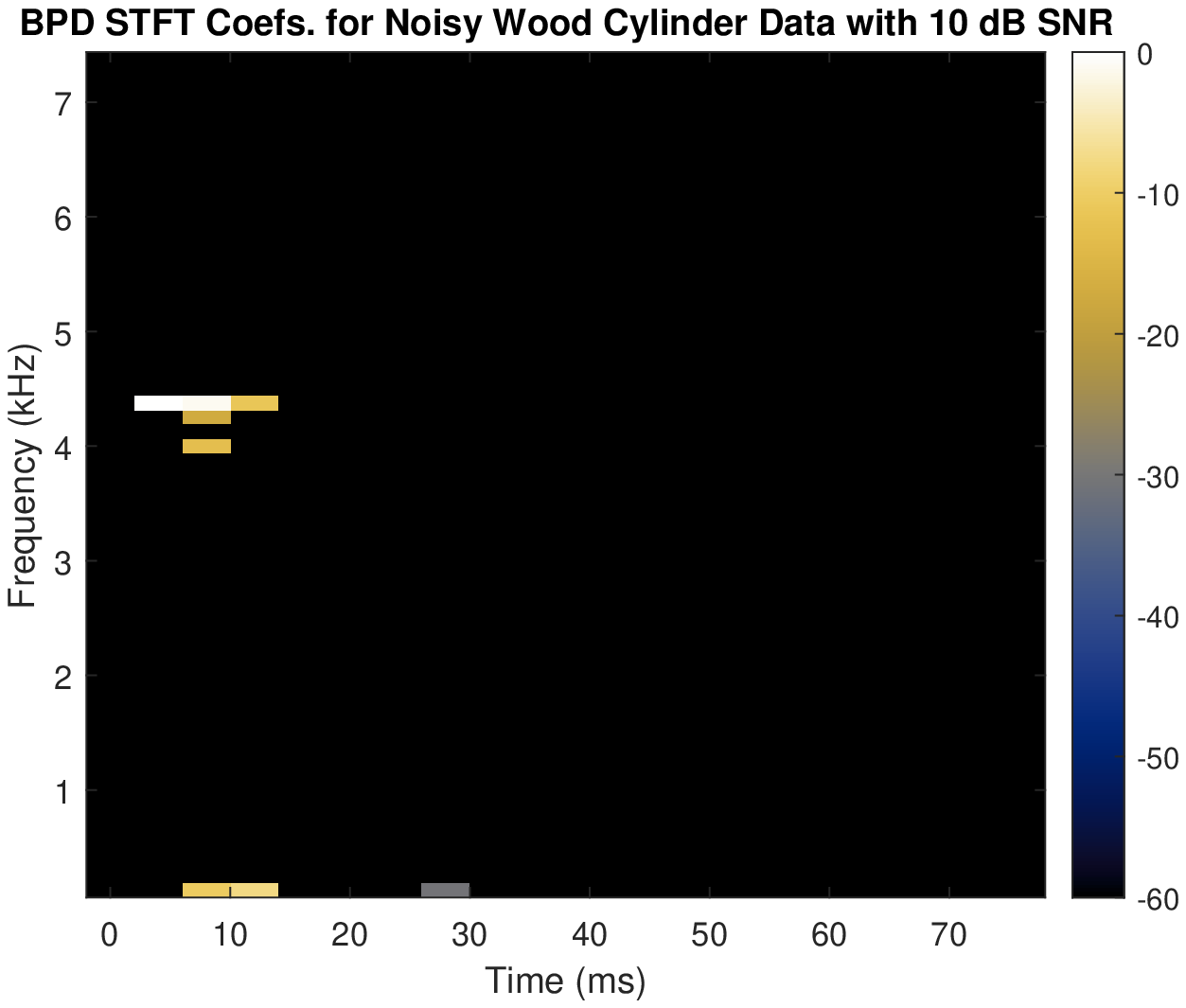}
\caption{Plots of \(\bh\), \(\bh_N\) with SNR 10dB and the reconstructed regularized signal (left) and scaled color plots of the STFT frame coefficients (right) for steel cylinder (top) and wood cylinder (bottom) tap data after \(1000\) iterations with \(\lambda = 0.1\lambda_{\max}\).  Intensities are shown on a dB scale relative to the maximum frame coefficient amplitude. The coefficients for the steel cylinder data have a relative reconstruction error of 17.3\% (15.23dB reconstructed SNR) and a sparsity of 97.3\%.  The coefficients for the wood cylinder data have a relative reconstruction error of 38.6\% (8.26dB reconstructed SNR) and a sparsity of 99.1\%.}
\label{fig:rtapstft}
\end{figure}

For the steel cylinder data, the STFT frame performance is comparable to that of the ESP frame, with the caveat that the actual number of nonzero coefficients is much smaller. This is less true for the wood cylinder data, where the reconstructed SNR goes from a small gain for the ESP frame to a small loss for the STFT. The STFT frame also visually does a poorer job of reconstructing the early time series (see bottom left plot in Figure~\ref{fig:rtapstft}).

\subsection{Denoising Analysis}
\label{sec:denoising-analysis}

In the previous sections $\lambda$ and SNR were fixed.  In order to illustrate a broader picture of the sparsity/reconstruction tradeoffs for ESP and STFT frames, respectively, Figure~\ref{fig:bpdperf} displays BPD results applied to variety of initial SNRs and \(\lambda\)s. The marker style denotes the algorithm, and colors denote a fixed SNR level (\(-30\)dB, \(-15\)dB, 0dB, 15dB, 30dB). A range of percentages of $\lambda_{\max}$ were computed for each algorithm. Each resulting coefficient vector is a point on the graph, where its position along the \(x\)-axis corresponds to its percentage of nonzero coefficients. Note $\lambda_{\max}$ is computed separately for each frame using~\eqref{eq:lambdamax}, and the \(y\)-axis is the reconstructed SNR.

\begin{figure}
\centering
\includegraphics[width=2.4in]{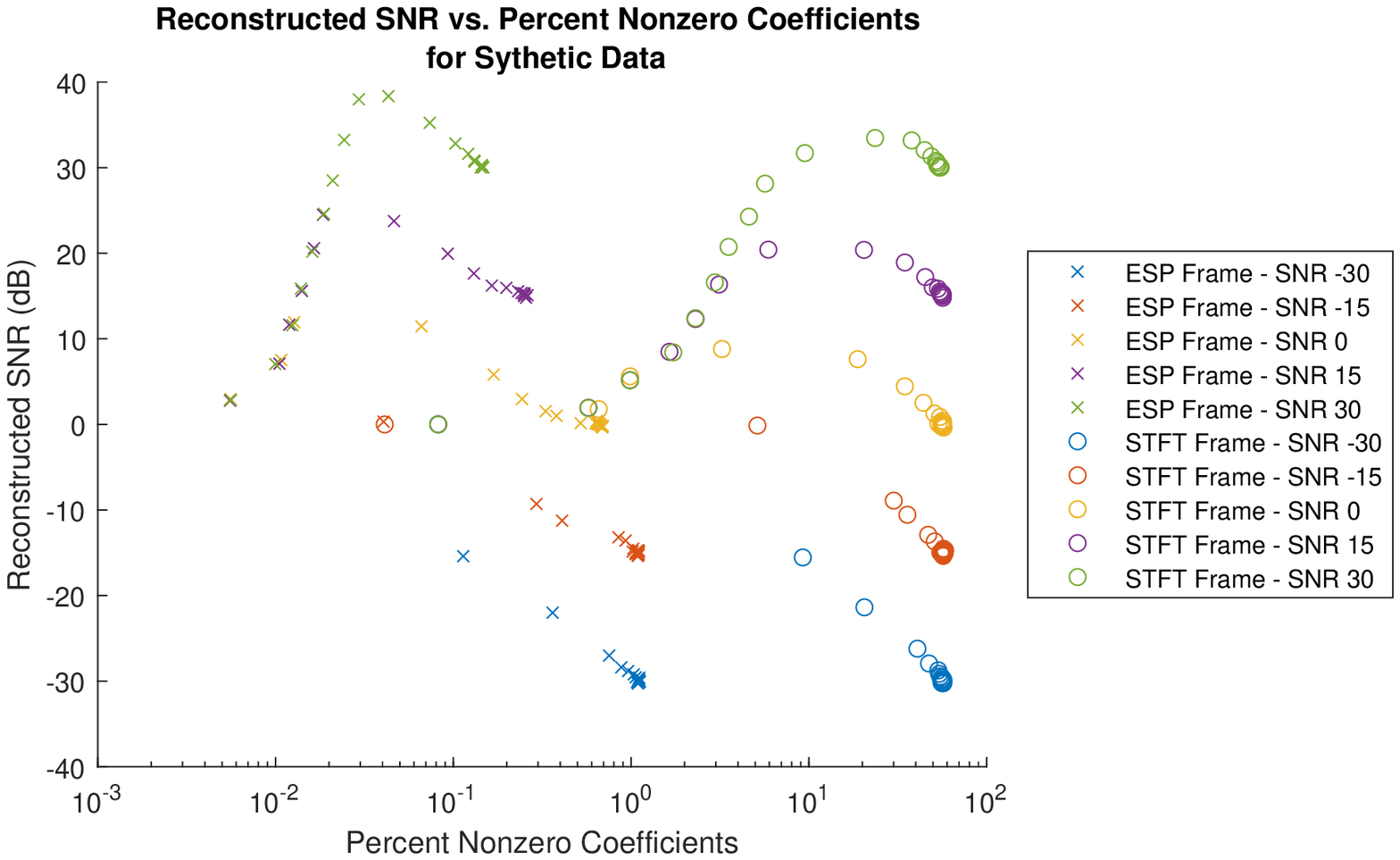}

\vspace{.1in}
\includegraphics[width=1.58in]{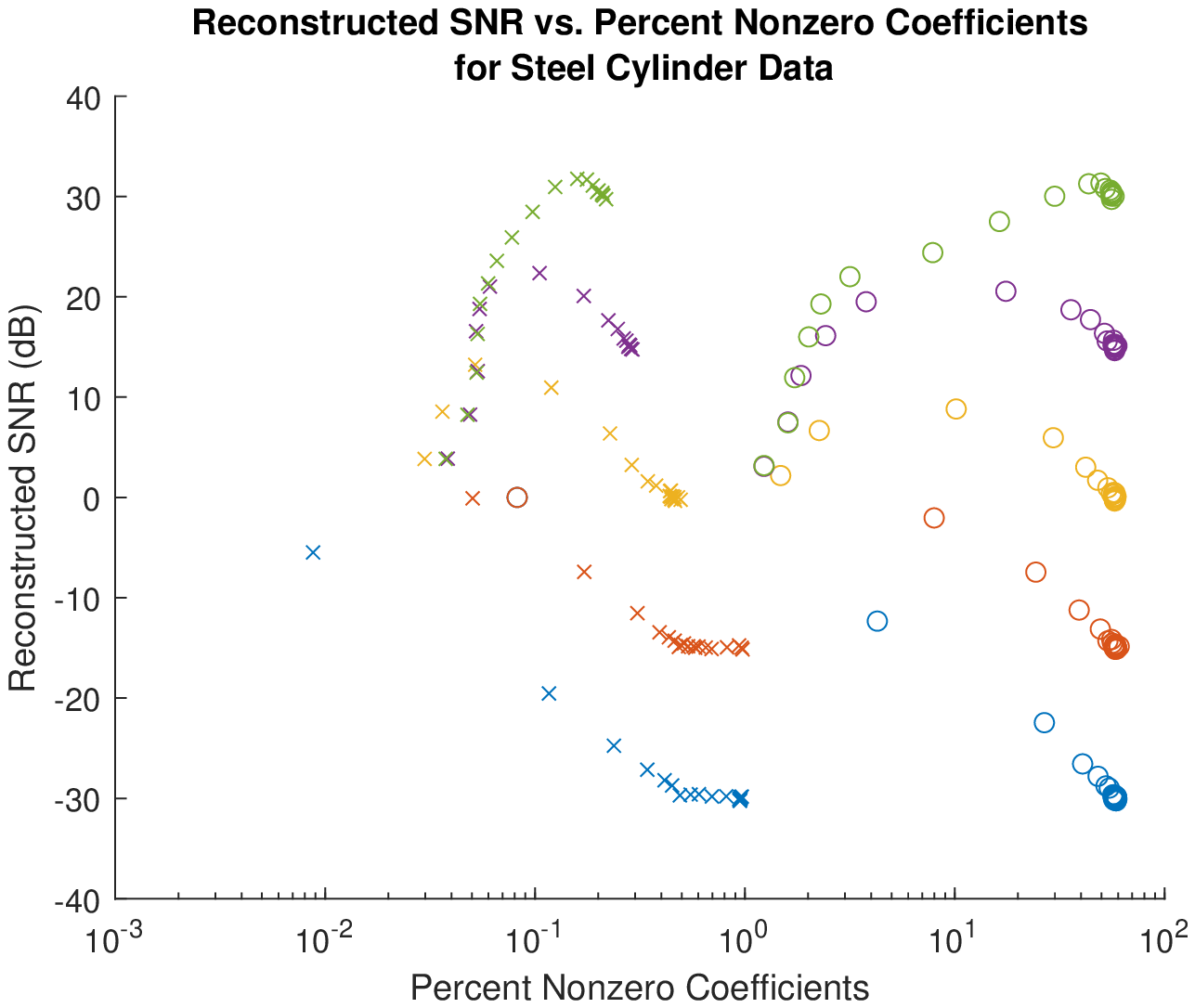}
\includegraphics[width=1.58in]{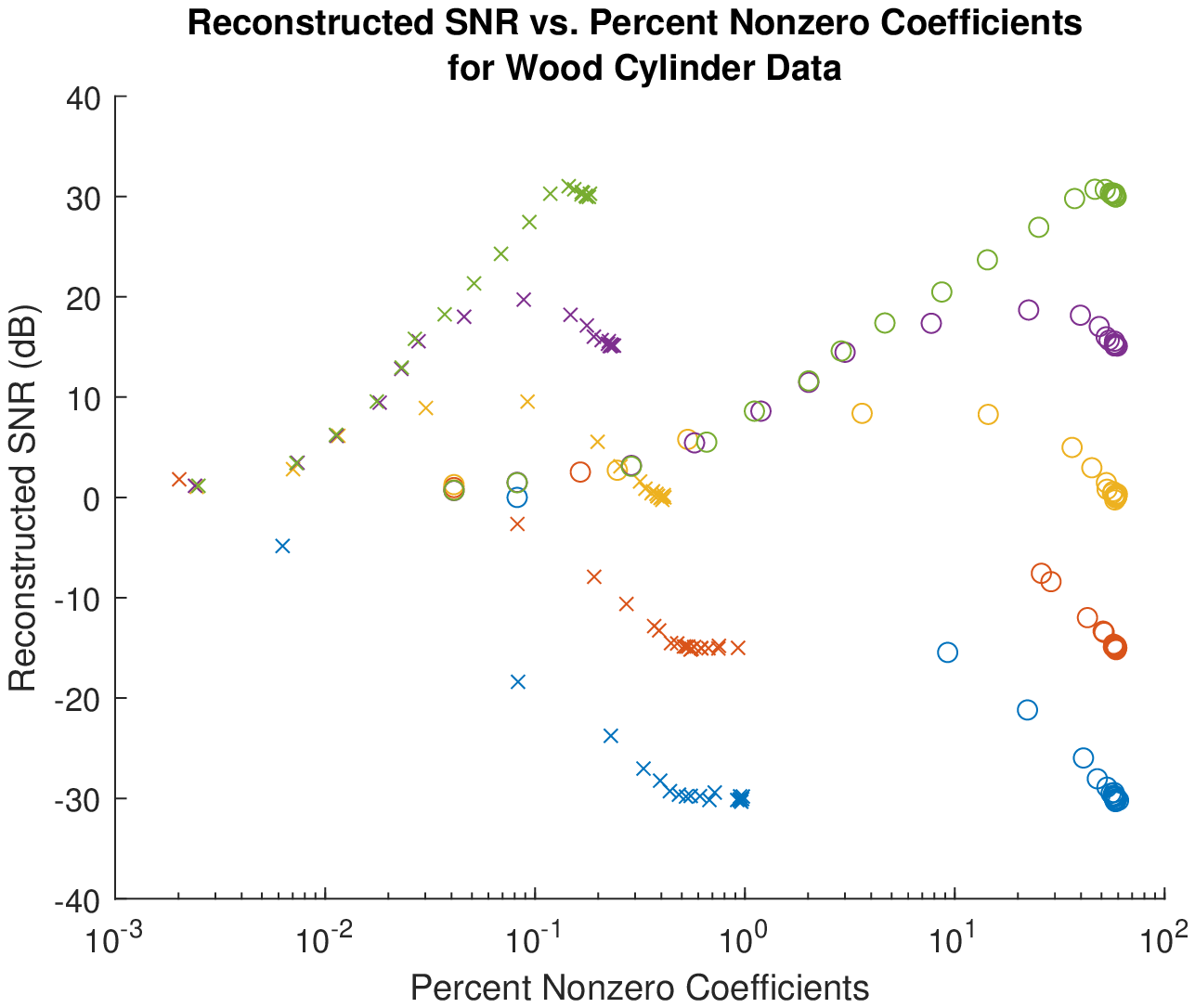}
\caption{Reconstructed SNR vs. percentage of nonzero coefficients for the synthetic signal (top), steel cylinder time series (bottom left) and wood cylinder time series (bottom right) using both an ESP frame and a STFT frame over SNR ranging from \(-30\)dB to 30dB. The BPD algorithm was applied to the time series with added Gaussian white noise at the specified initial SNR using \(1000\) iterations and \(\lambda\) ranging logarithmically from \(0.00001\lambda_{\max}\) to \(\lambda_{\max}\). }
\label{fig:bpdperf}
\end{figure}

In general, the STFT frame has a greater percentage of nonzero coefficients than the corresponding ESP frame coefficients, generally falling in the 50-60\% range. This is expected because the ESP frame is highly overdetermined and therefore has many more coefficients. For small \(\lambda\) the percentage of nonzero ESP frame coefficients ranges from approximately 15\% for large SNR to approximately 30\% for small SNR.

When the SNR of the noisy signal is less than zero, the reconstructed SNR increases and the number of nonzero coefficients decreases as \(\lambda\) grows until the BPD solution becomes zero.  This indicates the denoising process was unsuccessful.
Conversely, when the SNR of the noisy signal is greater than or equal to zero the reconstructed SNR increases as \(\lambda\) grows until it hits some maximum, and then decreases until it reaches zero. In this case, the regularization process is initially removing more noise power than signal power, resulting in an SNR gain, until it hits some optimal \(\lambda\) that depends on the frame and the signal.  After this optimal value, further increases in \(\lambda\) cause the regularization process to reduce both noise and signal power until the BPD solution becomes zero at \(\lambda_{\max}\).

The gain of this optimal reconstructed SNR over the SNR of the noisy signal is indicative of the frame's ability to denoise and is shown in Table \ref{tab:snr}.  The key takeaway of Table \ref{tab:snr} is that the optimal SNR gain is consistently larger for the ESP frame than it is for the STFT frame.  This is true across all three signals for the range of noise SNR's which produced successful denoising.  This maximum SNR gain is greater than was necessarily indicated by the examples in Sections \ref{sec:denoising-synthetic-bpd} and \ref{sec:denoising-exp-bpd}, particularly in the case of the wood cylinder time series.

\begin{table}
\centering
\begin{tabular}{c | c c c c c c}
~ &  0dB SNR & 15dB SNR & 30dB SNR\\ \hline
Synthetic ESP Gain&  {\bf11.9} &  {\bf9.5} & {\bf 8.5} \\
Synthetic STFT Gain &  8.8 &  5.4  & 3.4 \\ \hline
Steel ESP Gain &  {\bf13.2}  & {\bf7.4}  & {\bf1.8} \\
Steel STFT Gain &  8.8  & 5.5  & 1.3\\ \hline
Wood ESP Gain &  {\bf9.5} &  {\bf4.7}   & {\bf 1.0} \\
Wood STFT Gain &  8.3 &  3.7  & 0.7\\
 \hline \vspace{0.1mm}
\end{tabular}
\caption{Maximum Reconstructed SNR gain for signals denoised using BPD with ESP and STFT frames.
}
\label{tab:snr}
\end{table}

Overall we find that the ESP and STFT frames have similar performance with regards to denoising.  However, the ESP frame based denoising produces gains which range from 0.3dB to 5.1dB higher than the STFT frame.  The ESP frame gain is better for the synthetic signal than for the experimental signals as the ESP frame vectors are more closely aligned to the signal model in the synthetic signal case.

\section{Parameter Estimation}
\label{sec:params}

In addition to denoising, another important application of ESP frames is as a parameter estimation tool.  Since the unregularized ESP frame coefficients are correlation based we expect the coefficients to have peaks when the frame vector is well matched with the signal in question.  Our approach will be to use the parameter values associated to these peaks as estimates for feature parameters in the underlying signal.  While this produces an unbiased estimation in the case of a single atom by \eqref{eq:cs}, in the case of signals with multiple components this estimate is not necessarily unbiased.
If the frame is constructed so that the signal is known to have a sparse frame representation then both of these issues can be addressed via \(L_1\)-regularization.  Importantly, if the signal  does {\em not} have a sparse frame representation then the \(L_1\)-regularization process can introduce significant bias into the parameter estimation process.

For this section's analysis we will focus on the identification of resonances in the signals presented in Section \ref{sec:denoising}.  We use the frames discussed in that section to estimate the frequency and time constants with both unregularized and sparse coefficients.  Since the resolution on the time constant axis is poor we will interpolate to get more precise estimates.  As the time constants are sampled on a logarithmic scale we use a {\em geometric} average weighted by the coefficient amplitudes
\begin{align}
\label{eq:tau}
\tilde{\tau} &= \left(\tau_{l-1}^{|c_{k,l-1,m}|}\tau_{l}^{|c_{k,l,m}|}\tau_{l+1}^{|c_{k,l+1,m}|}\right)^{\beta}\ \text{s.t.} \\
\beta &= (|c_{k,l-1,m}|+|c_{k,l,m}|+|c_{k,l+1,m}|)^{-1}. \nonumber
\end{align}

We compare the performance of our ESP frame based estimates to Prony's Method \cite{marple}, a least squares regression based approach for estimating decaying resonances.  The basic concept behind Prony's Method is that we generate a least squares approximation for the signal using a sum of exponentially decaying sinusoids and then use the frequency and time constant associated to the component in the correct frequency range as our estimate.  Prony's Method assumes the signal can be modeled by a sum of decaying exponentials which start at time zero.  When this assumption holds, and the order of the least squares regression matches the number of poles in the signal, Prony's Method is capable of producing extremely accurate estimates.  However, it is also known that noise and late starting signals can adversely affect Prony's Method.  To address these issues, in the case of noise we will be utilizing the SVD-based noise reduction techniques described in \cite[Section 11.9]{marple}.  For late starting signals we will utilize a time shift to ensure the exponential decay starts at time zero.   Unlike the ESP frame approach this requires us to know, or estimate, the number of poles and the start time of the exponential decay.

Section~\ref{sec:param-synth} presents examples estimating resonance parameters using synthetic time series, as well as a comparison between the parameter estimation performance of ESP frames and Prony's Method.  Section~\ref{sec:param-exp} presents a similar analysis using experimental time series and also discusses the use of weighted BPD to encode prior knowledge when generating sparse coefficients.

 \subsection{Synthetic Time Series}
 \label{sec:param-synth}

 For the synthetic time series presented in Section \ref{sec:denoising-synthetic} we know that \(\bh\) contains two resonances with frequencies and time constants of 5kHz and 3ms, and 13kHz and 0.8ms that both start at 0.5ms.  Using Prony's Method on a shifted version of the clean signal \(\bS^{-50}\bh\) with 4 poles we can recover the parameters for each resonance exactly.  However, if we apply Prony's Method directly to \(\bh\) {\em without} the shift we find estimates of 5.03kHz and 4.46ms, and 12.86kHz and 1.21ms.  For comparison if we use the unregularized ESP frame coefficient peaks near 5kHz and 13kHz and \eqref{eq:tau} we estimate the frequency and time constants of the two resonances to be 5kHz and 2.52ms and 13kHz and 0.63ms.  Here we have recovered the frequency components exactly and the error in the time constants is better than when Prony's Method is applied without a shift.  That being said, for all future Prony's Method estimates we will apply any shifts necessary to ensure optimum performance.

While the estimates above were taken from a clean signal, we are generally interested in parameter estimation in the presence of noise.  If we add noise at 10dB SNR, as described in \ref{sec:denoising}, and use Prony's Method (with shifting and 30 poles filtered to 4 using SVD) to estimate the resonance parameters we get 5.01kHz and 2.50ms, and 12.95kHz and 0.55ms.  For comparison if we compute the BPD regularized ESP frame coefficients as in Figure \ref{fig:rh} we obtain estimates of 5kHz and 2.48ms, and 13kHz and 0.95ms.  In this case the regularized BPD approach does a slightly better job of estimating the frequency parameters and is about as accurate as Prony's Method at estimating the time constants.

\subsubsection*{Noise Analysis}
\label{sec:param-noise-synth}

While the above examples indicate that ESP frames can be reasonably utilized as a parameter estimation tool, a further comparative analysis with Prony's Method is warranted.  Specifically we wish to compare the bias and variance of ESP frame and Prony's Method based parameter estimates in the presence of added noise.  To this end noise was added to the synthetic time series at levels ranging from \(-15\)dB SNR, resulting in predominantly noise, to 30dB SNR, resulting in predominantly signal.  At each noise level the frequency and time constant of both resonance peaks was estimated using unregularized ESP frame coefficients, sparse ESP frame coefficients, and Prony's Method.  For the unregularized ESP frame coefficients the parameter estimates were generated using the coefficient peaks near 5kHz and 13kHz and \eqref{eq:tau}.  The regularized coefficients were generated using \(\lambda = 0.1\lambda_{\max}\) and \(1000\) iterations and the same parameter estimation process as the unregularized case.  Finally for Prony's Method we use a shifted version of the noisy time series with 30 poles filtered to 4 via SVD.  Each of these estimates was generated for \(100\) noise realizations and the resulting estimation mean and standard deviation was calculated and are plotted in Figure \ref{fig:ode-param} with bias and standard deviation values at the 30dB level shown in Table \ref{tab:ode-param}.

\begin{figure}
\centering
\includegraphics[width=1.58in]{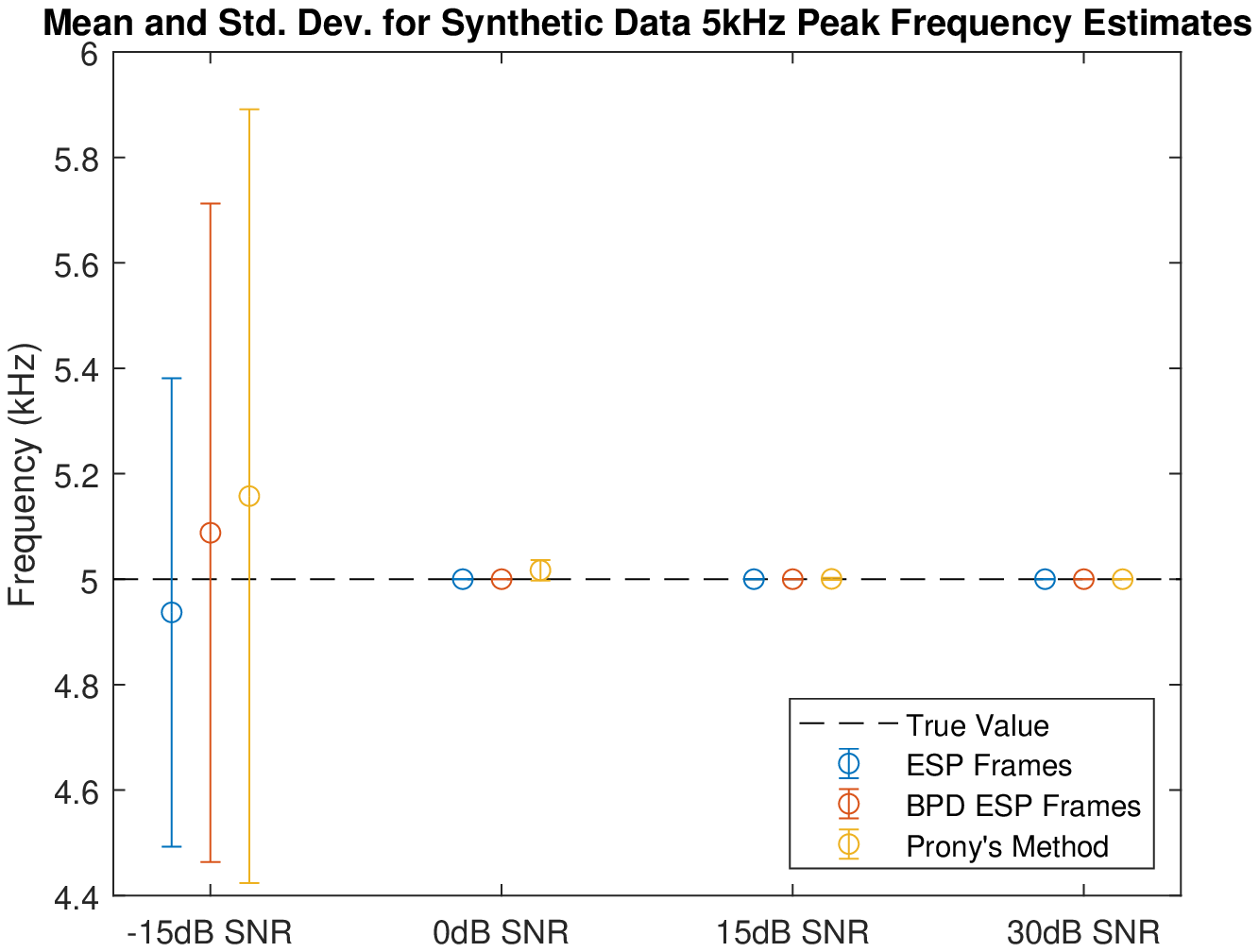}
\includegraphics[width=1.58in]{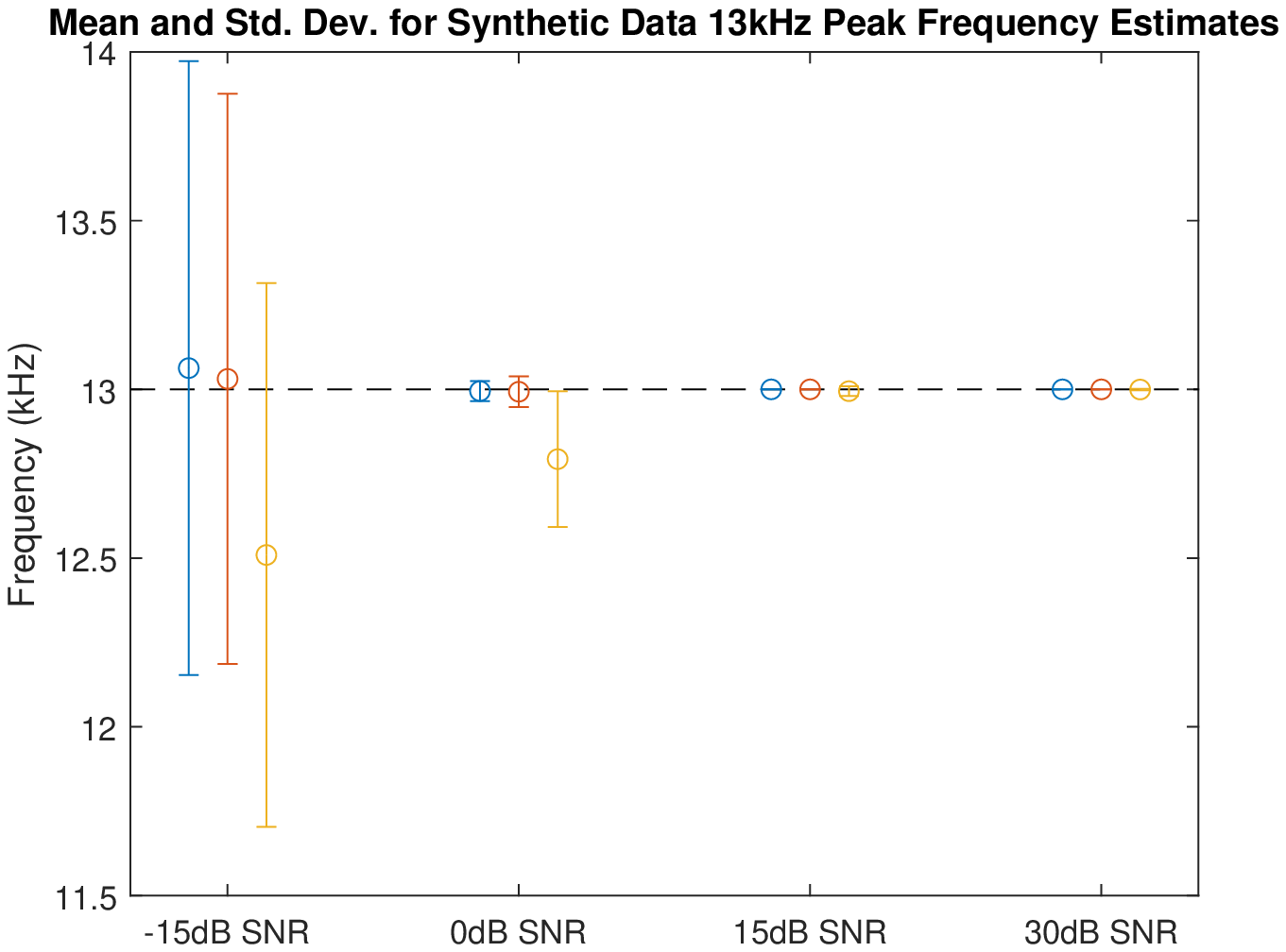}

\vspace{.1in}
\includegraphics[width=1.58in]{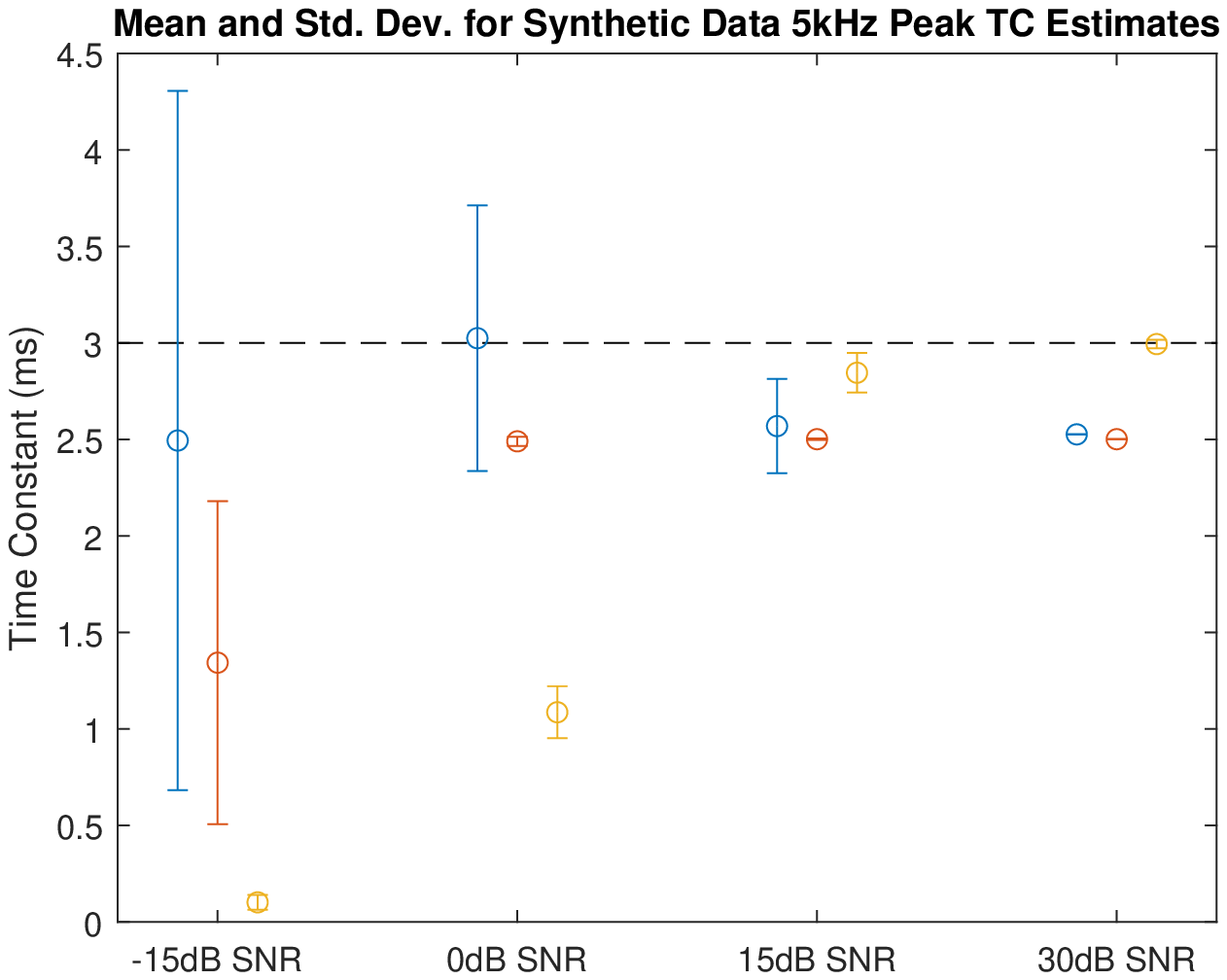}
\includegraphics[width=1.58in]{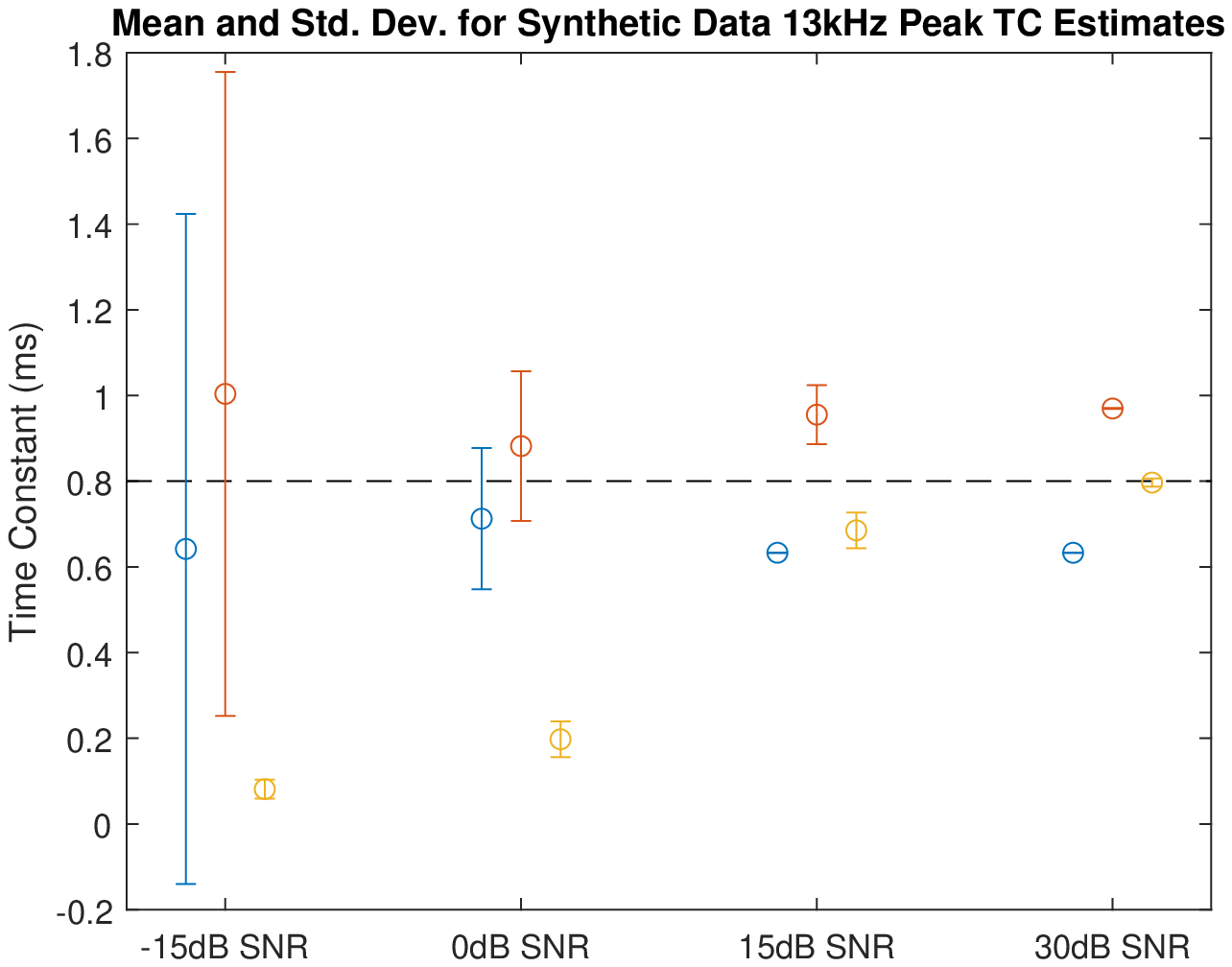}

\caption{Means and standard deviations for ESP frame based and Prony's Method based estimates of resonance peak frequency (top) and time constant (bottom) for the synthetic time series 5kHz peak (left) and 13kHz peak (right).  Mean is indicated by plotted point and standard deviation by the length of the whiskers.  True parameter values are indicated by the dashed line.  The mean and standard deviation were computed using \(100\) estimates from signals with added noise at the indicated SNR.}
\label{fig:ode-param}
\end{figure}

The frequency estimates have similar performance, for both bias and standard deviation, across all three estimation techniques and produce quality estimates at or above 0dB SNR.  The one exception is the Prony's Method 13kHz peak frequency estimate at 0dB SNR, which has a notably larger bias and standard deviation. There is comparatively more variability in the time constant estimates. None of the methods produces viable estimates at the -15dB SNR level.  At 0dB SNR the ESP frame based estimates are significantly better than the Prony's Method based estimate.  At 15dB SNR all three methods have similar performance while the Prony's Method estimates are significantly better at the 30dB SNR level, as can be seen in Table \ref{tab:ode-param}.

\begin{table}
\tiny
\[
\begin{array}{c | c c | c c}
~ & \text{5kHz Peak Freq.} & \text{13kHz Peak Freq.} & \text{5kHz Peak TC} & \text{13kHz Peak TC} \\ \hline
\text{ESP Bias} &  {\bf  0.0 }&            {\bf 0.0 } &    -0.47385     & -0.16711  \\
\text{ESP Std. Dev.} &  {\bf 0.0} &            {\bf 0.0 } &  0.00018629  &  9.1013\e{-5}  \\ \hline
\text{BPD ESP Bias} & {\bf 0.0 } &           {\bf 0.0} &     -0.49882 &      0.16968 \\
\text{BPD ESP Std. Dev.} &  {\bf  0.0 }  &         {\bf  0.0 } &   0.00070589 &    0.0010405\\ \hline
\text{Prony Bias} &     -1.9896\e{-5} &   -0.00046569  &   {\bf -0.006141 }  &    {\bf-0.003243} \\
\text{Prony Std. Dev.} &      0.00034626 &     0.0022198  &    {\bf  0.02189 }&  {\bf  0.0096209}
\end{array}
\]
\caption{Bias and standard deviations for ESP frame and Prony's Method parameter estimates at 30dB SNR.  Frequency values are in kHz and time constant (TC) values are in ms. }
\label{tab:ode-param}
\end{table}

\subsection{Experimental Time Series}
\label{sec:param-exp}

While we know the true values of the resonant frequencies and time constants for the synthetic data, and were able to precisely control the amount of added noise, we do not have this luxury for the experimental time series.  Instead we will use a Prony's Method based estimate created using the original experimental time series as the ``true'' values for the primary resonance peak time constant for the steel and wood cylinder data.  Using Prony's Method with 16 poles filtered to 8 using SVD we estimate the steel cylinder resonance peak at 4.80kHz has a time constant of 166.8ms and the wood cylinder peak at 4.32kHz has a time constant of 5.22ms.  As a point of comparison, if we use the unregularized ESP frame coefficients, Figure \ref{fig:tap}, to estimate the frequency and time constant we get 4.80kHz and 177.87ms for the steel cylinder data and 4.32kHz and 5.66ms for the wood cylinder data.  There is good agreement between the frequency estimates and the time constant estimates are reasonably close. This is a positive indication that we will be able to utilize the unregularized coefficients for parameter estimation.  On the other hand, if we apply the same estimation process to the BPD regularized coefficients, Figure \ref{fig:rtap}, we get an estimate of 4.80kHz and 54.95ms for the steel cylinder and 4.33kHz and 5.72ms for the wood cylinder.  While the wood cylinder estimate is similar to the unregularized estimate, the time constant for the steel cylinder estimate is significantly different.

\subsubsection*{Weighted Basis Pursuit Denoising}
\label{sec:params-exp-wbpd}

One potential method for dealing with the bias introduced by regularization, and to account for the fact that the regularization does not produce coefficients consistent with our understanding of the physics of the experiment, is to utilize a vectorized \(\blambda\) in Algorithm \ref{alg:bpd}.  We know from the physical setup of the experiment that, outside of noise components, the signal should have a zero time shift.  We can encode this prior knowledge into the BPD process by nonconstant \(\blambda\).  Specifically define \(\bw\) such that
\[
w_{l,k,m} = \begin{cases} 0.1 & m < 10 \\ 0.2 & m \geq 10. \end{cases}
\]
This vector places extra weight on coefficients with a time shift of greater than 0.156ms.  We allow time shifts of up to \(m=10\) because we would like to account for the signal to start at different phases.  We then define \(\blambda = \lambda_{\max} \bw\) and apply Algorithm \ref{alg:bpd}.  Using \(1000\) iterations produces the set of coefficients shown in Figure \ref{fig:wrtap}.

\begin{figure}
\centering
\includegraphics[width=1.58in]{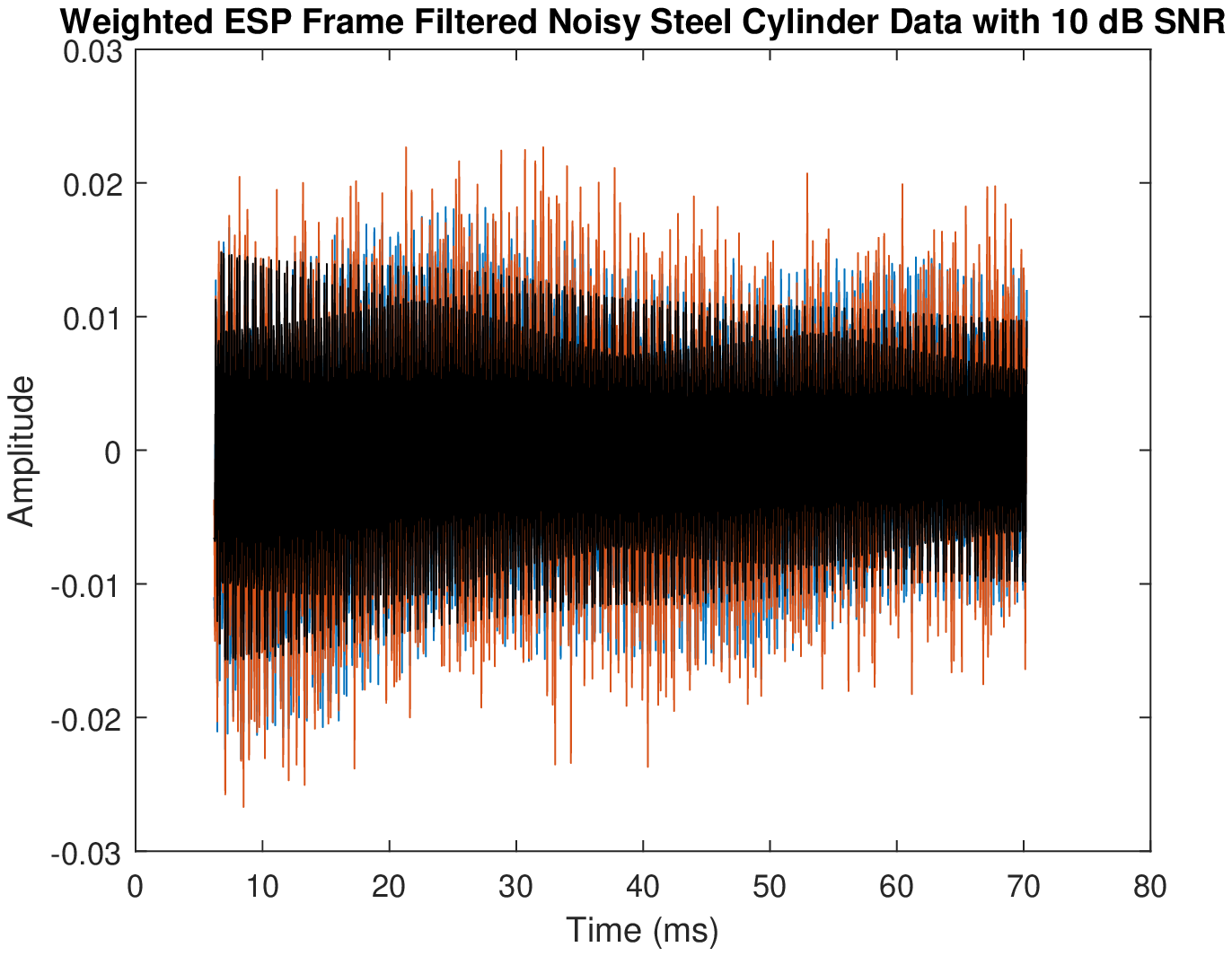}
\includegraphics[width=1.58in]{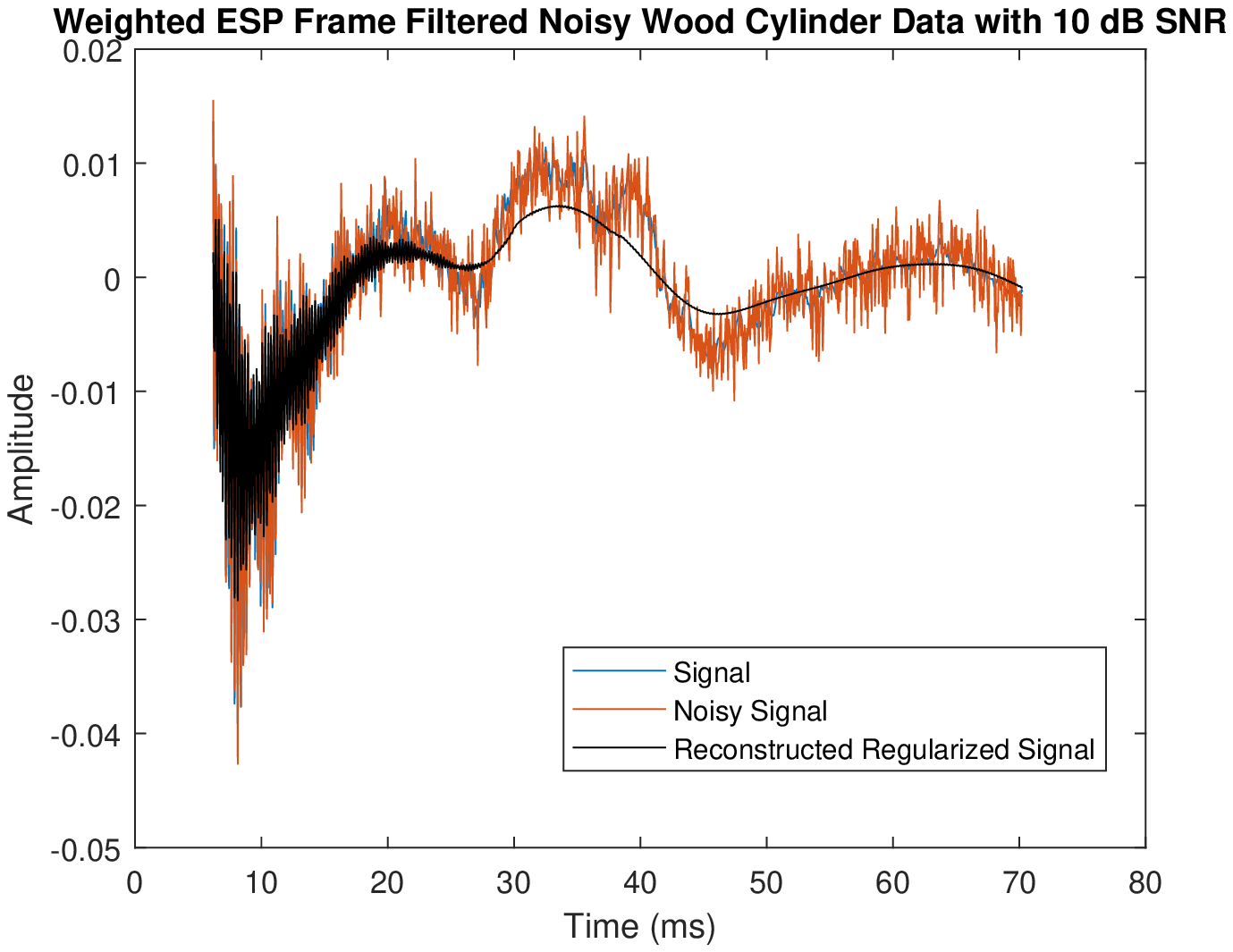}

\vspace{.1in}
\includegraphics[width=3.2in]{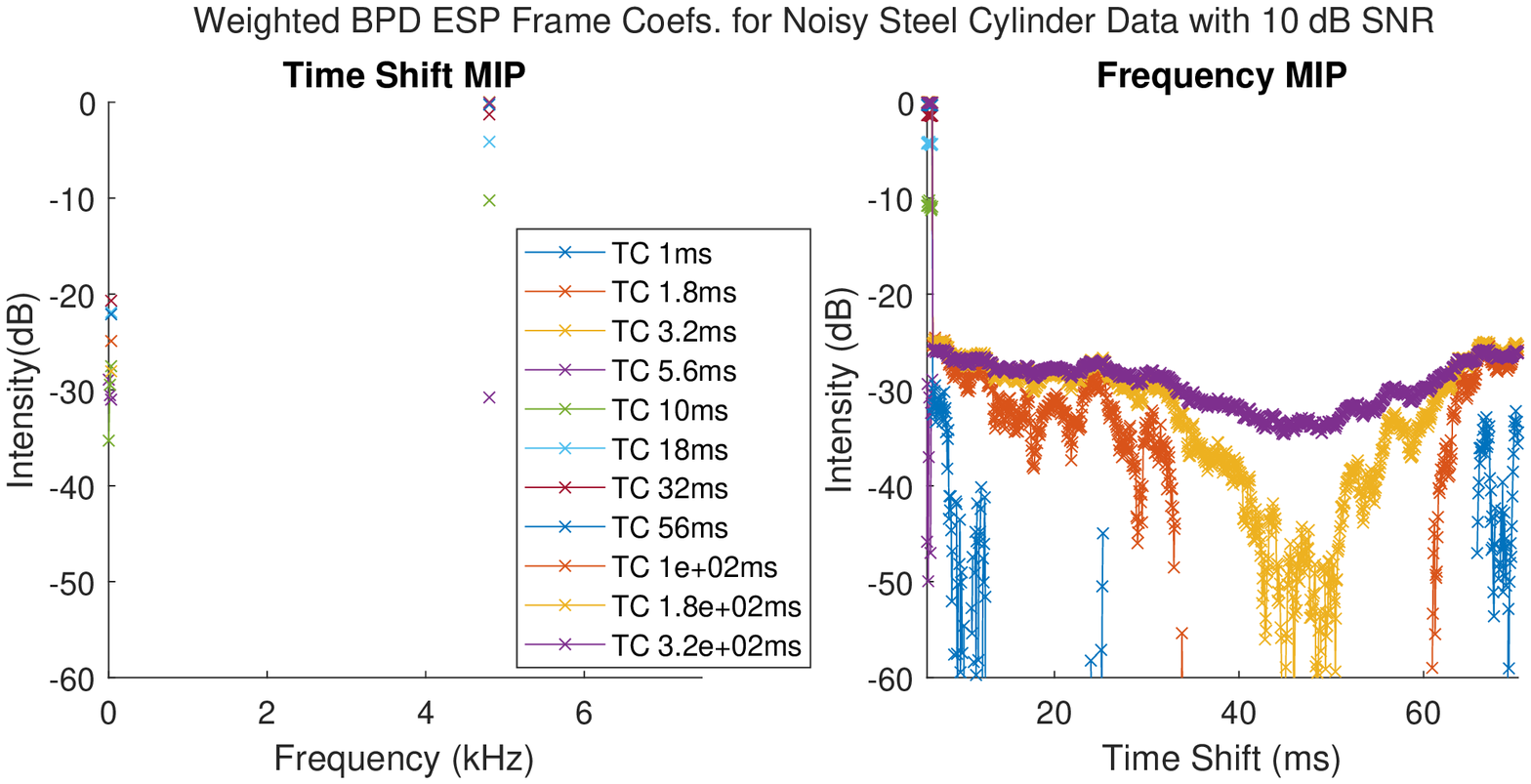}

\vspace{.1in}
\includegraphics[width=3.2in]{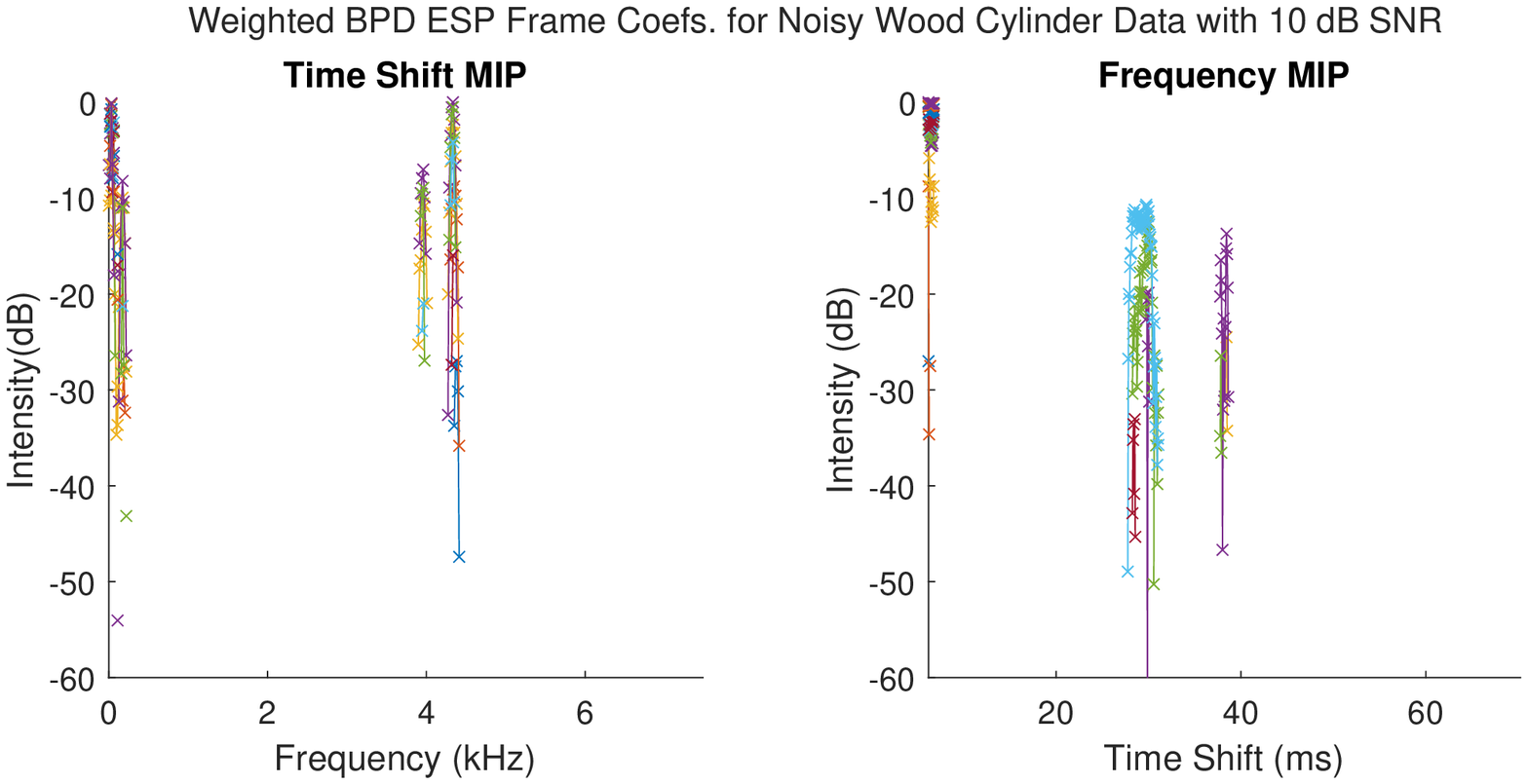}

\caption{Plots of \(\bh\), \(\bh_N\) with SNR 10dB and the reconstructed regularized signal (top left and top right), frame coefficient time shift MIP separated by time constant (middle left and bottom left) and frequency MIP separated by time constant (middle right and bottom right) for steel cylinder (top left and middle) and wood cylinder (top left and bottom) tap data after \(1000\) iterations with \(\blambda = \lambda_{\max}\bw\).  Intensities are shown on a dB scale relative to the maximum frame coefficient amplitude. The coefficients for the steel cylinder data have a relative reconstruction error of 24.0\% (12.4dB reconstructed SNR).  The coefficients for the wood cylinder data have a relative reconstruction error of 34.7\% (9.2dB reconstructed SNR). The steel cylinder coefficients have a resonance peak at 4.80kHz with an estimated time constant of 177.63ms, and the wood cylinder coefficients have a resonance peak at 4.33kHz with an estimated time constant of 5.84ms.}
\label{fig:wrtap}
\end{figure}

It is clear that the weighting is successful at concentrating the frame coefficients at lower time shifts.  For the wood cylinder most of the power is in the first millisecond, however there are some spikes at 30ms and 40ms which correspond to similar late time energy in Figure \ref{fig:tapstft}.  The reconstruction accuracy of the weighted basis pursuit coefficients is worse than for the unweighted basis pursuit coefficients, see Table \ref{tab:wbp-comp}, due to the increase in the average value of \(\blambda\).   More importantly the weighting improves the time constant parameter estimation for the steel cylinder to 177.63ms, while the estimate for the wood cylinder resonance peak time constant is consistent across all estimation techniques.

\begin{table}
\centering
\begin{tabular}{c | c c c c}
Dataset & Recon. Error (\%) & Peak Freq. (kHz) & Peak TC (ms) \\ \hline
Steel & 0.0  & 4.80 & 177.9 \\
BP Steel & 17.6 & 4.80 & {\bf 54.95} \\
WBP Steel & 24.0 & 4.80 & 177.63 \\ \hline
Wood & 0.0 & 4.32 & 5.66 \\
BP Wood & 27.6 & 4.33 & 5.72 \\
WBP Wood & 34.7 & 4.33 & 5.84
\end{tabular}
\caption{Reconstruction error and estimated resonance peak frequency and time constant for unregularized, BP regularized and weighted BP regularized ESP frame coefficients for steel and wood cylinder time series data.  The Prony's Method based estimates for the resonance peaks are 4.80kHz and 166.8ms for the steel cylinder and 4.32kHz and 5.22ms for the wood cylinder.}
\label{tab:wbp-comp}
\end{table}

\subsubsection*{Noise Analysis}
\label{sec:param-noise-exp}

As part of the broader investigation of ESP frames as a parameter estimation tool we performed the same analysis presented in Section \ref{sec:param-noise-synth} on the experimental time series.  Specifically we added noise to the experimental time series at levels ranging from \(-15\)dB SNR to 30dB SNR.  At each noise level the frequency and time constant of each resonance peak was estimated using unregularized ESP frame coefficients, {\em weighted} BPD regularized ESP frame coefficients, and Prony's Method.  For the unregularized ESP frame coefficients the parameters were estimated using \eqref{eq:tau}.  The regularized coefficients were generated using \(\blambda = \lambda_{\max}\bw\) and \(1000\) iterations as in the previous section.  Finally for Prony's Method we use 16 poles filtered to 8 via SVD.  Each of these estimates was generated for \(100\) noise realizations and the resulting estimate means and standard deviations were calculated and plotted in Figure \ref{fig:exp-param} with the bias and standard deviations at 30dB SNR shown in Table \ref{tab:exp-param}.  Recall the ``true'' parameter values were generated using Prony's Method as described at the start of Section \ref{sec:param-exp}.

\begin{figure}
\centering
\includegraphics[width=1.58in]{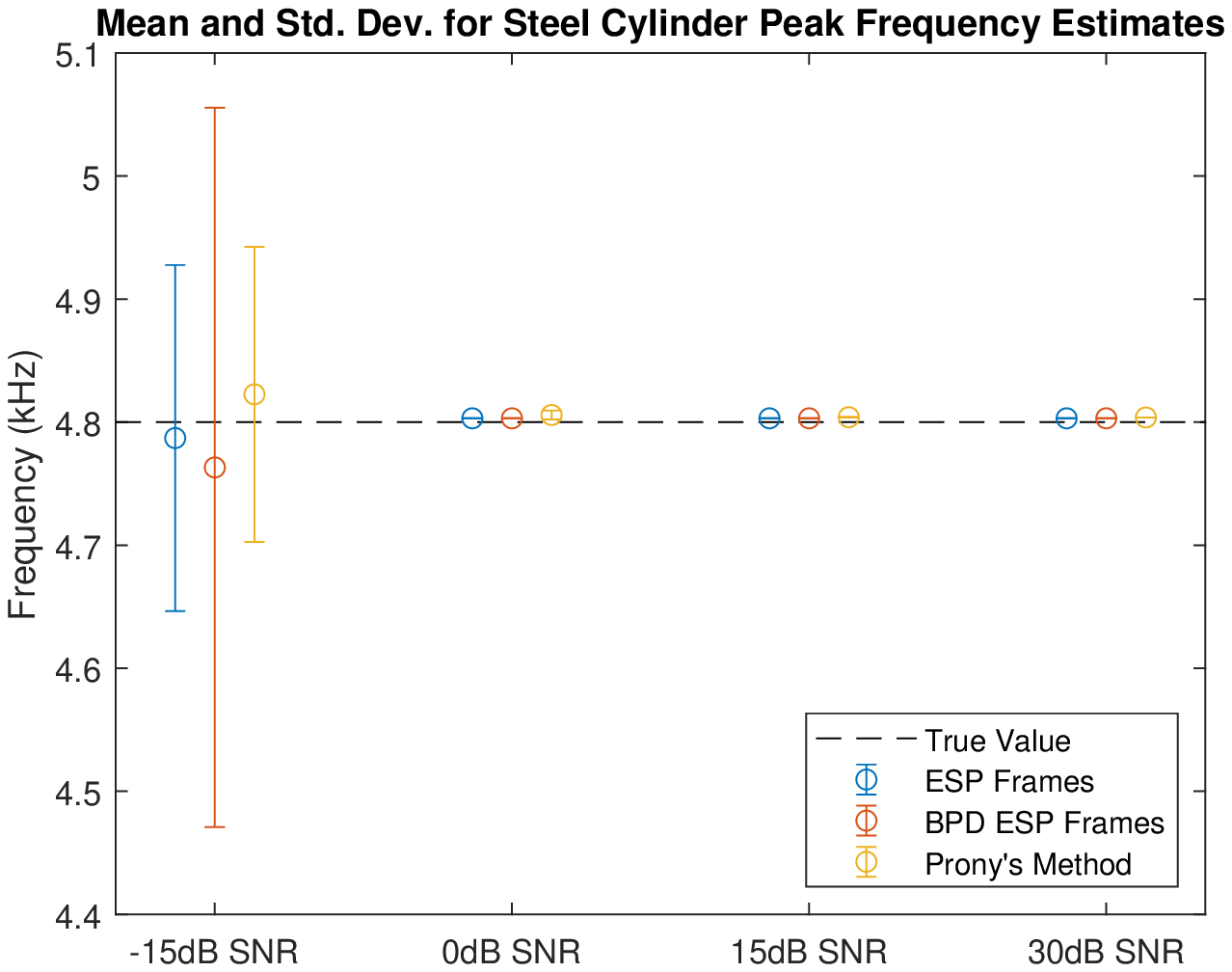}
\includegraphics[width=1.58in]{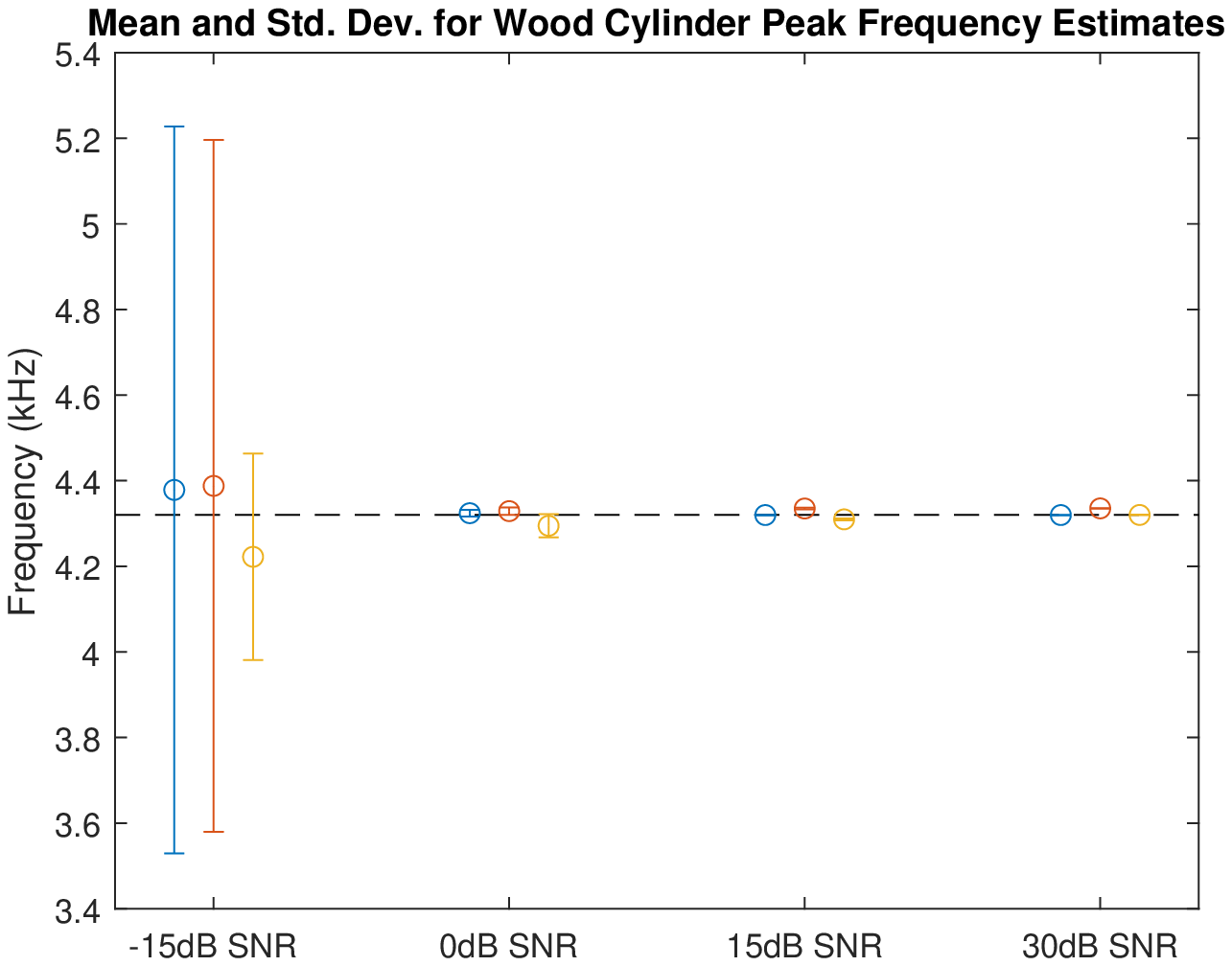}

\vspace{.1in}
\includegraphics[width=1.58in]{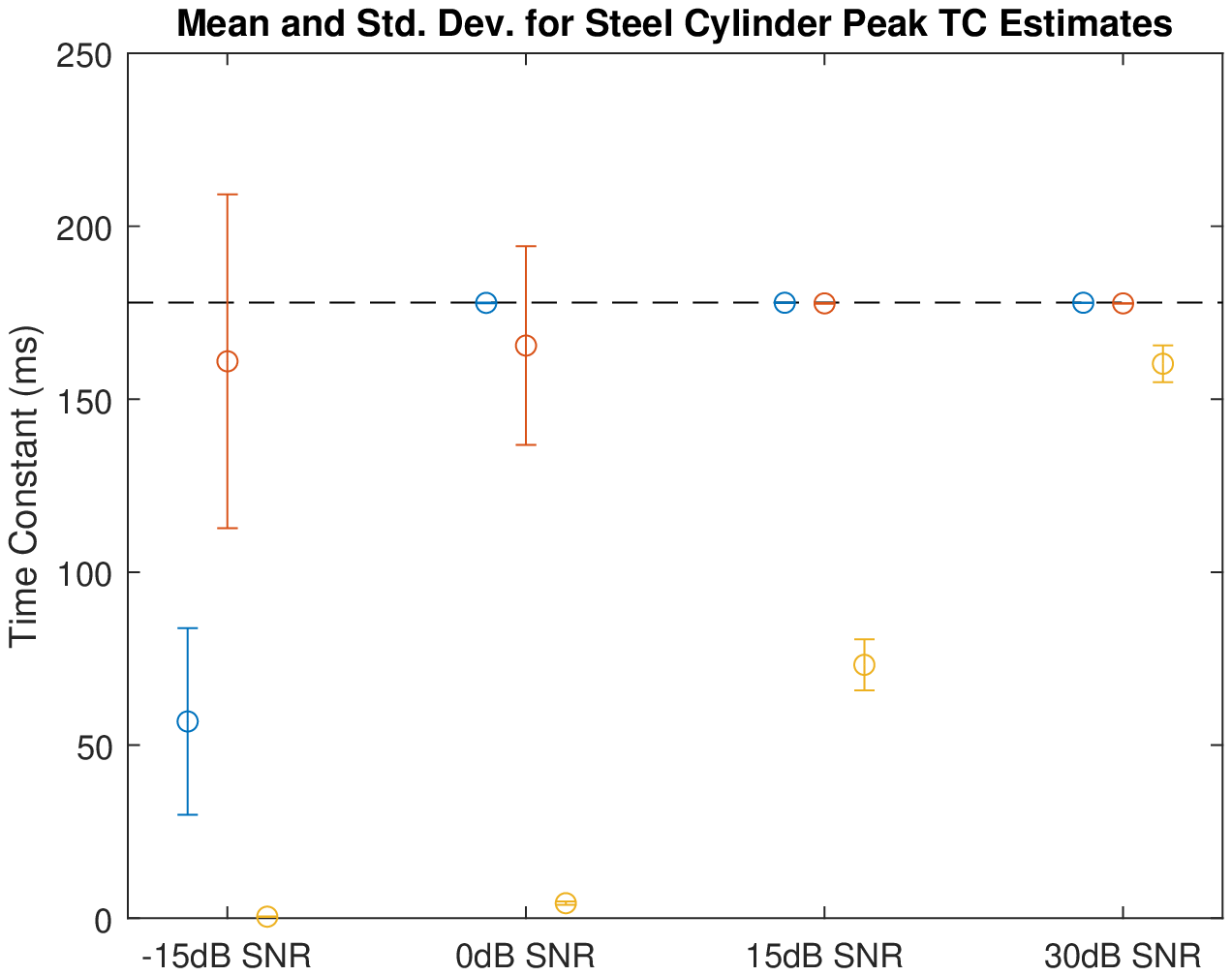}
\includegraphics[width=1.58in]{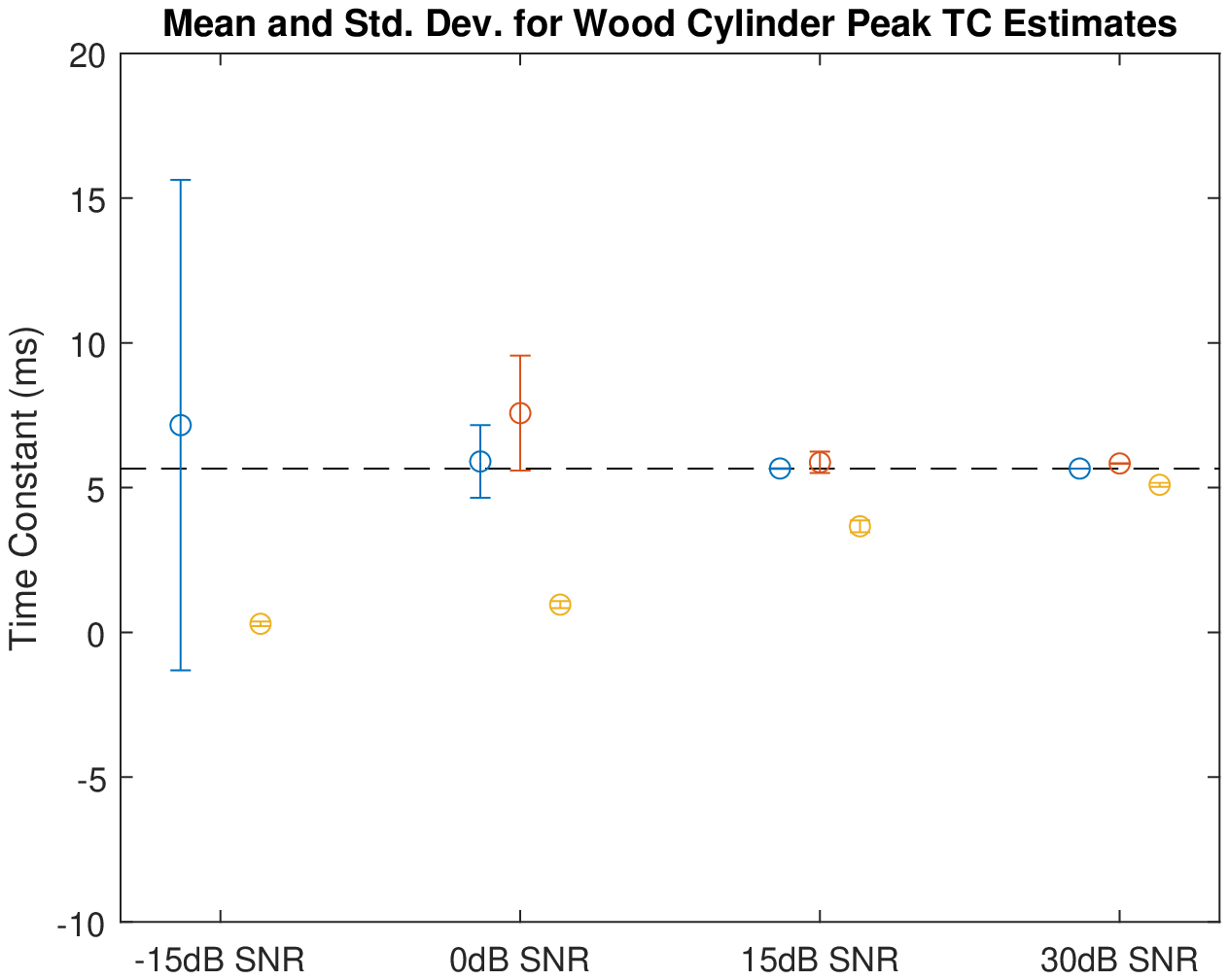}

\caption{Means and standard deviations for ESP frame based and Prony's Method based estimates of resonance peak frequency (top) and time constant (bottom) for the steel cylinder (left) and wood cylinder (right) experimental time series.  Mean is indicated by plotted point and standard deviation by the length of the whiskers.  The best measured value is indicated by the dashed line.  The \(-15\)dB SNR BPD ESP estimate mean for the wood cylinder data is outside the scale of the plot. Mean and standard deviation were computed using \(100\) estimates from signals with added noise at the indicated SNR.}
\label{fig:exp-param}
\end{figure}

As in Section \ref{sec:param-synth} all three methods have similar behavior with regards to estimation of the frequency parameter, producing quality estimates at or above the 0dB SNR level.
None of the methods produces a viable time constant estimate at -15dB SNR.  At 0dB SNR the ESP frame based estimates are reasonably close to the true value while the Prony's Method estimates are near zero.  At the 15dB and 30dB SNR levels the ESP frame based estimates have a significantly smaller bias than Prony's Method.  Overall we find that ESP frames can be used to estimate resonance peak parameters and are competitive with Prony's Method, particularly in the presence of noise.

\begin{table}
\tiny
\[
\begin{array}{c | c c | c c}
~ & \text{Wood Peak Freq.} & \text{Steel Peak Freq.} & \text{Wood Peak TC} & \text{Steel Peak TC} \\ \hline
\text{ESP Bias} &     -0.00031189  &   {\bf  0.0031189}  & {\bf  -0.00039911 }&  {\bf  -0.026285 }\\
\text{ESP Std. Dev.} &   0.0 &           {\bf  0.0 }  &   {\bf 0.0016578 }  & {\bf  0.004194 } \\ \hline
\text{BPD ESP Bias} &    0.015283  &    0.0031189  &       0.1756 &    -0.21042 \\
\text{BPD ESP Std. Dev.} &   0.0   &          0.0 &     0.0068052 &    0.020472 \\ \hline
\text{Prony Bias} &   {\bf -6.3537\e{-5} }  &  0.0038046  &     -0.56135  &    -17.662 \\
\text{Prony Std. Dev.} &  {\bf  0.00052306 } &   3.9251\e{-5} &      0.068351  &     5.3159
\end{array}
\]
\caption{Bias and standard deviations for ESP frame and Prony's Method parameter estimates at 30dB SNR.  Frequency values have are in kHz and time constant (TC) values are in ms.}
\label{tab:exp-param}
\end{table}

\section{Discussion}
\label{sec:conclusions}

This paper presented a method of constructing Parseval frames from any collection of complex envelopes. The resulting ESP frames can represent a wide variety of signal types as specified by their physical morphology. Since the ESP frame retains its Parseval property it is compatible with large scale and iterative optimization algorithms such as SALSA and ADMM and sparse sets of ESP frame coefficients can be generated using traditional convex optimization. This work presented examples of ESP frame generation, as well as \(L_1\)-regularized coefficient generation, for both synthetic and experimentally measured signals. The use of sparse coefficients for both denoising and parameter estimation was also demonstrated.

When seeking sparse sets of ESP frame coefficients we generally expect the signal will not exactly equal a small linear combination of frame vectors.  This can be due to the presence of noise or because of poor resolution in the envelope parameter.  While noise can be mitigated using any number of techniques, including BPD, the fact that the dimension of the ESP frame is given by \(N^2 L\), where \(L\) is the number of envelopes, means that achieving a very fine resolution along the envelope axis can be computationally infeasible.  In either case, when the signal does not exactly equal a small linear combination of frame vectors we expect that the optimal set of BP coefficients may not be particularly sparse and that high levels of reconstruction error may be needed to produce sparse coefficients using BPD.  For many applications, though, it is not necessary to achieve true sparsity and instead we simply desire the ESP frame coefficients to produce discrete peaks.

With regards to denoising, we found that ESP frames are competitive with the STFT as a noise reduction tool, producing larger SNR gains over a range of noise levels.  While in terms of percentage of zero coefficients the ESP Frame representations were sparser than the STFT representations, because the ESP frame is so large the STFT frame regularization ends up producing significantly fewer nonzero coefficients.  Additionally the ESP frame approach also takes longer to converge and is more computationally intensive.  In the ideal scenario we expect that ESP frame based denoising will outperform STFT based denoising since the ESP frame can be used to encode a desired signal model while the STFT is signal agnostic.

We also found that ESP frames are competitive with Prony's Method when applied to resonance parameter estimation.  The unregularized ESP frame coefficients perform about as well as the regularized ESP frame based estimates across all test cases while being less computationally intensive.   The ESP frame approaches produced viable estimates of the time constant parameter at a lower SNR than Prony's Method.  At very high SNR Prony's Method produced better estimates in the synthetic time series case, while the ESP frame estimates were better for the experimental time series.  This is consistent with the fact that Prony's Method is very accurate when its underlying signal model is a good match for the time series.   It is not thought that ESP frames will outperform Prony's Method in terms of accuracy in optimal conditions.  Instead however, the intention is to utilize ESP frames on signals where the number of poles or the start of the resonance component is not known {\em a priori}.

 There are a number of possible future applications for ESP frames, ranging from Multi-Component Analysis and filtering to generating feature sets for use in signal classification.  Another potential avenue of investigation is to try and allow the envelope parameter to vary as part of the \(L_1\)-regularization procedure.  This could enable the ESP frame envelopes to be more data informed and may further enhance sparsity.  Overall Enveloped Sinusoid Parseval frames are a flexible signal analysis tool, particularly when combined with convex optimization, and offer a wide range of applications. The ESP frame can be easily tuned to represent a wide variety of signals simply by providing their relevant envelopes. It is compatible with modern convex analysis techniques, and is efficient and practical to deploy.

\bibliographystyle{IEEEtran}
\bibliography{references}

\end{document}